%% file: main.tex
\documentclass[12pt]{article}
\usepackage[top=1in, left=1in, right=1in, bottom=1in]{geometry}	
\usepackage[parfill]{parskip}	
\usepackage[export]{adjustbox}
\usepackage{graphicx, xcolor}
\usepackage{amssymb}
\usepackage{mathrsfs}
\usepackage{enumitem}
\usepackage{epstopdf}
\usepackage{bbm}
\usepackage{amssymb}
\usepackage{amsmath}
\usepackage{amsfonts}
\usepackage{bm, bbm}

\usepackage{lscape}
\usepackage{rotating}
\usepackage{setspace}
\usepackage{threeparttable}
\usepackage{booktabs}
\usepackage{floatrow}
\usepackage{multirow}
\usepackage{caption}
\usepackage{subcaption}
\usepackage[compact]{titlesec}
\usepackage{lmodern}
\usepackage{comment}
\usepackage{floatrow}
\allowdisplaybreaks

\fontfamily{lmtt}\selectfont
\usepackage[T1]{fontenc}
\usepackage{natbib}
\bibpunct{(}{)}{;}{a}{}{,}
\usepackage{titling}
\newcommand{\subtitle}[1]{%
  \posttitle{%
    \par\end{center}
    \begin{center}\large#1\end{center}
    \vskip0.5em}%
}

\usepackage[colorlinks,
            linkcolor=red,
            anchorcolor=blue,
            citecolor=blue
            ]{hyperref}

\usepackage{amsthm}

\usepackage{ying}

\newtheorem{theorem}{Theorem}[section]
\newtheorem{assumption}{Assumption}[section]
\newtheorem{lemma}[theorem]{Lemma}
\newtheorem{prop}[theorem]{Proposition}
\newtheorem{remark}[theorem]{Remark}

\newtheorem{example}[theorem]{Example}

\DeclareMathOperator*{\Var}{{\rm Var}}

\DeclareMathOperator*{\Cov}{Cov}

\renewcommand{\hat}{\widehat}
\renewcommand{\bar}{\overline}

\def \aipw {{\textrm{AIPW}}}

\def\@#1\@{\begin{align}#1\end{align}}
\def\$#1\${\begin{align*}#1\end{align*}}
\usepackage{color-edits}
\addauthor[Ying]{ying}{magenta}
\addauthor[Jose]{jose}{blue}

\usepackage{color}
\begin{document}
\pagestyle{plain}

\newcommand{\blind}{0}

\newcommand{\tit}{\Large Cross-Balancing for Data-Informed Design and \\ Efficient Analysis of Observational Studies}

\if0\blind

{\title{\Large \tit
\thanks{We thank Ambarish Chattopadhyay and Yige Li for helpful comments and conversations. This work was partly supported by two awards from the Patient Centered Outcomes Research Initiative
(PCORI; ME-2022C1-25648, ME-2024C2-40180).}\vspace*{.3in}}
\author{\normalsize Ying Jin\thanks{Department of Statistics and Data Science, University of Pennsylvania, Academic Research Building, Room 441, Philadelphia, PA 19104; email: \url{yjinstat@wharton.upenn.edu}.} \and \normalsize Jos\'{e} R. Zubizarreta\thanks{Departments of Health Care Policy, Biostatistics, and Statistics, Harvard University, 180 Longwood Avenue, Office 215-A, Boston, MA 02115; email: \url{zubizarreta@hcp.med.harvard.edu}.}}
\date{}

\maketitle
\date{}
}\fi

\if1\blind
\title{ \tit}
\date{}
\maketitle
\fi

\vspace{-.5cm}
\begin{abstract}
Causal inference starts with a simple idea: compare groups that differ by treatment, not
much else. Traditionally, similar groups are constructed using only observed covariates; however, it remains a long-standing challenge to incorporate available outcome data into the study design while preserving valid inference. In this paper, we study the general problem of covariate adjustment, effect estimation, and statistical inference when balancing features are constructed or selected with the aid of outcome information from the data. We propose cross-balancing, a method that uses sample splitting to separate the error in feature construction from the error in weight estimation. Our framework addresses two cases: one where the features are learned functions and one where they are selected from a potentially high-dimensional dictionary. 
In both cases, we establish mild and general conditions under which cross-balancing produces consistent, asymptotically normal, and efficient estimators. In the learned-function case, cross-balancing achieves finite-sample bias reduction relative to plug-in-type estimators, and is multiply robust when the learned features converge at slow rates. 
In the variable-selection case, cross-balancing only requires a product condition on how well the selected variables approximate true functions. 
We illustrate cross-balancing in extensive simulations and an observational study, showing that careful use of outcome information can substantially improve both estimation and inference while maintaining interpretability.

\end{abstract}

\begin{center}
\end{center}
\clearpage
\doublespacing


\section{Introduction}
\label{section_intro}
\vspace{-.25cm}

\subsection{Using outcome information for weighting}
\label{section_intro1}
\vspace{-.25cm}

Causal inference begins with a simple idea: compare groups that differ by treatment, but little else. 
In observational studies, this requires adjusting for differences between treatment groups, a process known as balancing covariates. 
A classic approach uses the propensity score---the probability of receiving treatment given observed covariates \citep{rosenbaum1983central}---to achieve balance. 
However, because this approach ensures balance only asymptotically, more recent methods also focus on balancing the covariates directly.
Examples include the approaches by \citet{hainmueller2012entropy}, \citet{zubizarreta2015stable}, and \citet{chan2016globally}, among others; see \citet{ben2021balancing} for a review.

These balancing weighting methods have several appealing features. 
First, they are transparent and easy to interpret, as they aim to compare ``apples to apples.'' 
Second, by mimicking certain aspects of randomized experiments, their construction can be seen as part of the design stage of an observational study, as they use only pre-treatment information, preserving the objectivity of the analysis and permitting valid inference \citep{rubin2007design}.
Furthermore, the resulting estimators often admit a linear representation, enabling straightforward diagnostics where each unit receives a distinct weight \citep{chattopadhyay2023implied}.

Despite its appeal, a fundamental question concerns which covariates one should select for balance.  
This choice is central to the estimator's performance and involves a careful trade-off between model expressiveness and interpretability.  
This choice is governed by the dual role of the balancing weights, as they eliminate finite-sample bias in the part of the outcome model attributable to the selected features, and implicitly define a propensity score model based on them \citep{imai2014covariate, zhao2017entropy, wang2020minimal}.  
Although this decision process often relies on domain expertise, the complexity of modern datasets makes it  desirable to complement subjective knowledge with data-driven approaches to determine which covariate adjustments to make for accurate causal estimates.

This raises, in turn, a broader challenge: how to leverage not only observed covariates but also outcome information in a manner that can be considered part of the study design and yields valid inferences.  
Leveraging outcome information holds much unrealized potential, as imbalances in the prognostic score---the expected outcome given the covariates \citep{hansen2008prognostic}---can translate directly into treatment effect estimation bias \citep{stuart2013prognostic}.  
Indeed, methods that incorporate outcome information, such as outcome kernels \citep{zhao2019covariate}, the outcome-adaptive lasso \citep{shortreed2017outcome}, and small planning samples for the design of observational studies \citep{heller2009split}, have proven effective. 

Two natural ideas emerge for constructing balancing features using outcome information:\vspace{-0.5em}
\begin{itemize}\setlength\itemsep{-0.5em}
\item[(i)] \emph{Learn features as functions of raw covariates}: following the outcome modeling principle, learning the true regression function (the prognostic score) using flexible prediction methods can help address finite-sample bias in the outcomes due to imbalance. The resulting learned function can then be used as a balancing feature.
\item[(ii)] \emph{Select features from a candidate set:} starting with raw covariates and their transformations (e.g., quadratic or kernel functions), one can use variable selection methods to find features most relevant to the outcome. Using the selected variables as the balancing features can help reduce complexity and improve interpretability.
\end{itemize}
\vspace{-0.5em}
However, these approaches pose several design and analysis challenges. It is often unclear which properties of the constructed features ensure accurate estimation and inference. 
Using the same data for both feature construction and weight estimation also introduces ``double-dipping'' bias,  compromising the validity of the inferences. 
Moreover, because these methods use outcome information at the design stage, they may further raise concerns about the objectivity of the study.
The core challenge is to account for learning and selection errors to ensure accurate estimation and valid inference, while maintaining the interpretability and transparency of the balancing approach.
Technically, the nature of statistical errors differs: in (i), the uncertainty comes from the randomness in learned functions, while in (ii), it arises from selecting a random subset of features, which is inherently discrete.

\subsection{Overview of cross-balancing}
\label{section_intro2}
\vspace{-.2cm}

In this paper, we study the general problem of covariate adjustment, effect estimation, and statistical inference when balancing features are constructed in a data-driven manner. 
To address the aforementioned challenges, we propose cross-balancing, a method based on sample splitting that disentangles the estimation error in feature construction and selection from that in weight estimation, and achieves asymptotic optimality under relatively weak conditions.
This method builds on the intuitive idea of randomly splitting a dataset into two parts~\citep{chernozhukov2018double}, so that knowledge gained from one part can be applied to the other, and their results are subsequently integrated.
In this way, we leverage the outcomes as an additional source of information that is typically not utilized in balancing.

\begin{figure}[h!]
    \centering
    \includegraphics[width=0.8\linewidth]{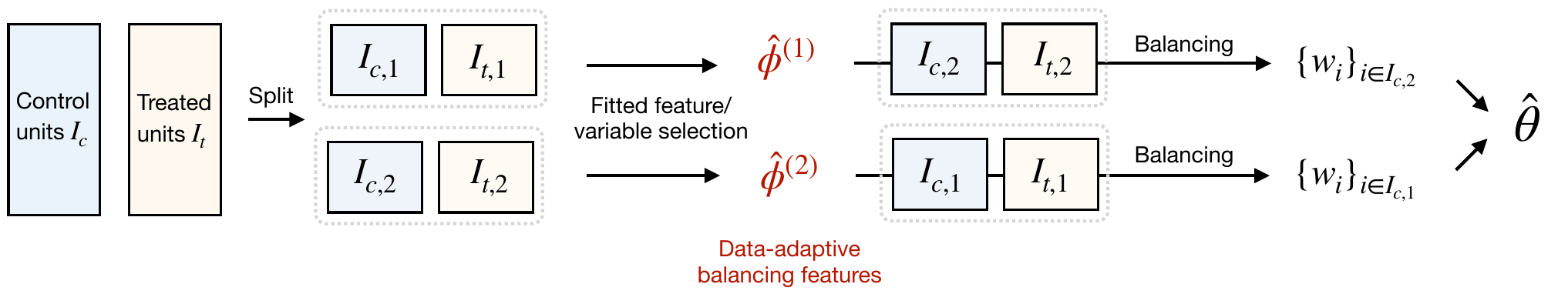}
    \caption{Diagram of the cross-balancing process.\vspace{-1em}}
    \label{fig:overview}
\end{figure}

As depicted in Figure \ref{fig:overview}, cross-balancing proceeds in two stages: in the first stage, features are built using one fold of the data; and in the second stage, a balancing algorithm is applied to the other fold to obtain the weights. 
The roles of the two folds are then interchanged to obtain the weights of the first fold. 
This method is flexible, as the feature learning/selecting stage can leverage any black-box prediction model to learn the prognostic score, or can employ any variable selection method to find a subset of relevant features to balance. 
By not restricting to any specific algorithm in the learning/selecting stage, we characterize the general conditions on the quality of data-driven features, allowing slow convergence rates, under which cross-balancing yields consistent estimation, valid inference, and asymptotic efficiency. 
Along the way, we also show the finite-sample bias reduction and multiple robustness due to the unique calibration property of balancing weights. 
Our analysis provides practical guidance in building data-driven balancing features in both cases: while a consistent prognostic score suffices for a consistent estimator, including treatment-relevant features protects against misspecification and yields valid inference under mild conditions. 

As discussed, a useful principle in observational studies is to approximate the randomized experiment that would ideally have been conducted under similar conditions. In cross-balancing, this principle is maintained, as in concept there are two hypothetical experiments that inform each other by using the outcome information from the other experiment. In this sense, one fold serves as the ``pilot study'' for the other~\citep{heller2009split}.
Importantly, sample splitting directly addresses the ``double-dipping'' bias that arises when the same data is used both for feature construction  and weight estimation, thereby ensuring the validity of inferences. 
In doing so, cross-balancing preserves the interpretability and transparency of the balancing approach, while leveraging the power of modern machine learning tools.

\subsection{Related literature and outline}
\label{section_intro3}
\vspace{-.2cm}

There has been much work on what variables should be adjusted when estimating treatment effects. At the center of these proposals, either implicitly or explicitly, is the propensity score. This is because the propensity score provides the most concise summary of the data that balances all covariates; and conversely, balancing all covariates requires estimating this single propensity score. In practice, researchers typically aim to balance as many variables as possible~\citep{rosenbaum1995overt}. However, when only a subset of variables is relevant to the treatment, including more than the necessary variables leads to ``overadjustment'' \citep{rotnitzky2010note}. The underlying principle is that when the true model is smaller (involving a smaller set of variables), it is more efficient to only adjust for the relevant subset.

This reasoning leads us to focus on variables that play a key role in the outcome model. A covariate is worth adjusting for only if it is related to the outcome, as otherwise the imbalance in this covariate would not incur confounding bias in the outcomes. Following this logic, \cite{schneeweiss2009high} use the covariate-outcome relationship to inform variable selection in propensity score methods. Through simulations, \cite{brookhart2006variable} suggest that including variables related only to the outcome may improve efficiency. \cite{shortreed2017outcome} propose outcome adaptive lasso for selecting covariates for inclusion in propensity score models. Similarly, balance in the prognostic score can be used as an effective metric to assess the quality of propensity score models \citep{stuart2013prognostic}.

Moving beyond propensity score only methods, other approaches employ outcome regression to augment the propensity score model, resulting in the semiparametrically efficient AIPW estimator \citep{robins1994estimation}. Its double robustness property makes it well suited to use machine learning for both propensity and outcome models, each of which individually can converge at slow rates, yet their combination enables fast rates (e.g., \citealt{chernozhukov2018double}, \citealt{bruns2025augmented}, \citealt{yang2025outcome}). Given this general property, it is natural to include all variables to fit these models, but there is room for improvement by focusing on relevant variables. Along this line, \cite{belloni2014inference} and \cite{moosavi2023costs} show that the ``double selection'' strategy, which selects variables related to either outcome or treatment, yields favorable performance in high-dimensional settings. 
Cross-balancing echoes this line of research, and it turns out that selecting features relevant to treatment and outcome is important for asymptotic optimality under relatively weak conditions.

The remainder of the paper is organized as follows. 
Section~\ref{section_setup} presents the setup and assumptions, as well as an overview of balancing weights. 
Section~\ref{sec:fitted} introduces cross-balancing with learned features, describing the procedure, a convergence analysis, and inferential guarantees. 
Section~\ref{section_cross2} covers cross-balancing with selected features following a similar structure. 
Section~\ref{sec:simu} presents an extensive simulation study. 
Section~\ref{section_case} concludes with a case study.
\vspace{-2.5em}


\section{Setup}
\label{section_setup}
\vspace{-.4cm}

\input{section_setup}


\section{Cross-balancing with learned features}
\label{sec:fitted}
\vspace{-.2cm}

\input{section_cross1}
 
\section{Cross-balancing with selected variables}
\label{section_cross2}
\vspace{-.2cm}

\input{section_cross2}


\input{section_simu}

\section{Observational study}
\label{section_case}
\vspace{-.25cm}

\input{section_case}

\setlength{\bibsep}{-1pt}
{\footnotesize 
\bibliography{bibliography}
\bibliographystyle{asa}
}

\newpage 
\setcounter{page}{1}
\begin{center}
    {\Large {\bf Supplementary Materials for} ``Cross-Balancing for Data-Informed Design and Efficient Analysis of Observational Studies''}
\end{center}

\appendix
\input{appendix}

\end{document}

%% file: section_setup.tex
\subsection{Estimand and assumptions}
\label{section_setup1}
\vspace{-.2cm}

We begin with the standard observational study setting, where each subject is described by observed covariates $X_i \in \RR^p$, a treatment indicator $T_i\in \{0,1\}$, and an observed outcome $Y_i \in \RR$, such that the triplets $\{(X_i, T_i, Y_i)\}_{i=1}^n$ are independent and identically distributed.
We define causal effects in terms of potential outcomes, $Y_i(1)$ if treated and $Y_i(0)$ if control \citep{neyman1923application, rubin1974estimating}, and assume the Stable Unit Treatment Value Assumption (SUTVA), so that $Y_i = Y_i(T_i) = T_i Y_i(1) + (1-T_i) Y_i(0)$ \citep{rubin1980randomization}.
We also require the following conditions on the treatment assignment \citep{rosenbaum1983central}.
\begin{assumption}\label{assump:unconfound_overlap}
    Treatment assignment is unconfounded given observed covariates: $T \indep $ $\{ Y(1),Y(0) \} \given X$. Each individual has a positive probability of receiving either treatment: $0 < e(x) < 1$ for all $x$, where $e(x) = \mathbb{P}(T=1|X=x)$ is the propensity score.
\end{assumption}

Our main objective is to estimate the counterfactual mean $\theta_0 = \mathbb{E}[Y(0)|T=1]$. 
This is the building block for the average treatment effect on the treated (ATT), as
$\mathbb{E}[Y(1)-Y(0)|T=1] = \mathbb{E}[Y|T=1] - \theta_0$.
Because we rely on information about $Y(0)$ from control subjects to estimate the counterfactuals for the treated, the primary challenge under Assumption~\ref{assump:unconfound_overlap} is to properly adjust for imbalances in observed covariates between the two groups.

Although our focus is causal inference, the same issues arise when transferring results from one population to another. 
Suppose we observe $(X_i, Y_i)$ from a source distribution $P$, and covariates $X_j$ from a target distribution $Q$, and our objective is to estimate the mean outcome in the target population, $\mathbb{E}_Q[Y]$. 
We consider a classical transfer problem under covariate shift, assuming the distributions differ only in $X$. 
Denote by $w(x)$ the (unknown) density ratio relating the two covariate distributions, so
$\frac{dQ}{dP}(x, y) = w(x)$.
The methods we present for adjustment and inference in causal studies are equally suited to this setting.


Define the prognostic score $m_0(x) = \mathbb{E}[Y(0) \mid X = x]$, and the covariate shift between the treated and control groups through the weight function
$
w^*(x) = \frac{e(x)}{1-e(x)} \cdot \frac{\mathbb{P}(T=0)}{\mathbb{P}(T=1)}.
$
The estimand of interest can be represented in two equivalent forms
$ 
\theta_0 = \mathbb{E}[w^*(X)Y \mid T=0] = \mathbb{E}[m_0(X)\mid T=1].
$
This forms the basis of several seminal approaches to estimating $\theta_0$, including inverse propensity weighting (IPW), which estimates $e(x)$ and reweights control units, and outcome regression which models $m_0(x)$ and averages the predictions over treated units; combining them leads to the doubly robust augmented IPW (AIPW) estimator.

\subsection{Minimal weights}
\label{section_setup3}
\vspace{-.2cm}

Within the weighting framework, one approach to construct weights is to directly estimate the inverse propensity score by solving an optimization problem of the form  
\vspace{-1em}
{\setlength{\belowdisplayskip}{6pt}\@
\label{eq:minimal_weights}
\arg\min_{w} \left\{ \mathcal{D}(w) : \, w \in \mathcal{B} \cap \mathcal{A} \right\},
\@}
where $\mathcal{D}$ is a measure of dispersion of the weights, $\mathcal{B}$ are general balance conditions, and $\mathcal{A}$ are weight adjustment considerations. 
This approach allows for the construction of various types of weights, such as those in \cite{hainmueller2012entropy}, \cite{zubizarreta2015stable}, and \cite{chan2016globally}.
See \cite{wang2020minimal} for an analysis of the asymptotic properties and practical considerations of such weights (\ref{eq:minimal_weights}).
In what follows, we restrict $\mathcal{D}$ to the variance of the weights in order to control the variance of a linear estimator, define $\mathcal{B}$ via finite-sample balancing conditions on a vector of transformations of the raw features, and set $\mathcal{A}$ to ensure that the estimator is translation invariant.  
That is, we solve
\vspace{-0.5em}
{\small\setlength{\belowdisplayskip}{8pt}\@
\label{eq:minimal_weights2}
\hat{w} = \arg\min_{w} \left\{ \sum_{i \in \mathcal{I}_c} w_i^2 : \, \bigg| \frac{1}{|\mathcal{I}_c|} \sum_{i\in \mathcal{I}_c} w_i \hat\phi_k(X_i) - \frac{1}{|\mathcal{I}_t|} \sum_{i\in \mathcal{I}_t} \hat\phi_k(X_i) \bigg|\leq \delta_k, k \in \mathcal{K}; \frac{1}{|\mathcal{I}_c|}\sum_{i\in \mathcal{I}_c} w_i = 1 \right\},
\@}
where $\hat\phi_k(X_i)$ are the \emph{balance features}, which are transformations of the covariates (we emphasize that they can be data-dependent to fit in our framework), $\delta_k$ are the tolerances, and $\cI_t,\cI_c\subseteq \{1,\dots,n\}$ are the index sets of  treated and control units, respectively. 

The key challenge we address is  choosing the functions $\hat\phi_k$ in a data-dependent fashion while enabling reliable estimation and inference (though the exact form of our balancing optimization program differs from~\eqref{eq:minimal_weights2} by procedure design). 
We present methods and theory of cross-balancing when $\hat\phi_k$'s are learned transformations of the raw covariates in Section~\ref{sec:fitted} and when $\hat\phi_k$'s are a selected subset of (a dictionary of) features in Section~\ref{section_cross2}.

\noindent\textbf{Notation.}
We close the section with the notation used throughout the paper.
We use the standard symbols $o(\cdot)$, $o_P(\cdot)$, and $O_P(\cdot)$, where $X_n = o_P(1)$ means $X_n \xrightarrow{P} 0$ and $X_n = O_P(1)$ means $\{X_n\}$ is stochastically bounded.
 We also use $\Omega(\cdot)$ to denote the inverse of stochastic boundedness where $X_n = \Omega(a_n)$ if $X_n / a_n$ is bounded away from zero in probability. For a fixed function $f(\cdot)$, we denote its $L_2$-norm under a distribution $P$ by $\|f\|_{L_2(P)}:=\EE_{X\sim P}[f(X)^2]$. 
For a vector $\eta\in \RR^p$, we denote $\|\eta\|:=(\sum_{j=1}^p \eta_j^2)^{1/2}$ as its Euclidean norm, $\|\eta\|_1:=\sum_{j=1}^p |\eta_j|$ as its $L_1$-norm, and $\|\eta\|_\infty=\max_{j=1}^p |\eta_j|$ as its $L_\infty$-norm.

%% file: section_cross1.tex
  
As discussed, there has long been interest in leveraging outcome information such as the estimated prognostic score, since it provides the most succinct summary of the relationship between outcome and covariates \citep{rubin2000combining}. 
According to \cite{hansen2008prognostic}, however, ``same-sample estimation of the prognostic score tends to make inference less reliable'' due to overfitting. 
This challenge persists with balancing weights. 
In this section, we introduce the cross-balancing method when estimated prognostic scores and other learned transformations of the raw covariates are used as balancing features, analyze its consistency and inferential properties, and provide recommendations on the choice of such features. 

\subsection{Procedure}
\vspace{-.2cm}

The procedure begins by randomly dividing the control and treated indices into two equal-sized folds, $\cI_c = \cI_{c,1} \cup \cI_{c,2}$ and $\cI_t = \cI_{t,1}\cup \cI_{t,2}$, respectively. 
Using data from the first fold, $\cI_{c,1} \cup \cI_{t,1}$, we estimate a set of features $\hat\phi^{(2)}(\cdot)$  which will be applied to  the second fold $\cI_{c,2} \cup \cI_{t,2}$. 
The process is repeated in the other direction: we estimate $\hat\phi^{(1)}(\cdot)$ on data from $\cI_{c,2} \cup \cI_{t,2}$ and apply it to $\cI_{c,1} \cup \cI_{t,1}$. 
For notational convenience, we let $\hat{\phi}(X_i) = \hat{\phi}^{(1)}(X_i)$ for $i \in \cI_{c,1} \cup \cI_{t,1}$ and $\hat{\phi}(X_i) = \hat{\phi}^{(2)}(X_i)$ for $i \in \cI_{c,2} \cup \cI_{t,2}$. 

The weights $\hat{w}_i$ for $i \in \cI_{c,1}$ are obtained by solving the  problem analogous to~\eqref{eq:minimal_weights2}: 
\begin{equation}
\hspace{-.25cm} \label{eq:opt_w2} \small
\argmin_{w} \Bigg\{ \sum_{i \in \mathcal{I}_{c,1}} w_i^2 : \, \bigg| \frac{1}{|\mathcal{I}_{c,1}|} \sum_{i\in \mathcal{I}_{c,1}} w_i \hat\phi_k(X_i) - \frac{1}{|\mathcal{I}_{t}|} \sum_{i\in \mathcal{I}_t} \hat\phi_k(X_i) \bigg|\leq \delta_k, k \in \mathcal{K}; \frac{1}{|\mathcal{I}_{c,1}|}\sum_{i\in \mathcal{I}_{c,1}} w_i = 1 \Bigg\}.
\end{equation}
Similarly, we compute $\hat{w}_i$ for $i \in \cI_{c,2}$ by solving the same optimization problem on the corresponding folds.
The estimator for the mean potential outcome is then 
\vspace{-1em}
{\setlength{\belowdisplayskip}{8pt}%
\$
\hat\theta_0 = \frac{1}{|\cI_c|} \sum_{i\in \cI_c} \hat{w}_i Y_i. 
\$}
Throughout this section, we assume that each $\hat\phi$ maps to $\RR^d$ for some fixed $d \in \NN^+$, and consider the asymptotic regime where $n \rightarrow \infty$ with $d$ held fixed. 
When appropriate, the program in~\eqref{eq:opt_w2} can be modified to enforce nonnegativity on the weights, which can improve the interpretability and robustness of our method. 

\subsection{Convergence analysis}
\vspace{-.2cm}

We first establish the convergence of the weights and the resulting estimator. 
 \vspace{-0.5em}
\begin{theorem}[Convergence and consistency]
\label{thm:convergence}
Suppose $\delta_n = o(1)$ and $\|\hat\phi^{(k)} - \phi^*\|_{L_2} = o_P(1)$ for some fixed function $\phi^* \colon \RR^p \to \RR^d$, where $\mathbb{E}[\phi^*(X)\phi^*(X)^\top \mid T = 0]$ is full rank, for $k = 1, 2$. 
Then the cross-balancing estimator satisfies $\hat\theta_0 = \bar\theta_0 + o_P(1)$, where $\bar\theta_0 = \mathbb{E}[\bar{w}(X)Y(0) \mid T = 0]$ and $\bar{w}(x) = \phi^*(x)^\top \bar\lambda$ for some fixed $\bar\lambda\in \RR^d$.
Moreover, $\bar\theta_0 = \theta_0$, and thus $\hat\theta_0$ is consistent, provided that at least one of the following holds: 
\vspace{-0.5em}
\begin{enumerate}[label=(\roman*)]
\setlength\itemsep{-0.5em}
    \item (\textit{Treatment model.}) There exists a vector $\lambda^* \in \RR^p$ such that $\frac{e(x)}{1-e(x)} \propto \phi^*(x)^\top \lambda^*$.
    \item (\textit{Outcome model.}) There exists a vector $\beta^* \in \RR^p$ such that $m_0(x)  = \phi^*(x)^\top \beta^*$.
\end{enumerate}
\end{theorem}
 
The proof of Theorem~\ref{thm:convergence} relies on convex analysis and is provided in Appendix~\ref{app:proof_convergence}.  
The first part of the theorem establishes convergence of the cross-balancing estimator to some $\bar\theta_0$ under arbitrary rates of convergence for the learned features.  
The second part connects to existing work on model double robustness in balancing estimators \citep{zhao2017entropy,wang2020minimal}, showing that balancing weighting methods implicitly estimate two models: one for the prognostic score and another for the inverse propensity score, where the model class is linear in $\hat\phi(\cdot)$.
Unlike existing work, allowing  $\hat\phi(\cdot)$ to be learned from a separate fold increases the expressiveness and robustness of the estimator.
 
Theorem~\ref{thm:convergence} further suggests that any model consistent for the true prognostic score suffices for consistent estimation; however, to maximize robustness to model misspecification, $\hat\phi$ should ideally approximate both the true weights and the prognostic score.
A straightforward idea is to fit models for the propensity and prognostic scores, set $\hat{w}^{(k)}(x):= \hat{e}^{(k)}(x)/\{1-\hat{e}^{(k)}(x)\}$ and $\hat{m}^{(k)}(x)$ for $m_0(x)$, and then take $\hat\phi^{(k)}(x) = \{\hat{m}^{(k)}(x), \hat{w}^{(k)}(x)\}$. 
By Theorem~\ref{thm:convergence}, $\hat\theta_0$ is consistent if either $\hat{m}^{(k)}(x)$ is consistent for $m_0(x)$ or $\hat{e}^{(k)}(x)$ is consistent for $e(x)$, or if any linear combination of the two converges to either $m_0(x)$ or $e(x)$. 
This form of double robustness naturally motivates a comparison with the AIPW estimator, which uses the same two models with cross-fitting~\citep{chernozhukov2018double}:
\vspace{-0.5em}
{\setlength{\belowdisplayskip}{8pt}\$
\hat\theta_0^{\text{AIPW}} := \frac{1}{n_c}\sum_{i\in \cI_c} \hat{w}(X_i) \{Y_i - \hat{m}(X_i)\} + \frac{1}{n_t}\sum_{j\in \cI_t} \hat{m}(X_j).
\$}
Interestingly, our cross-balancing estimator has a finite-sample bias reduction property when compared to the AIPW estimator. 
To see this, let $\Pi^{(k)}_{\hat\phi}$ denote the projection residual operator onto the linear span of $\hat\phi$. 
Specifically, for any function $f\colon \cX \to \RR$, define
$\Pi^{(k)}_{\hat\phi}[f](X_i) := f(X_i) - \hat\eta^\top \hat\phi^{(k)}(X_i)$ for $i \in \cI_{c,k},$
where 
$\hat\eta = \argmin_{\eta} \sum_{i\in \cI_{c,k}} \big\{f(X_i) - \eta^\top \hat\phi^{(k)}(X_i)\big\}^2$ is the vector of ordinary least squares coefficients. The proof of Proposition~\ref{prop:bias_reduc} is in Appendix~\ref{app:subsec_proof_bias}.  

\begin{prop}\label{prop:bias_reduc}
Suppose $\hat\phi^{(k)}$ includes both $\hat{w}^{(k)}(\cdot)$ and $\hat{m}^{(k)}(\cdot)$, each estimated using the held-out fold, and $\delta=0$. 
Then,
\vspace{-0.5em}
\begin{align}
    \label{eq:aipw_bias}
    \EE\big[\hat\theta_0^{\textnormal{AIPW}} - \theta_0\big] &= \EE\bigg[ \frac{1}{n_c} \sum_{k=1}^2 \sum_{i\in \cI_{c,k}}  \big( \hat{w}^{(k)} - w^* \big)(X_i) \cdot \big( \hat{m}^{(k)} - m_0 \big)(X_i) \bigg].
\end{align}
\vspace{-0.5em}
In contrast,  the cross-balancing estimator obeys
\vspace{-0.5em}
\begin{align}
    \label{eq:cbal_bias}
    \EE\big[\hat\theta_0 - \theta_0\big] = \EE\left[ \frac{1}{n_c} \sum_{k=1}^2 \sum_{i\in \cI_{c,k}} \Pi^{(k)}_{\hat\phi}\big[\hat{w}^{(k)} - w^*\big](X_i) \cdot \Pi^{(k)}_{\hat\phi}\big[\hat{m}^{(k)} - m_0\big](X_i) \right].
\end{align} 
\end{prop}
\vspace{-1em}
The bias formula~\eqref{eq:aipw_bias} for the AIPW estimator, which involves the product of two estimation error terms, is well established in the literature.  
In contrast, formula~\eqref{eq:cbal_bias} reveals that the bias of the cross-balancing estimator is governed by the product of the projection residuals of these two error terms. 
In other words, cross-balancing \emph{projects out} the component of each estimation error that is explainable by the linear span of $\hat\phi$, even though $w^*$ and $m_0$ are entirely unknown. 
Technically, this is due to the linear calibration enforced by (i) the balancing conditions and (ii) the $L_2$-norm minimization objective.
Although Proposition~\ref{prop:bias_reduc} does not guarantee that $|\EE[\hat\theta_0 - \theta_0]| \leq |\EE[\hat\theta_0^{\text{AIPW}} - \theta_0]|$ in all cases, our simulations indicate that cross-balancing often achieves considerably smaller bias than the AIPW estimator when using the same models $\{\hat{m}^{(k)}, \hat{w}^{(k)}\}$, especially when they are estimated using generic machine learning algorithms without carefully tuned hyperparameters.
 
\subsection{Inferential guarantees}
\vspace{-.2cm}

We now examine the inferential properties of $\hat\theta$ with fitted features. 
The following theorem establishes its asymptotic normality when the learned features converge at nonparametric rates. 
For brevity, the details regarding variance estimation are in the proof in Appendix~\ref{app:subsec_proof_inference}.
\begin{theorem}\label{thm:inference}
    Suppose the following conditions hold:
    \begin{enumerate}[label=(\roman*)]
    \setlength\itemsep{-0.5em}
        \item $\EE\left[\{Y(0)-m_0(X)\}^2 \mid T=0\right] \leq M_1$ for some constant $M_1 > 0$.
        \item The imbalance tolerance satisfies $\delta_n = o(n^{-1/2})$.
        \item For $k=1,2$, $\|\hat\phi^{(k)} - \phi^*\|_{L_2} = o_P(n^{-1/4})$ for some fixed $\phi^*\colon \cX\to \RR^d$, where $\EE[\|\phi^*\|^2] \leq M_2$ for a constant $M_2>0$, and $\EE[\phi^*(X)\phi^*(X)^\top\mid T=0]$ is full rank.
        \item There exist unknown vectors $\lambda^*,\beta^* \in \RR^d$ such that $w^*(x) = (\lambda^*)^\top \phi^*(x)$ and $m_0(x) = (\beta^*)^\top \phi^*(x)$, where $w^*(x)\equiv(1-p)e(x)/[p(1-e(x))]$.
    \end{enumerate}
    Then, $\sqrt{n}(\hat\theta_0 - \theta_0) \rightsquigarrow N(0, \sigma_0^2)$
    where $\sigma_0^2 = \Var\big\{\frac{(1-T)w^*(X)[Y(0)-m_0(X)]}{1-p} + \frac{T m_0(X)}{p}\big\}$
    can be consistently estimated and attains the semiparametric efficiency bound.
\end{theorem} 

In Theorem~\ref{thm:inference}, condition (i) imposes a mild moment restriction, while condition (ii) ensures that the bias due to approximate balance is of lower order. 
The core requirements are (iii) and (iv), which specify that the learned features converge  to a deterministic feature vector whose linear span contains both the oracle weight function $w^*(\cdot)$ and the prognostic function $m_0(\cdot)$. 
As discussed following Theorem~\ref{thm:convergence}, a natural idea is to include in the feature set $\{\hat{m}^{(k)}(x), \hat{w}^{(k)}(x)\}$, which estimate $m_0(\cdot)$ and $w^*(\cdot)$, respectively. Asymptotic normality follows
as long as these estimators converge at the rate $o_P(n^{-1/4})$ to (functions whose span includes) the true $m_0(\cdot)$ and $w^*(\cdot)$.
Condition (iii) is similar in spirit to the well-known rate-doubly-robust results~\citep{belloni2014inference,chernozhukov2018double}. 
In contrast to the usual consistency conditions required for the AIPW estimator, condition (iv)  is arguably milder: if one is not confident as to which model class is correct, one can incorporate a handful of them into $\hat\phi^{(k)}$. 
Valid inference is still attainable as long as at least one of the estimators converges to a limit whose linear span contains the true functions.

Theorem~\ref{thm:inference} underscores the importance of incorporating both the propensity score and the prognostic score in the choice of balancing features for the adaptive design of observational studies, especially when the estimated prognostic score is expected to converge at a slow rate. 
This aligns with practical guidelines for covariate balancing which emphasize the use of variables that are predictive of either the treatment or the outcome; see, e.g., \cite{kainz2017improving}. 
While prognostic score balancing approaches have been explored in the literature~\citep[e.g.,][]{hansen2008prognostic,stuart2013prognostic,nguyen2024confounder}, they may be insufficient for robust inference when the prognostic score is estimated  with slow convergence rates. 
In such cases, explicit modeling of the propensity score remains necessary to achieve fast convergence rates, unless the functions $\{m_0(x), w^*(x)\}$ are perfectly linearly related.   

\begin{remark}[Multiple validity with parametric models]
\label{rem:multi_param}
    As a final remark, the requirement that the limiting function $\phi^*$ provides a good approximation to both $m_0(x)$ and $w^*(x)$ is necessary for root-$n$ inference rates only when $\hat\phi^{(k)}(\cdot)$ converges slowly, as with nonparametric   methods. 
In contrast, when $\hat\phi^{(k)}$ is fitted using parametric models that converge quickly to some $\phi^*$, it suffices for valid inference that any entry of $\phi^*$ is well-specified for either $m_0(\cdot)$ or $w^*(\cdot)$. 
This observation is formalized in Theorem~\ref{thm:multi_param} presented in Appendix~\ref{app:subsec_multi_param}. 
\end{remark}

%% file: section_cross2.tex

This section introduces a second strategy for adaptive observational study design, which selects covariate features from a broader set or dictionary defined by the investigator. Standard guidelines typically recommend including covariates that are related to the treatment and/or outcome, guided by domain knowledge. However, supplementing domain expertise with data-driven feature selection techniques can be highly valuable as these methods harness information in the data itself. 
Selecting features from such a dictionary can improve interpretability and support integration with subject-matter expertise, since the selection process remains transparent and the resulting covariates are accessible for scientific scrutiny.

\subsection{Procedure}
\label{section_cross2.1}
\vspace{-.2cm}

Throughout this section, we slightly override the notation and let $X \in \mathbb{R}^p$ denote the entire dictionary of candidate features, which may consist of the (potentially high-dimensional) raw covariates in the data or transformations thereof. 
In keeping with the previous strategy, we randomly split the data, using one fold for covariate selection and the other for balancing the selected covariates.
In the first step, we use data from $\mathcal{I}_{c,k}\cup\mathcal{I}_{t,k}$ to find covariates that are predictive in a linear outcome or an inverse linear treatment model. 
In the second step, we balance the selected covariates using data from the other fold. 
For example, taking $k=1$, let $\hat{S} := \hat{S}^{(1)}$ denote the set of selected covariates from $\mathcal{I}_{c,2}\cup \mathcal{I}_{t,2}$. 
The weights $\hat{w}_i$ for $i \in \mathcal{I}_{c,1}$ are then obtained by solving 
\begin{equation}
\label{eq:opt_w2_sel}\small
\arg\min_{w} \left\{ \sum_{i \in \mathcal{I}_{c,1}} w_i^2 : \, \bigg| \frac{1}{|\mathcal{I}_{c,1}|} \sum_{i\in \mathcal{I}_{c,1}} w_i X_{i,\hat{S}} - \frac{1}{|\mathcal{I}_{t}|} \sum_{i\in \mathcal{I}_t} X_{i,\hat{S}} \bigg|\leq \delta_n; \frac{1}{|\mathcal{I}_{c,1}|}\sum_{i\in \mathcal{I}_{c,1}} w_i = 1 \right\},
\end{equation}
where we set a constant tolerance $\delta_n \in \mathbb{R}^+$ for all selected covariates.
Similarly, we obtain weights $\hat{w}_i$ for $i \in \mathcal{I}_{c,2}$. 
The resulting estimator is then $\hat\theta_0 = \frac{1}{|\cI_c|}\sum_{i \in \mathcal{I}_c} \hat{w}_i Y_i.$  

In contrast to the preceding section with a fixed  number of balancing features, we now allow both the original and selected feature dimensions to grow with the sample size $n$. 
The sources of estimator error also differ: here, the main source of uncertainty stems from $\hat{S}$ being a random subset of the given dictionary.
In this context, we highlight two key technical challenges that arise from feature selection in both low- and high-dimensional settings: 
\vspace{-0.25em}
\begin{itemize}\setlength\itemsep{-0.5em}
    \item The first is the analysis with a random, fuzzy set of balancing features. When $\hat{S}$ is chosen from low-dimensional  features, model selection often involves searching over an exponentially large space, and the selected set $\hat{S}$ may not converge to any fixed subset. 
    In high-dimensional settings where  $p$ grows with $n$, even in the well-studied regime with sparsity, accurately recovering the true predictors (if any) typically requires strong conditions~\citep{van2009conditions,buhlmann2011statistics}. 
    \item The second challenge is analysis with misspecification. In any regime, it may be unrealistic to assume that the outcome is a perfect linear combination of a subset of the available features. This prompts the question of how much deviation from linearity can be permitted while still achieving valid inference. 
\end{itemize}

To accommodate broad settings, we develop conditions for $\hat{S}$ that are mild and general enough to yield theoretical guarantees.
We begin with a standard sub-Gaussian condition.
\begin{assumption}
For each $j\in\{1,\ldots,p\}$, the variables $\{X_{ij}\}_{i=1}^n$ are i.i.d., sub-Gaussian random variables with mean zero. 
Furthermore, there exists an unknown function $\sigma(S)$ such that for any subset $S\subseteq\{1,\ldots,p\}$, the random vector $X_{i,S}$ is $\sigma^2(S)$-sub-Gaussian; that is, for any $\lambda \in \RR^{|S|}$, the linear combination $X_{i,S}^\top \lambda$ is sub-Gaussian with parameter $\|\lambda\|^2\sigma^2(S)$.
\end{assumption}
 
The general function $\sigma(S)$ here allows for feature correlation which is useful when $X_i$'s are given transformations of raw features. 
In the special case where entries of $X_i$ are independent, we may simply set $\sigma(S) = 1$. 
In high-dimensional statistics, it is common to assume $\sigma(\{1,\ldots,p\}) = \sigma^2$ for some constant $\sigma^2$~\citep{wainwright2019high,rigollet2023high}.

\subsection{Convergence analysis}
\label{section_cross2.2}
\vspace{-0.2cm}
 
Rather than characterizing the selection set $\hat{S}$, our analysis centers on its approximation capacity. 
For any fixed subset $S\subseteq\{1,\ldots,p\}$ and any function $f\colon\cX\to\RR$, we define 
\vspace{-0.5em}
{\setlength{\belowdisplayskip}{8pt}\$
\epsilon(S; f) = \inf_{\eta\in\RR^{|S|}} \left\|f(X) - \eta^\top X_S\right\|_{L_2(\PP_{\cdot \mid T=0})},
\vspace{-0.5em}
\$}
with the minimizer denoted by $\eta(S;f)\in\RR^{|S|}$. 
Intuitively, $\epsilon(S;f)$ measures how well $f$ can be approximated by a linear combination of variables in $S$. 
Similar quantities are used in~\cite{belloni2014inference} to describe the quality of variable selection and its role in inference when ordinary least squares, instead of generic balancing, is applied to the selected variables.
Theorem~\ref{thm:convergence_varsel_alt} establishes the convergence of $\hat\theta$ when the selection size may grow with $n$.
\begin{theorem}\label{thm:convergence_varsel_alt}
Suppose that for each $k=1,2$, the following conditions hold:
\begin{enumerate}[label=(\roman*)]
\setlength\itemsep{-0.5em}
    \item \textnormal{Approximation error:} There exist fixed functions $\bar{w},\,\bar{m} \colon \mathbb{R}^p \to \mathbb{R}$ such that
    $ 
    \|\bar{w} - \eta(\hat{S}^{(k)}, w^*)^\top X_{\hat{S}^{(k)}}\|_{L_2(\mathbb{P}_{\cdot\mid T=0})} = o_P(1)$,  $\epsilon(\hat{S}^{(k)};\bar{m}) = o_P(1)$, 
    and $\mathbb{E}[(Y(0)-\bar{m}(X))^2 \mid X] \leq M_1$ as well as $\bar{w}(X) \leq M_2$ almost surely, for some constants $M_1, M_2 > 0$. Moreover, $\|\eta(\hat{S}^{(k)}, w^*)\| \leq M_3$ for some constant $M_3 > 0$. 
    \item  \textnormal{Imbalance tolerance:} $\delta_n \|\eta(\hat{S}^{(k)};\bar{m})\|_1 = o_P(1)$, $\delta_n \|\eta(\hat{S}^{(k)};w^*)\|_1 = o_P(1)$, and
    $ 
    \delta_n = \Omega_P\big( \sqrt{\log |\hat{S}^{(k)}| / n} + \|\bar{w} - \eta(\hat{S}^{(k)}, w^*)^\top X_{\hat{S}^{(k)}}\|_{L_2(\mathbb{P}_{\cdot\mid T=0})} \big).
    $ 
    \item \textnormal{Selection size:} $|\hat{S}^{(k)}| = o_P(n)$.  
\end{enumerate}    
Then there exists a fixed value $\bar\theta \in \mathbb{R}$ such that $\hat\theta_0 = \bar\theta + o_P(1)$. 
Moreover, if either $\bar{w} = w^*$ or $\bar{m} = m_0$, then $\hat\theta_0 = \theta_0 + o_P(1)$.
\end{theorem}

The proof is provided in Appendix~\ref{subsec:proof_convergence_varsel_alt}, which relies on analysis techniques in penalized regression. 
The key technical challenge is to trade off the bias  (controlled by $\delta_n$) and variance (controlled by the eigenvalues of $X_{\hat{S}}X_{\hat{S}}^\top$) when balancing a \emph{random}, \emph{large} set of features in the presence of model misspecification. To avoid strong eigenvalue conditions on $X_{\hat{S}}X_{\hat{S}}^\top$, we leverage a sufficiently large value of $\delta_n$ (in analogy to penalized linear regression).  
A parallel result via convex analysis is in Appendix~\ref{app:subsec_varsel_conv}; it posits stronger eigenvalue conditions but allows smaller $\delta_n$, which may suit low-dimensional settings. 
 
\subsection{Inferential guarantees}
\label{section_cross2.3}

Finally, we establish inferential guarantees for the cross-balancing estimator following variable selection.
The proof of Theorem~\ref{thm:inf_varsel} is provided in Appendix~\ref{app:proof_thm_inf_varsel}.

\begin{theorem}
    \label{thm:inf_varsel}
    Suppose the following conditions hold for $k=1,2$:
    \begin{enumerate}[label=(\roman*)]
    \setlength\itemsep{-0.5em}
        \item \textnormal{Approximation error.} $\epsilon(\hat{S}^{(k)};w^*) = o_P(1)$ and $\epsilon(\hat{S}^{(k)};m_0) = o_P(1)$, with $\epsilon(\hat{S}^{(k)};w^*) \cdot \epsilon(\hat{S}^{(k)};m_0) = o_P(n^{-1/2})$. In addition, $\mathbb{E}\big[(Y(0) - \bar{m}(X))^2 \mid X\big] \leq M_1$ and $w^*(X) \leq M_2$, $|w^*(X) - \eta(\hat{S}^{(k)};w^*)^\top X_{\hat{S}}| \leq M_2$ almost surely, for some constants $M_1, M_2 > 0$. Also, $\|\eta(\hat{S}^{(k)}, m_0)\|_1 \leq M_3$ for some constant $M_3 > 0$.
        \item \textnormal{Imbalance tolerance.} $\delta_n = o_P(n^{-1/2})$.  
        \item \textnormal{Eigenvalues.} There exists a function $\xi(\cdot)\colon 2^{[p]} \to \mathbb{R}^+$ such that, with probability tending to one,
        $ 
        \frac{1}{n}\|X_{c,\hat{S}^{(k)}}v\|_2^2 \geq \xi(\hat{S}^{(k)}) \|v\|_1^2
        $ 
        for all $v\in \mathbb{R}^{|\hat{S}^{(k)}|}$, 
        where $X_{c,\hat{S}^{(k)}}$ denotes the data matrix whose $i$-th row is $X_{i,\hat{S}^{(k)}}$ for $i \in \cI_{c,k}$. Moreover, $\log(|\hat{S}^{(k)}|)/\xi(\hat{S}^{(k)}) = o_P(n)$ and $\epsilon(\hat{S}^{(k)};m_0) \cdot [\log(|\hat{S}^{(k)}|)/\xi(\hat{S}^{(k)})]^{1/2} = o_P(1)$.
    \end{enumerate}  
    Then $\sqrt{n}(\hat\theta_0 - \theta_0) \rightsquigarrow N(0,\sigma_0^2)$, where  $\sigma_0^2 = \Var\big\{\frac{(1-T)w^*(X)[Y(0)-m_0(X)]}{1-p} + \frac{T m_0(X)}{p}\big\}$ achieves the semiparametric efficiency bound for $\theta_0$. 
\end{theorem}

\begin{remark}[Finite-sample analysis]
    Although Theorem~\ref{thm:inf_varsel} establishes an asymptotic result, we also present a finite-sample analysis of the cross-balancing estimator in Appendix~\ref{app:subsec_thm_finite}. 
    This analysis gives an explicit error bound in terms of $\delta_n$, $\epsilon(\hat{S}^{(k)},w^*)$, $\epsilon(\hat{S}^{(k)};m_0)$, and $\xi(\hat{S}^{(k)})$, without using $o_P(\cdot)$ terms. 
    These finite-sample results can be used to recover both Theorem~\ref{thm:convergence_varsel_alt} and Theorem~\ref{thm:inf_varsel} as special cases under their respective assumptions.
\end{remark}

In Theorem~\ref{thm:inf_varsel}, condition (i) requires that the product $\epsilon(\hat{S}^{(k)};m_0)\cdot\epsilon(\hat{S}^{(k)};w^*)$ is $o_P(n^{-1/2})$, which quantifies the robustness of cross-balancing to both variable selection errors and model misspecification. 
Notably, this condition does not require that $m_0$ or $w^*$ be exactly linear functions of any subset of variables, nor that $\hat{S}^{(k)}$ include the set of truly relevant variables (if any). 
Instead, it suffices that linear combinations of the selected variables can jointly approximate the prognostic score and the true weights up to the stated product rate.

This has important practical implications for data-driven variable selection in cross-balancing procedures. For example, selecting only outcome-related variables may yield a small $\epsilon(\hat{S}^{(k)};m_0)$ but may not guarantee a sufficiently small product with $\epsilon(\hat{S}^{(k)};w^*)$; analogous issues arise when selecting only propensity-related variables. 
Our results theoretically justify  the practical guidance of~\cite{kainz2017improving}: to select the union of outcome-related and treatment-related variables. 
This strategy is particularly useful when one is concerned with variable selection uncertainty or model misspecification.

The last two assumptions in Theorem~\ref{thm:inf_varsel} address different sources of bias. 
Condition (ii) ensures that the bias in the outcomes from inexact balancing does not undermine root-$n$ inference. Condition (iii) controls the error in approximating the implicit weight model with an imperfect $|\hat{S}|$; it is automatically satisfied if $|\hat{S}|$ remains fixed (or bounded) as $n \to \infty$ and the covariance matrix of $X$ is full rank.
When $|\hat{S}|$ grows with $n$, condition (iii) resembles the restricted eigenvalue or compatibility conditions commonly used in high-dimensional statistics, which we discuss below.
\begin{remark}\label{rem:lasso_conditions}
    In condition (iii), setting $\xi(S) = |S|^{-1}$ recovers the compatibility condition~\citep{van2007deterministic}, which is typically weaker than the restricted isometry property~\citep{candes2006near} and the restricted $\ell_2$-eigenvalue condition~\citep{bickel2009simultaneous} (see also~\cite{van2009conditions} for a comparison). With this choice, condition (iii) is a combination of three requirements: (a) $\frac{1}{n}\|X_{c,\hat{S}^{(k)}}v\|_2^2 \cdot |\hat{S}^{(k)}| \geq \|v\|_1^2$ for all $v \in \mathbb{R}^{|\hat{S}^{(k)}|}$; (b) $|\hat{S}^{(k)}|\log(|\hat{S}^{(k)}|) = o_P(n)$; and (c) $\epsilon(\hat{S}^{(k)};m_0)\cdot |\hat{S}^{(k)}|^{1/2} \cdot \log(|\hat{S}^{(k)}|)^{1/2} = o_P(1)$. Here, (a) ensures that the selected columns of $X_c$ are not excessively correlated, a requirement commonly imposed in high-dimensional linear regression for the true sparse support, yet here applied to the selected features and can be readily checked empirically. Technically, (a) facilitates control of the bias arising from the (implicit) Lasso estimation of weights with a small penalty $\delta_n = o_P(n^{-1/2})$.
    Conditions (b) and (c) restrict the size of $\hat{S}^{(k)}$, which further controls the error in the implicit weight estimation. Under typical rates such as $\epsilon(\hat{S}^{(k)};w^*) = o_P(n^{-1/4})$ and $\epsilon(\hat{S}^{(k)};m_0) = o_P(n^{-1/4})$, condition (c) is satisfied if $|\hat{S}^{(k)}| \log(|\hat{S}^{(k)}|) = O_P(n^{1/2})$. 
\end{remark}

\subsection{Examples of variable selection procedures}
\vspace{-0.2cm}
Theorem~\ref{thm:inf_varsel} treats the variable selection procedure as a black-box, meaning that the selected set $\hat{S}^{(k)}$ produced by any selection procedure that satisfies the required conditions leads to the stated inferential properties. 
To further contextualize our results, we now discuss two concrete examples of selection procedures tailored to low- and high-dimensional settings, respectively, and their impact on the cross-balancing estimator. 
For convenience, we use $\hat{S}=\hat{S}^{(k)}$ to refer to variables selected from $\cI_t\cup\cI_c:=\cI_t^{(k)}\cup \cI_c^{(k)}$. The first example uses linear regression to select variables $\hat{S}$ related to both the prognostic and propensity scores.

\begin{example}[Linear regression]
\label{ex:lr}
     For the prognostic score, define $\hat\Sigma_c = \hat\EE_c[XX^\top]$, and let $\hat\beta = \hat\Sigma_c^{-1} \hat\EE_c[XY]$ be the OLS coefficient of the control outcomes. We select $\hat{S}_m=\{j\colon |\hat\beta_j/\hat\sigma_j| \geq \underline{\sigma}\}$ for some fixed constant $\underline{\sigma}>0$, where $\hat\sigma_j^2$ estimates the variance of $\hat\beta_j$. For example, $\underline{\sigma}=\Phi^{-1}(1-\alpha/2)/\sqrt{n_c}$ for $n_c=|\{i\colon T_i=0\}|$ yields an entry-wise $t$-test. For the propensity score,  
    we set $\hat\lambda = \hat\EE_c[XX^\top]^{-1} \hat\EE_t[X]$ and select $\hat{S}_w = \{j\colon |\hat\lambda_j/\hat\tau_j|>\underline{\tau}\}$ for some constant $\underline{\tau}>0$, where $\hat\tau_j^2$ estimates the variance of $\hat\lambda_j$. Then, the selected variables are $\hat{S}=\hat{S}_m\cup \hat{S}_w$. 
\end{example}

This procedure is justified in Proposition~\ref{prop:sel_lr} under approximately linear models. Its proof is in Appendix~\ref{app:proof_sel_lr}, where the construction of $\hat\sigma_j^2$ and $\hat\tau_j^2$ is detailed.

\begin{prop}\label{prop:sel_lr} 
Suppose $\hat\sigma_j^2$ and $\hat\tau_j^2$ are consistent for the true variances $\sigma_j^2$ and $\tau_j^2$. 
%
Assume there exist subsets $S_m,S_w\subseteq [p]$ such that $m_0(x) = x_{S_m}^\top \beta^* + \epsilon_m(x)$ and $w^*(x) = x_{S_w}^\top \lambda^* + \epsilon_w(x)$ for some vectors $\beta^*\in \RR^{|S_m|},\lambda^*\in \RR^{|S_w|}$ so that $\|\beta^*\|_\infty \geq c_0$ and $\|\lambda^*\|_\infty \geq c_0$ for some constant $c_0>0$. 
    Assume the residuals obey $\|\epsilon_m\|_{L_2}=o_P(n^{-1/4})$ and $\|\epsilon_w\|_{L_2}=o_P(n^{-1/4})$, and the thresholds obey $\underline{\sigma} \leq c_0 / \sigma_j + o_P(n^{-1/4})$ and $\underline{\tau} \leq c_0/\tau_j + o_P(n^{-1/4})$. Then, as $n\to \infty$, the selection set $\hat{S}$ obeys $\epsilon(\hat{S} ;w^*)=o_P(n^{-1/4})$ and $\epsilon(\hat{S};m_0)=o_P(n^{-1/4})$.
\end{prop}

In Proposition~\ref{prop:sel_lr}, under approximately linear outcome  and propensity models, as long as the significant thresholds $\bar\sigma$ and $\bar\tau$ decay properly, such as $\bar\sigma=o_P(1)$ and $\bar\tau=o_P(1)$ as $n\to \infty$, we achieve the desired conditions for the selection set. 

As our second example, we consider the high-dimensional setting where the dictionary of (transformed) features is of a large size. In such scenarios, a natural idea is to use the  Lasso~\citep{tibshirani1996regression} for high-dimensional variable selection. 

\begin{example}[Lasso]\label{ex:lasso}
    We consider Lasso estimators $\hat\beta = \argmin_{\beta} \frac{1}{n_c}\sum_{i\in \cI_c} (Y_i-X_i^\top \beta)^2 + \nu_1 \|\beta\|_1$ and $\hat\lambda = \argmin_{\lambda} (\lambda^\top \hat\Sigma_c \lambda - 2\lambda^\top \hat\EE_t[X] + \nu_2 \|\lambda\|_1)$ where $\hat\Sigma_c = \frac{1}{n_c}\sum_{i\in\cI_c} X_iX_i^\top$ and $\hat\EE_t[X]=\frac{1}{n_t}\sum_{j\in \cI_t}X_j$. Then, we select  $\hat{S}_m=\{j\colon \hat\beta_j\neq 0\}$ and $\hat{S}_w = \{j\colon \hat\lambda_j \neq 0\}$. 
\end{example}

Proposition~\ref{prop:sel_lasso} justifies setting $\hat{S} = \hat{S}_m \cup \hat{S}_w$ to approximately recover a powerful set of  features and the needed conditions are relatively weak. The proof is in Appendix~\ref{app:proof_sel_lasso}. 
 
As preparation, we say the compatibility condition~\citep{buhlmann2011statistics} with constant $\phi_0$ holds for a subset $S\subseteq \{1,\dots,p\}$ and data matrix $X \in \RR^{n\times p}$ if for any $\beta\in \RR^p$ obeying $\|\beta_{S^c}\|_1\leq 3\|\beta_S\|_1$, it holds that $ \|\beta_S\|_1^2 \leq (|S|/\phi_0^2) \|X\beta\|_2^2/n$. 

\begin{prop}\label{prop:sel_lasso}
Suppose each entry of $X$ is sub-Gaussian, and $\max_{j=1}^p|X_j|\leq M_1$ almost surely. 
Suppose there exists subsets $S_m, S_w\subseteq \{1,\dots,p\}$ with $s_m = |S_m|$, $s_w = |S_w|$, and constant vectors $\beta^*, \lambda^*\in \RR^p$ such that $\beta^*_{S_m^c}=0$ and $\lambda^*_{S_w^c}=0$, and $m_0(x) = x^\top \beta^* + \epsilon_m(x)$, $w^*(x)=x^\top \lambda^* + \epsilon_w(x)$ for some residual functions $\epsilon_m(\cdot)$ and $\epsilon_w(\cdot)$. Let $\hat{S}=\hat{S}_m\cup \hat{S}_w$ be obtained with penalties $\nu_1,\nu_2>0$ as in Example~\ref{ex:lasso}, and assume the following conditions: 
\begin{enumerate}[label=(\roman*)]
\setlength\itemsep{-0.5em}
    \item $\max\{|Y-m_0(X)|,w^*(X),|\epsilon_m(X)|,|\epsilon_w(X)|\}\leq M_2$ for some constant  $M_2>0$. 
    \item For any sufficiently small $\varepsilon>0$, there exists a sample size $n>0$ such that $\nu_1 \geq 4(\|\EE_c[X_j\epsilon_m(X)]\|_\infty + 2M_2 \sqrt{\frac{2\log(2p/\varepsilon)}{n_c}})$ and $\nu_2 \geq 4(\|\EE_c[X_j\epsilon_w(X)]\|_\infty + 2M_2 \sqrt{\frac{2\log(2p/\varepsilon)}{n_c}})$.
    \item The compatibility condition holds for data $X$ and subsets $S_m$, $S_w$ with constant $\phi_0$.
\end{enumerate}
Then, 
$
\epsilon(\hat{S};m_0) = \Omega(\max\{\|\epsilon_m\|_{L_2}, |S_m|\nu_1\})
$, 
$
\epsilon(\hat{S};w^*) = \Omega(\max\{\|\epsilon_w\|_{L_2}, |S_w|\nu_2\})
$ 
as $n\to \infty$. 
\end{prop}

There are several key approximation error quantities in Proposition~\ref{prop:sel_lasso}. Besides the functional approximation error $\|\epsilon_m\|_{L_2}$ and $\|\epsilon_w\|_{L_2}$ which enter directly as a term in $\epsilon(\hat{S};m_0)$ and $\epsilon(\hat{S};w^*)$, two other key terms are $\nu_1|S_m|$ and $\nu_2|S_w|$, which are also be small under mild conditions.  
First, $|S_m|$ and $|S_w|$ are small under approximately sparse models. Regarding $\nu_1$, taking the outcome model as an example, the ideal choice is $\nu_1\asymp \|\EE_c[X_j\epsilon_m(X)]\|_\infty + \sqrt{\log (p)/n}$. If the entries in $X$ are independent and $\epsilon_m(x)$ is only a function of features in $S_m^c$, the final approximation error is $\epsilon(\hat{S};m_0)=O(\|\epsilon_m\|_{L_2} + |S_m|\sqrt{\log(p)/n})=o_P(n^{-1/4})$ given that the residual has an $L_2$-norm of $o_P(n^{-1/4})$ and $|S_m|=o_P(n^{-1/4}\sqrt{\log (p)})$.  
In general, we always have $\|\EE_c[X_j\epsilon_m(X)]\|_\infty \leq \|\epsilon_m\|_{L_2}$, which translates to $\epsilon(\hat{S};m_0) = \Omega(\|\epsilon_m\|_{L_2}|S_m| + |S_m|\sqrt{\log(p)/n})$ with an ideal choice of $\nu_1\asymp \|\epsilon_m\|_{L_2} + \sqrt{\log (p)/n}$. Again, we emphasize that we do not require the true models to be exactly sparse linear for these results to hold.

%% file: section_simu.tex

\section{Simulation study}
\vspace{-.2cm}
\label{sec:simu}

We evaluate the performance of cross-balancing through various simulation study designs. 
Section~\ref{subsec:simu_fitted} examines the variant introduced in Section~\ref{sec:fitted}, where the balancing covariates are predictions from machine learning models. 
Section~\ref{subsec:simu_varsel} focuses on the variant introduced in Section~\ref{section_cross2}, where the balancing covariates are selected from a larger dictionary.

\subsection{Evaluation of cross-balancing with learned features}
\vspace{-.2cm}
\label{subsec:simu_fitted}

We first evaluate the variant where the features are learned, random functions. 
To assess the roles of the training algorithms and the data-generating process (DGPs), we consider two broad settings: 
(i) \emph{Low-dimensional setting,} with nonlinear regression functions (Section~\ref{subsubsec:simu_fitted_lowd}) fit by random forests in the \texttt{grf} R package~\citep{athey2019generalized}.  
(ii) \emph{High-dimensional setting,} where the true functions depend on a sparse subset of variables (Section~\ref{subsubsec:simu_fitted_highd}), fit by cross-validated Lasso in the \texttt{glmnet} R package~\citep{hastie2014glmnet}.  
 
Following Theorem~\ref{thm:inference}, throughout this section we define $\hat{m}(x)$ as an estimated prognostic score for $\EE[Y(0)\mid X=x]$. 
We also define $\hat{w}(\cdot)\propto \hat{e}(\cdot)/(1-\hat{e}(\cdot))$ as a preliminary weight, where $\hat{e}(\cdot)$ is a fitted propensity score model for $e(x)$. 
We compare seven methods.   
Two cross-balancing procedures: (i) \texttt{Cross\_Bal}: the method introduced in Section~\ref{sec:fitted}, with $\hat\phi(x) = (\hat{m}(x), \hat{w}(x))$; (ii) \texttt{Cross\_Bal\_prog}: the same as above but with $\hat\phi(x) = \hat{m}(x)$, i.e., only the estimated prognostic score, included to show the role of the propensity score in our procedure. 
Two non-sample-splitting counterparts: (iii) \texttt{Naive\_Bal}: using all data to fit $\hat{m}(x), \hat{w}(x)$ and to obtain balancing weights without sample splitting; (iv) \texttt{Naive\_Bal\_prog}: the non-sample-splitting version of \texttt{Cross\_Bal\_prog}. 
Also, (v) \texttt{AIPW}: the cross-fitting version of the AIPW estimator~\citep{robins1994estimation} as in~\cite{chernozhukov2018double}, with the same sample splitting scheme and same fitted outcome models and preliminary weights; we include this method to examine bias reduction properties (Proposition~\ref{prop:bias_reduc}). 
Finally, two oracle methods: (vi) \texttt{Oracle\_CBal}: the oracle version with $\hat\phi(x) = (m_0(x), w^*(x))$, which describes cross-balancing in the absence of estimation error; 
(vii) \texttt{Oracle\_AIPW}: the AIPW estimator using $(m_0(x), w^*(x))$, which serves as a benchmark for the highest achievable performance. 

We evaluate estimators by bias, standard deviation, and root mean squared error (MSE). 
For inference, we construct confidence intervals (CIs) using two approaches. (i) Bootstrap: we apply the procedures to $B=1000$ resamples with replacement from the original data, and define the CI by the empirical $\alpha/2$-th and $(1-\alpha/2)$-th quantiles of the bootstrap estimates. (ii) Wald-type: we estimate the asymptotic variance in Theorem~\ref{thm:inference} and construct CIs by normal approximation.
All simulations are repeated independently $N=500$ times. 

\subsubsection{Low-dimensional settings}
\vspace{-.25cm}
\label{subsubsec:simu_fitted_lowd}

 Details of the DGPs  in the low-dimensional setting are in Appendix~\ref{app:dgp_fitted_lowd}; we outline the design ideas here.
The first is the well-known example in~\cite{kang2007demystifying}, which highlights the challenges of extreme propensity scores when building weights. 
 Settings 2–5 follow~\cite{chattopadhyay2020balancing}, with logistic propensity scores and linear or nonlinear prognostic scores, under strong or weak overlap. 
Random forests are expected to adequately recover the prognostic and propensity scores, but may not converge at a parametric rate. 

\begin{figure}[htbp]
    \centering
    \includegraphics[width=\linewidth]{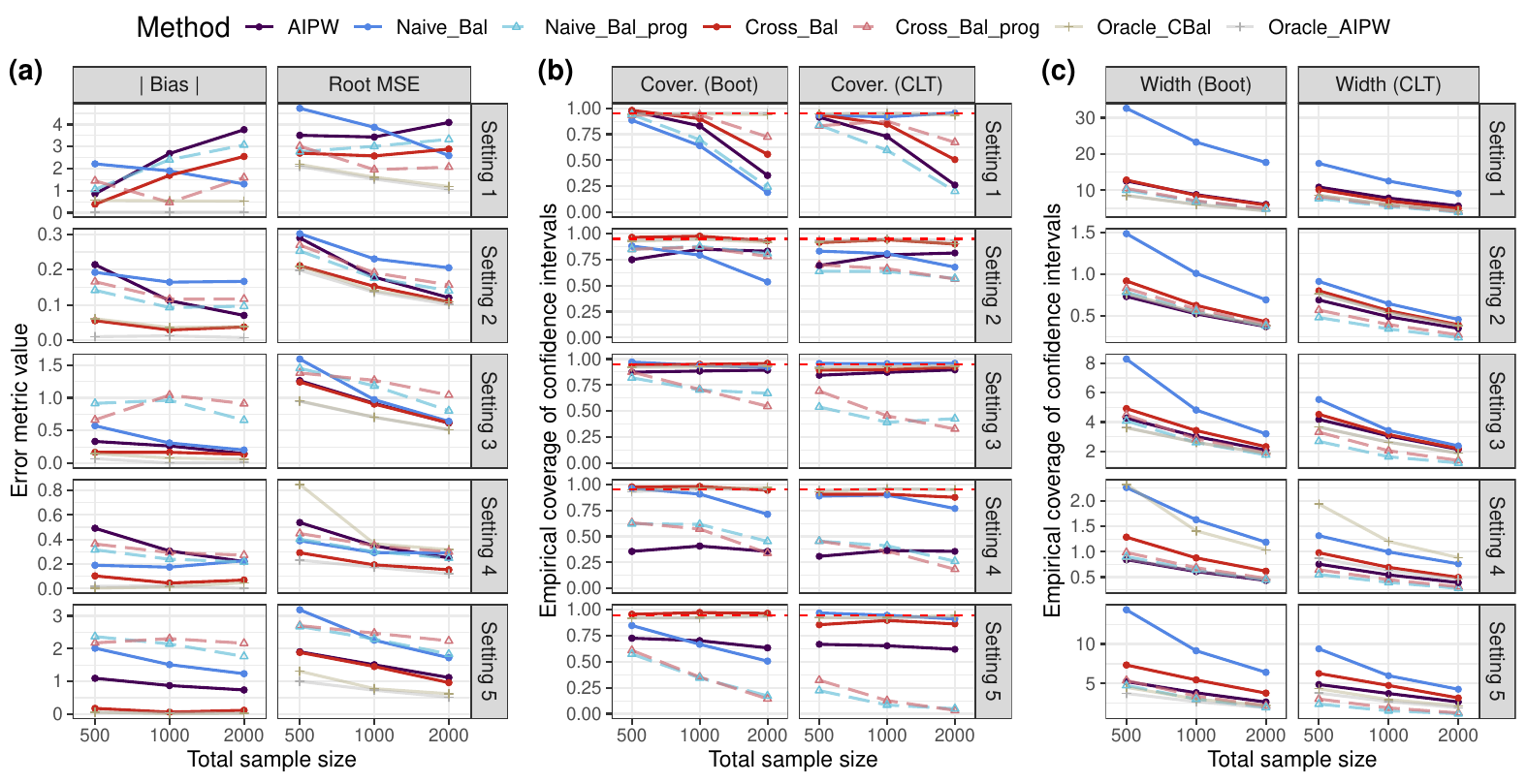}
    \caption{{\small Estimation and inference performance of cross-balancing and other methods using learned features in low-dimensional settings. 
Each row corresponds to a simulation setting, and each column to a performance metric. 
\textbf{(a)}: estimation bias and root mean squared error (RMSE) for each method. 
\textbf{(b)}: empirical coverage of confidence intervals using the bootstrap (Boot) and Wald-type (CLT) inference. 
\textbf{(c)}: average lengths of both types of confidence intervals.}}
    \label{fig:err_lowd}
\end{figure}
Figure~\ref{fig:err_lowd} displays the performance of each method across various sample sizes. 
Panel (a) shows that cross-balancing achieves the lowest estimation error among all non-oracle methods, frequently approaching the performance of their oracle counterparts. 
When comparing \texttt{Cross\_Bal} with \texttt{Naive\_Bal}, we observe that sample splitting substantially improves accuracy, particularly for complex models such as random forests with slow convergence rates. 
The comparison of \texttt{Cross\_Bal} and \texttt{AIPW} further highlights a meaningful reduction in bias (see Proposition~\ref{prop:bias_reduc}), and underscores the benefit of explicitly pursuing in-sample balance. 
Finally, compared to methods that rely solely on the estimated prognostic score, incorporating a propensity score model improves accuracy by reducing the product error.

Panels (b) and (c) display the coverage and lengths of the CIs. 
In the most challenging Setting 1, coverage  declines for all methods as sample size grows as random forests exhibit  larger bias due to extreme propensity score.  
Nevertheless, \texttt{Cross\_Bal} maintains the highest coverage.
In the remaining four settings, \texttt{Cross\_Bal} attains nominal coverage and produces comparatively short CIs.
By contrast, \texttt{Naive\_Bal} exhibits low coverage even with wider CIs, which highlights the importance of sample splitting for valid inference. 
Methods that rely solely on the prognostic score exhibit substantial under-coverage, underscoring the importance of incorporating a propensity score model for robustness.
Finally, in Settings 4-5 (weak overlap), the decline in estimation quality (particularly the propensity score) results in lower coverage for \texttt{AIPW}, further highlighting the advantage of direct covariate balance.

\subsubsection{High-dimensional settings}
\vspace{-.25cm}
\label{subsubsec:simu_fitted_highd}
 
In this part, we use four DGPs with high dimensions ($d=100$) and apply the Lasso to fit the prognostic and propensity score models. The DGPs are described in detail in Appendix~\ref{app:dgp_fitted_highd}, and we outline the ideas here. 
In Setting 1, both the prognostic score and the propensity score are sparse linear functions.
In Setting 2, the prognostic score is sparse linear while the propensity score is sparse linear in certain transformations of the features.
In Setting 3, the prognostic score is sparse linear in transformed features, and the propensity score is sparse linear in the raw features.
In Setting 4, both scores are approximately sparse linear in the raw features, each with a small additive nonlinear term. 
Since the Lasso produces a linear fit, these DGPs investigate the robustness of cross-balancing to model misspecification.

The performance of the various methods is summarized in Figure~\ref{fig:err_highd}. 
In panel (a), we again observe that \texttt{Cross\_Bal} achieves the lowest error among the methods considered. 
Consistent with the preceding part, incorporating an additional propensity score model and enforcing in-sample balance consistently improve performance. 
However, the gap between \texttt{Cross\_Bal} and \texttt{Naive\_Bal} is less pronounced here, which we attribute to the stable estimation of the Lasso. 
With sparse linear models (Setting 2), the prognostic model tends to be estimated more accurately than the propensity model, resulting in a smaller performance gap.
 
As shown in panels (b) and (c), 
\texttt{Cross\_Bal} achieves good coverage with relatively narrow confidence intervals (CIs) under both bootstrap and Wald-type inference. 
In contrast, the two prognostic-score-only methods consistently undercover, highlighting the advantages of incorporating the propensity score model for valid inference.
\texttt{AIPW} also undercovers in Settings 1, 3, and 4, though it attains the best coverage in Setting 2, yet with much wider CIs.
Lastly, while \texttt{Naive\_Bal} provides reasonable coverage, it typically results in wider CIs.

\begin{figure}[htbp]
    \centering
    \includegraphics[width=\linewidth]{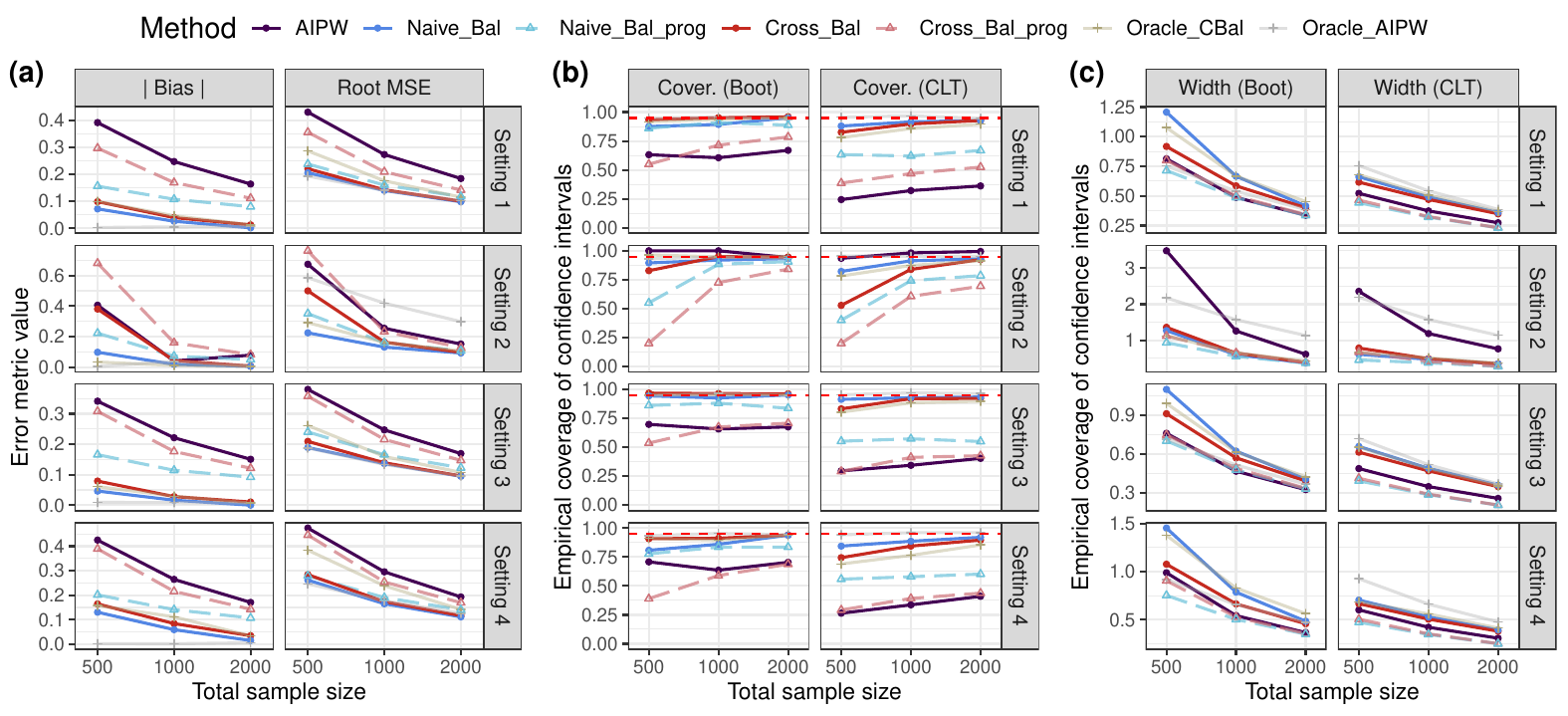}
    \caption{{\small Estimation and inference performance of cross-balancing and other methods using learned features in high-dimensional settings. All other details are as described in Figure~\ref{fig:err_lowd}.}}
    \label{fig:err_highd}
\end{figure}

\subsection{Evaluation of cross-balancing with selected variables}
\label{subsec:simu_varsel}
\vspace{-.2cm}

We next study the variant in Section~\ref{section_cross2} where balancing covariates are selected from a large dictionary of features.
We consider two settings where variable selection can be helpful: (1) \emph{Selection from raw features} (Section~\ref{subsubsec:simu_varsel_raw}).
Among high-dimensional ($d=100$) raw features, the investigator believes that only a subset is relevant. 
Selecting a smaller set of balancing covariates improves robustness by reducing the number of balancing constraints, and enhances interpretability. (2) \emph{Selection from a basis expansion} (Section~\ref{subsubsec:simu_varsel_basis}).  
Given a handful of raw features ($d=10$), to increase model expressiveness, one includes nonlinear transformations of them and yields a much higher-dimensional set of basis functions. 
Selecting the most relevant ones to balance  improves expressiveness, interpretability, and robustness.  

We assess cross-balancing when combined with different variable selection strategies.
To streamline the exposition, we focus only on cross-balancing-based methods and Wald-type CIs.
We compare four estimators:
(i) \texttt{Cross\_Bal}: In one fold of data, we select variables relevant for the prognostic and propensity score models, and take their union as the  balancing covariates for creating weights in the other fold, and vice versa; (ii) \texttt{Outc\_Only}: Only variables selected in the prognostic model are used in balancing; (iii) \texttt{Prop\_Only}: Only variables selected in the propensity model are used in balancing; (iv) \texttt{SBW}: Stable balancing weights, using all the raw features or their transformations without variable selection.
Methods (i)-(iii) are combined with three widely-used variable selection strategies, without any parameter tuning: 
(a) \texttt{lm}: linear regression, where a (generalized) linear model is fitted for the outcome or treatment using all features, and variables with $t$-test p-values less than 0.05 are selected; (b) \texttt{aic}: forward stepwise regression with the AIC criterion from the \texttt{stats} R library, where the selection starts from an empty model and iteratively adds variables from the dictionary; (c) \texttt{lasso}: cross-validated-Lasso (for the prognostic score) or GLM-Lasso (for the propensity score), using the \texttt{glmnet} R package with the \texttt{1se} rule, and variables with nonzero fitted coefficients are selected. 
We assess estimation accuracy, the average number of selected balancing covariates, and the coverage and width of the CIs. 
Across all settings, we set the total sample size $n=1000$, tolerance  $\delta = 0.01$, and enforce positive weights in the \texttt{sbw} R package\footnote{Experiments without enforcing positive weights yield qualitatively similar results; we show results with positive weights to investigate the robustness of our theoretical results.}.  
Results are averaged over $N=500$ independent runs.

\subsubsection{Selection from original variables}
\label{subsubsec:simu_varsel_raw}
\vspace{-.25cm}

In this section, we specify three DGPs in which only a small subset of the $d = 100$ raw features in $X$ are relevant for the prognostic and propensity scores; however, precisely recovering them is challenging, leading to nontrivial variable selection error. Detailed descriptions of all DGPs are in Appendix~\ref{app:dgp_varsel_raw}, and we provide the design ideas here.
The first DGP is adapted from~\cite{kang2007demystifying}, but with the true prognostic and propensity scores as linear and log-linear functions, plus a small nonlinear additive component so perfect model recovery is unattainable due to non-linearity. 
The second DGP specifies both models as linear in a subset of the features, but includes two features in both models that have a very weak signal in the prognostic score. Thus, variable selection methods often fail to identify them in the prognostic model. 
The third setting is inspired by~\cite{belloni2014inference} where both scores combine several strong linear signals and quadratically decaying linear terms of other features; thus, variable selection error persists for the variables with weak signals.

The performance of all competing methods is summarized in Figure~\ref{fig:simu_varsel_raw}. 
Compared to \texttt{SBW}, applying variable selection with \texttt{Cross\_Bal} yields substantially more parsimonious models, greater estimation accuracy, and narrower CIs. 
Among the balancing methods incorporating variable selection, \texttt{Cross\_Bal} delivers the most accurate estimates and reliable inference. 
In contrast, \texttt{Outc\_Only} and \texttt{Prop\_only} often omit important variables, resulting in higher estimation errors. 
Specifically, in settings 1 and 3, the propensity-only approach suffers from elevated error and poor coverage, while in setting 2, \texttt{Outc\_Only} performs worse (as discussed, in this setting important variables can be missed due to weak signals in the outcome model).  
Comparing variable selection strategies, \texttt{lasso} produces the most concise models, typically reducing the number of features from $d=100$ to ten or fewer without compromising accuracy.
However, this advantage is observed most clearly when combined with \texttt{Cross\_Bal}.
The \texttt{aic} also performs well overall, but tends to select larger models.

\begin{figure}
    \centering
    \includegraphics[width=0.9\linewidth]{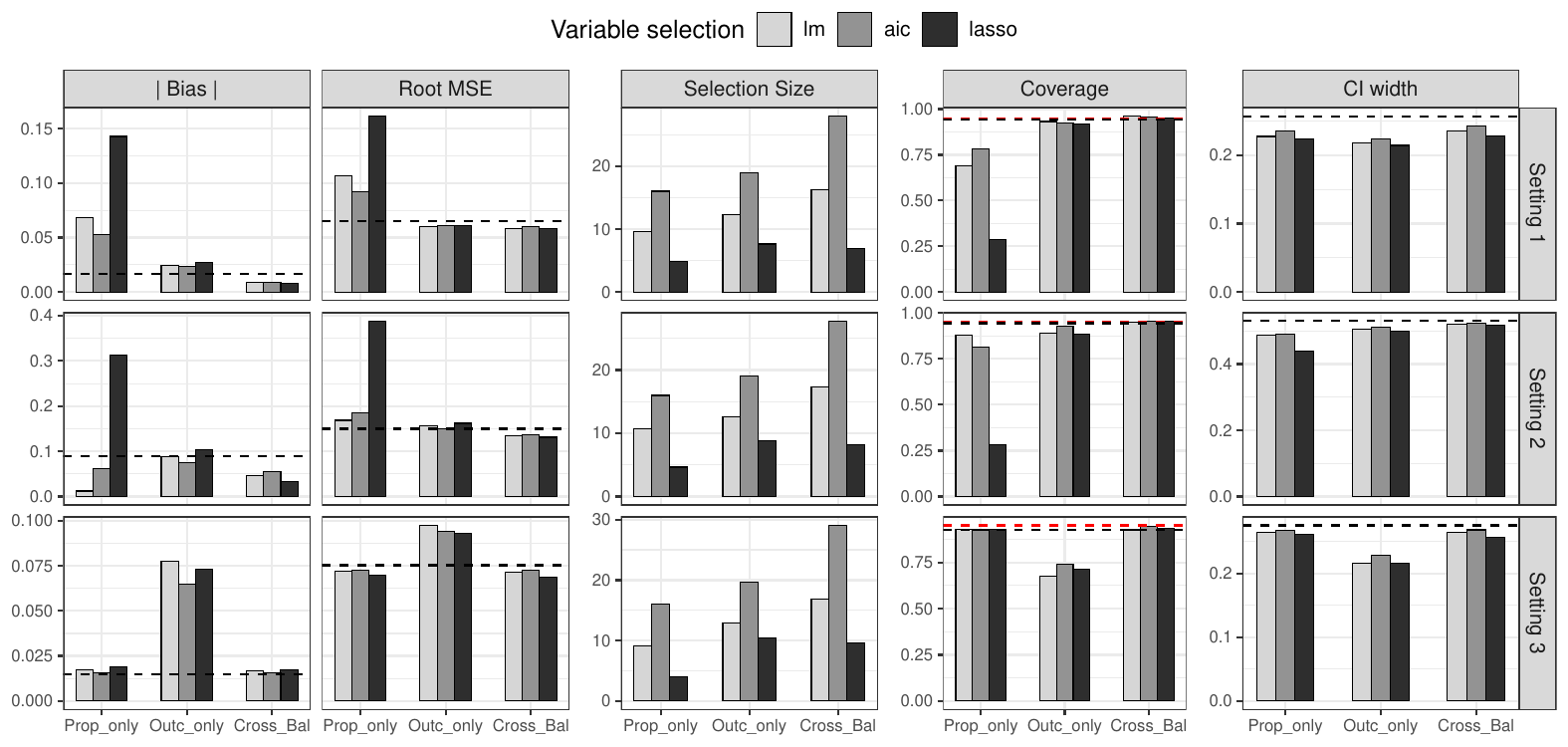}
    \caption{
    {\small Performance of cross-balancing for covariate selection from high-dimensional raw features. Each column displays a different performance metric, and each row corresponds to a simulation setting. Within each subplot, results are shown for three estimators combined with three variable selection strategies. The performance of \texttt{SBW} (which does not involve variable selection) is indicated by a dashed horizontal line.}}
    \label{fig:simu_varsel_raw}
\end{figure}

\subsubsection{Selection from a dictionary of basis functions}
\label{subsubsec:simu_varsel_basis}
\vspace{-.25cm}

We next consider settings where transformations (basis functions) of a few raw features are generated, and selection is performed from this expanded set before balancing.
We specify three settings inspired by~\cite{kang2007demystifying}, where prognostic and/or propensity scores may be complex nonlinear functions of the raw features (details in Appendix~\ref{app:dgp_varsel_basis}). 
In the first setting, the true prognostic and propensity scores are linear in certain transformations of the raw features that differ from, but can be approximated by, the basis functions. 
For settings 2-3, we modify either the prognostic or propensity score in setting 1 to be linear in several basis functions of the raw features, with a small nonlinear additive term. 
Here, the true models are not exactly in the span of the basis functions but can be well approximated. The goal is to select a small subset of relevant basis functions to improve expressiveness while preserving interpretability.
More specifically, starting with the raw features, we construct a set of basis functions that includes: (i) all raw features $X_j$, their squares $X_j^2$, and cubes $X_j^3$; and (ii) cubic splines for each individual raw feature, with 10 equally spaced knots. 
We note that these basis functions may exhibit considerable dependence. 
The same variable selection and estimation methods as in Section~\ref{subsubsec:simu_varsel_raw} are then applied.

Figure~\ref{fig:simu_varsel_basis} presents the performance of all methods. 
The large dictionary of basis functions poses substantial challenges for \texttt{sbw} (dashed line), resulting in much higher estimation error and poorer coverage compared to \texttt{Cross\_Bal}. 
Among the methods with variable selection, \texttt{Cross\_Bal}, which takes the union of variables associated with the outcome and the propensity score, performs the best. 
\texttt{Outc\_Only} also yields competitive performance, whereas \texttt{Prop\_only} results in poorer estimates. 
We also observe that these variable selection methods generally require stronger signals in the propensity model than in the prognostic model to recover relevant features.
When comparing the three variable selection strategies, \texttt{lm} proves to be less effective, potentially due to the dependence among basis functions and the limitations of its heuristic approach. 
In contrast, both \texttt{aic} and \texttt{lasso} produce concise sets of selected variables, leading to accurate estimation and valid Wald-type inference.

\begin{figure}
    \centering
    \includegraphics[width=0.9\textwidth]{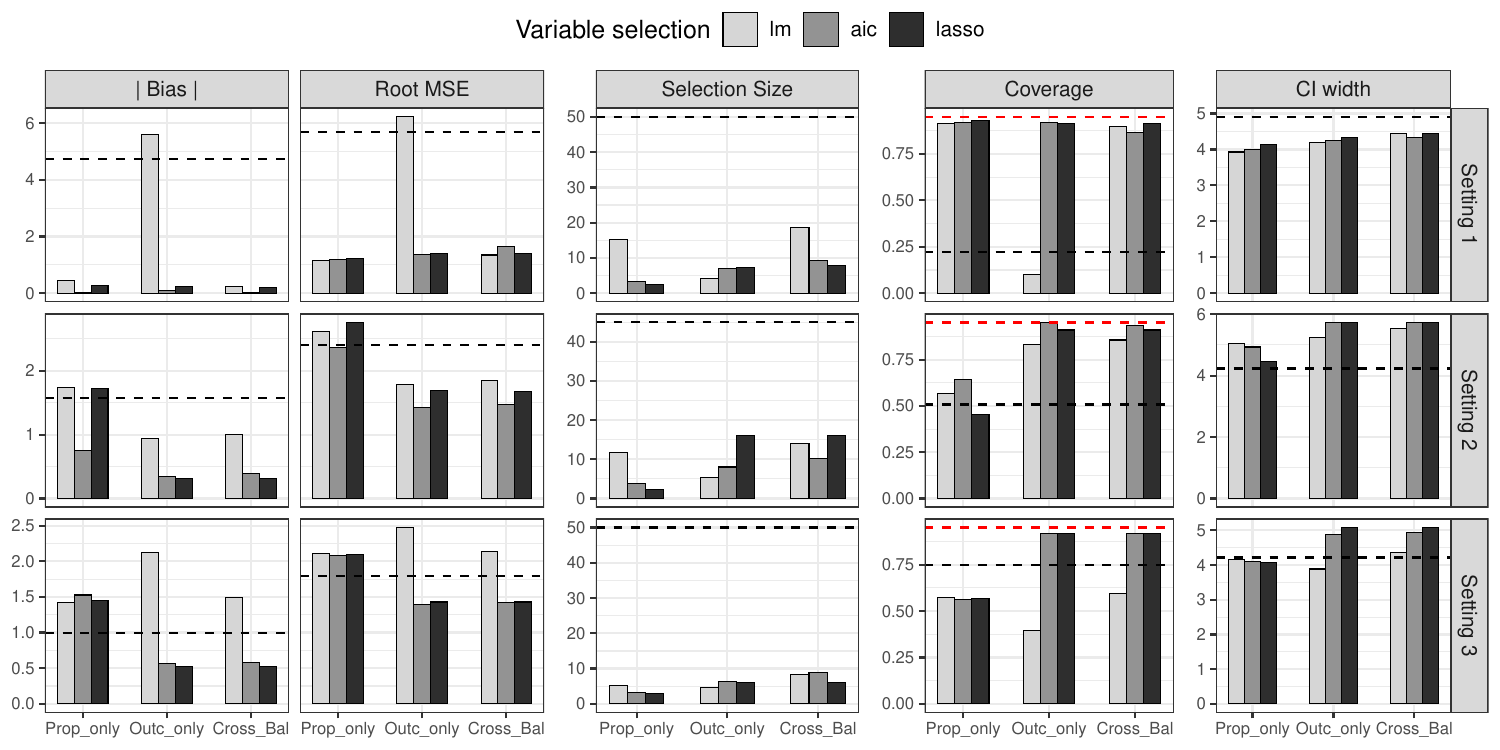}
    \caption{{\small Performance of cross-balancing and other methods when selecting from cubic spline basis functions across all settings in Section~\ref{subsubsec:simu_varsel_basis}. Details are otherwise the same as in Figure~\ref{fig:simu_varsel_raw}.}}
    \label{fig:simu_varsel_basis}
\end{figure} 

To further examine the quality and role of variable selection, we include Figure~\ref{fig:approx_err} in Appendix~\ref{app:subsec_simu_varsel} to show the approximation errors $\epsilon(\hat{S}^{(k)}; m_0)$, $\epsilon(\hat{S}^{(k)}; w^*)$, and their product across all settings in this section, where $\hat{S}^{(k)}$ denotes the selected variables for each of the three estimators, \texttt{Cross\_Bal}, \texttt{Outc\_Only}, and \texttt{Prop\_Only}. 
The product of these errors reflects the approximation error in cross-balancing that arises from imperfect modeling. 
As expected, \texttt{Cross\_Bal} achieves the smallest approximation error for a given variable selection strategy, since its selected set includes that of the others. 
Outcome approximation error is larger when using propensity-selected variables alone, and vice versa. 
Crucially, combining the sets leads to a substantial reduction in the product approximation error, which directly affects the accuracy of the final estimator. 
Among selection strategies, both \texttt{aic} and \texttt{lasso} tend to yield better approximation quality, and therefore we recommend their use in practice.

%% file: section_case.tex

In this section, we demonstrate the use of cross-balancing in an observational study of the effect of binge drinking on high blood pressure \citep{rosenbaum2023second}. The data come from the US National Health and Nutrition Examination Survey (NHANES, 2017–March 2020, prior to the pandemic) and are available through the \texttt{iTOS} package for R \citep{rosenbaum2025introduction}.

We primarily focus on the variable selection variant of cross-balancing to estimate the counterfactual mean $\EE[Y(0) \mid T = 1]$. 
We randomly split the data into two folds, construct a dictionary of basis functions and apply the method in Section~\ref{section_cross2} to perform variable selection and compute balancing weights. 
(For comparison, cross-balancing with learned features using random forests produces an estimate of similar magnitude.) 
Variable selection is performed using three procedures: \texttt{lm}, \texttt{aic}, and \texttt{lasso}, following the same steps as in the simulation study. 
We consider a nested sequence of four feature dictionaries. 
Each successive dictionary extends the previous by including: (i) raw features; (ii) cubic spline transformations of the raw features; (iii) all pairwise interactions among the raw features, including terms up to cubic order; and (iv) all pairwise interactions among the spline-transformed features, also up to cubic order.
In each of the 12 configurations (three selection methods $\times$ four feature sets), we find outcome-related variables $\hat{S}_m$ and propensity-related ones $\hat{S}_w$. 
We then compute the cross-balancing estimator using three sets: the union $\hat{S} = \hat{S}_m \cup \hat{S}_w$ (\texttt{Cross\_Bal}), the outcome-only set $\hat{S} = \hat{S}_m$ (\texttt{Outc\_only}), and the propensity-only set $\hat{S} = \hat{S}_w$ (\texttt{Prop\_only}). 

Figure~\ref{fig:real_estimator} shows the corresponding estimates along with their $95\%$ CIs.
We find that incorporating cubic terms expands the expressiveness of the balancing features and yields visible differences compared to estimates that use only the original features. 
Increasing the size of the feature dictionary causes a slight increase in the length of the confidence intervals.

\begin{figure}
    \centering
    \includegraphics[width=0.8\linewidth]{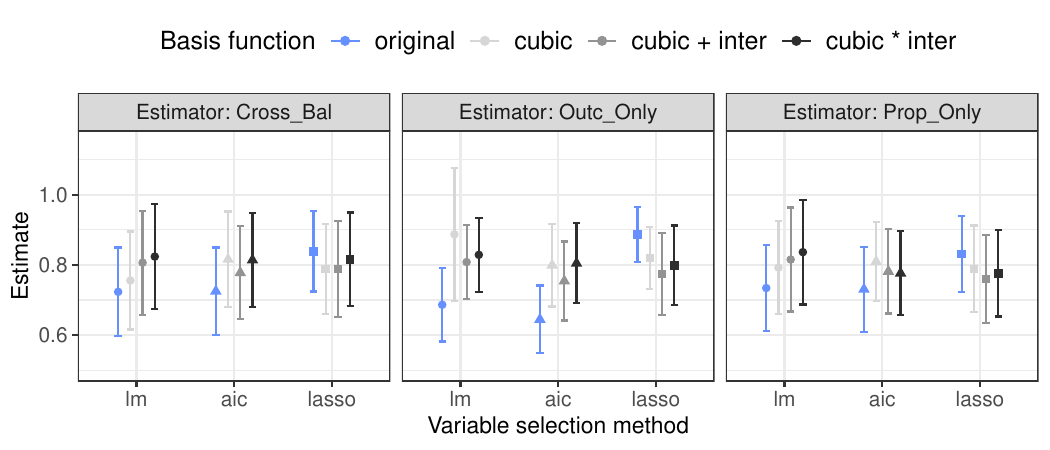}
    \caption{Point estimates and 95\% confidence intervals for the counterfactual mean $\EE[Y(0)\mid T=1]$. Each subplot displays results for one estimator: cross-balancing with both outcome- and treatment-related variables, outcome-related variables only, or treatment-related variables only. Within each subplot, point estimates are shown for three variable selection methods, applied to four feature dictionaries.}
    \label{fig:real_estimator}
\end{figure}

Figure~\ref{fig:real_imbalance} displays the standardized mean differences of all features in dictionary (ii) between treated and control groups, both before (triangles) and after (dots) weighting with cross-balancing using the set $\hat{S} = \hat{S}_m \cup \hat{S}_w$. 
Cross-balancing selects a concise set of relevant transformed features, listed in the left panel. 
As expected, the weights substantially improve balance on these selected variables. 
Notably, the right panel shows that balance is also markedly improved (approaching exact balance) for features \emph{not explicitly selected}, despite the fact that balance on these variables was not directly enforced. 
These results show that cross-balancing effectively finds and balances the most critical features, while also conferring improved balance more broadly across the covariate space.
Figure~\ref{fig:real_imbalance_large} in Appendix~\ref{app:subsec_case_study} shows the balancing metrics with dictionary (iii) which convey similar messages.

\begin{figure}
    \centering
    \includegraphics[width=0.9\linewidth]{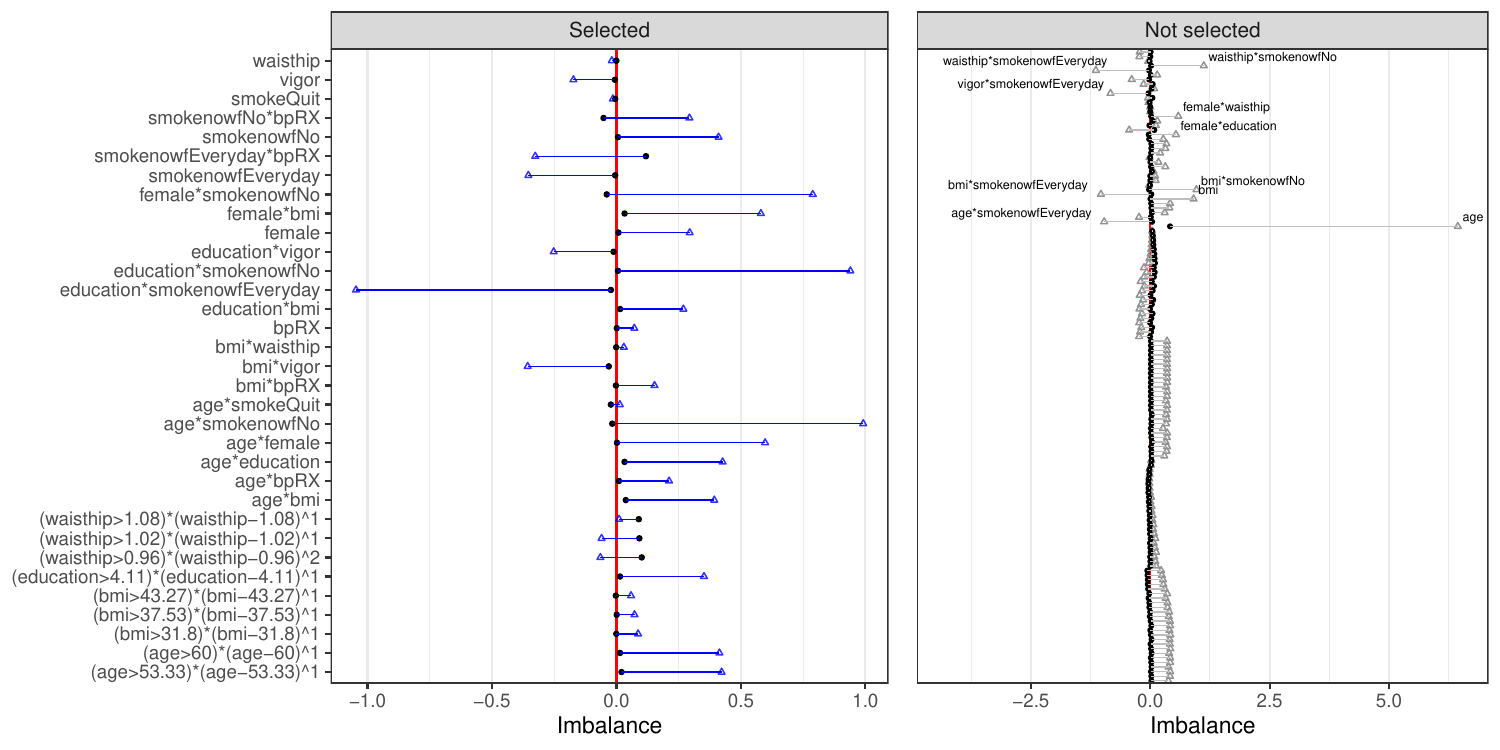}
    \caption{
    {\small Imbalance in the NHANES data before (triangles) and after (dots) weighting with cross-balancing, where cubic spline basis functions are selected using cross-validated Lasso. The left panel displays imbalance for variables explicitly balanced in at least one fold (the union selected across two folds), while the right panel shows variables not selected and thus not explicitly balanced.}}
    \label{fig:real_imbalance}
\end{figure}  

\begin{figure}[h!]
    \centering
    \includegraphics[width=0.9\linewidth]{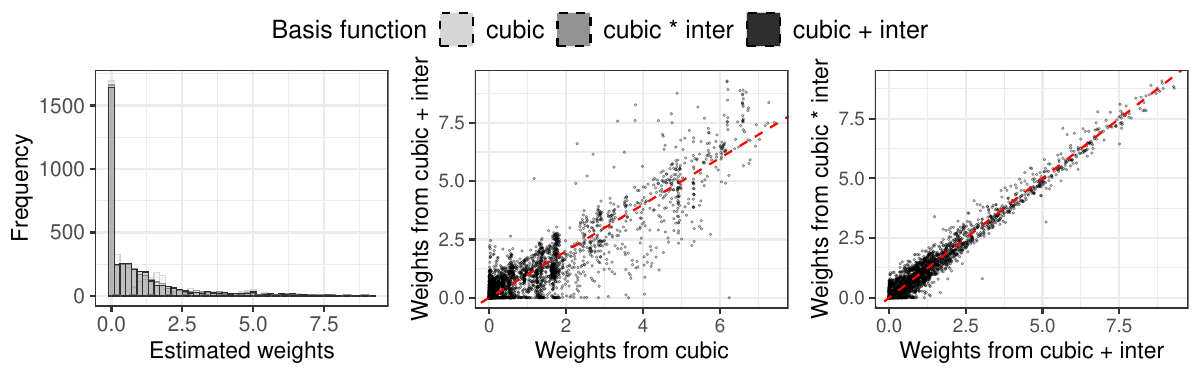}
    \caption{
    {\small Learned weights from cross-balancing, with variables selected using cross-validated Lasso. The left panel shows the empirical distributions of the weights for three dictionaries. The middle and right panels display pairwise comparisons of individual weights estimated with each dictionary.}}
    \label{fig:real_weights}
\end{figure}

As a robustness check,  Figure~\ref{fig:real_weights} depicts the weights $\{w_i\}$ produced by \texttt{Cross\_Bal} with Lasso variable selection for three feature dictionaries, excluding the one with only raw features. 
We see the empirical distribution of the weights remains stable across the three dictionaries, and the estimated weights in the middle and right panels do not change substantially as the feature dictionary is expanded. These results suggest that in this case study adding cubic terms to the spline functions offers little practical benefit for weight estimation. 

\section{Discussion}
\vspace{-0.2cm}
This paper presents two complementary strategies for building data-informed balancing features in observational studies: learning key functions and selecting relevant features. Although they rely on distinct techniques, both strategies yield estimators with strong theoretical guarantees and robust empirical performance in simulation and real-world studies.

We emphasize that the strategies in sections \ref{sec:fitted} and \ref{section_cross2} are not mutually exclusive; indeed, they can—and arguably should—be combined. For example, the estimated prognostic score can be included in the candidate feature dictionary. A straightforward extension of our theory would then establish consistency and asymptotic normality of the resulting estimator, provided the combined set approximates the prognostic and weight functions with a product rate of $o_P(1/\sqrt{n})$. 

Also, these strategies do not stand in opposition to the traditional guidelines of using subject-matter knowledge to inform the design of observational studies; rather, they can complement and strengthen each other. 
Cross-balancing provides a principled framework for observational study design, where one half of the data can be used to design the study, while the other half is reserved for analysis, and vice versa. 
If data-informed feature selection happens to omit covariates deemed essential by expert judgement, these variables can always be incorporated into the analysis to preserve study interpretability and scientific relevance.
Ultimately, cross-balancing supports both efficient estimation and valid inferences, helping to maintain the rigor of the study.
Future work includes the extension of cross-balancing to longitudinal settings with limited overlap.

%% file: appendix.tex

\section{Additional discussion, details and results}

\subsection{Multiple validity with parameteric models}
\label{app:subsec_multi_param}

Theorem~\ref{thm:multi_param} formalizes the discussion in Remark~\ref{rem:multi_param}. 

\begin{theorem}[Multiply-valid inference with parametric models]\label{thm:multi_param}
    Let $\{m_{j,\theta_j}(\cdot)\}_{j=1}^{J}$ be distinct parametric models for the prognostic score $m_0(x)$, each parameterized by $\theta_j\in \Theta_j \subseteq \RR^{d_j}$, and let $\{e_{\ell,\eta_\ell}(\cdot)\}_{\ell=1}^{L}$ be distinct parametric models for the propensity score $e(x)$, each parametrized by $\eta_\ell\in \Lambda_\ell \subseteq \RR^{d_\ell}$. 
    Suppose $\{\hat\theta_j^{(k)}\}_{j=1}^{J}$ and $\{\hat\eta_\ell^{(k)}\}_{\ell=1}^{L}$ are fitted on $\cI_{c,k}\cup\cI_{t,k}$ and satisfy $\hat\theta_j^{(k)} = \bar\theta_j+ \frac{1}{n}\sum_{i\in\cI_{c,k}\cup\cI_{t,k}} f_j(D_i)+O_{P}(1/{n})$ and $\hat\eta_\ell = \bar\eta_\ell^{(k)} + \frac{1}{n}\sum_{i\in\cI_{c,k}\cup\cI_{t,k}} g_\ell(D_i)+O_{P}(1/n)$ for some fixed parameters $\{\bar\theta_j\}_{j=1}^{J}\cup\{\bar\eta_\ell\}_{\ell=1}^{L}$ and influence functions $\{f_j\}_{j=1}^J\cup\{g_\ell\}_{\ell=1}^L$. 
	Define $\hat\theta_0$ as the cross-balancing estimator constructed with the features
    \$
    \hat\phi^{(k)}(x) = \Big\{m_{1,\hat\theta_1^{(k)}}(x),\dots,m_{J,\hat\theta_{J}^{(k)}} (x),\textstyle{\frac{e_{1,\hat\eta_1^{(k)}}(x)}{1-e_{1,\hat\eta_1^{(k)}}(x)}},\dots \textstyle{\frac{e_{L,\hat\eta_{L}^{(k)}}(x)}{1-e_{L,\hat\eta_{L}^{(k)}}(x)}}\Big\},\quad k=1,2.
    \$
	Then $\hat\theta_0 = \theta_0 + n^{-1}\sum_{i=1}^n \gamma(X_i,T_i) + O_{P}(n^{-1})$ for some function $\gamma(\cdot)$ under the following conditions: (i) The tolerance level satisfies $\delta_n = o(n^{-1/2})$. (ii) The vector
    $ 
    (m_{1,\bar\theta_1}(X), \ldots, m_{J,\bar\theta_J}(X)$, $\frac{e_{1,\bar\eta_1}(X)}{1-e_{1,\bar\eta_1}(X)}, \ldots, \frac{e_{L,\bar\eta_L}(X)}{1-e_{L,\bar\eta_L}(X)})
    $ 
    with $X\sim \mathbb{P}_{X| T=0}$ is of full rank.
    (iii) Either $m_0(x) = m_{j,\bar\theta_j}(x)$ for some $j \in [J]$, or $e(x) = e_{\ell,\bar\eta_\ell}(x)$ for some $\ell \in [L]$.
\end{theorem}

In Theorem~\ref{thm:multi_param}, the first two conditions are regularity requirements analogous to those used previously. 
Condition~(iii) says that, if $\hat\phi^{(k)}$ includes several parametric candidate models for both the prognostic and propensity functions, the cross-balancing estimator remains consistent and asymptotically normal provided at least one model is correctly specified. 
Thus, practitioners may incorporate multiple plausible models to enhance robustness to model misspecification.

\subsection{Alternative convergence analysis for variable selection}
\label{app:subsec_varsel_conv}

In this part, we provide an alternative convergence analysis for cross-balancing after variable selection, which might better suite low-dimensional settings. 

For a symmetric matrix $A \in \mathbb{R}^{d \times d}$, let $\nu_{\min}(A)$ denote its minimum eigenvalue and $\|A\|_{\mathrm{op}} = \sup_{\|x\|=1} |x^\top A x|$ its operator norm. 
For any $S \subseteq \{1, \dots, p\}$, define $\Sigma_{c,S} = \mathbb{E}[X_S X_S^\top \mid T = 0]$ as the population covariance matrix in the control group, and let $\nu_{\min}^\Sigma(S) = \nu_{\min}(\Sigma_{c,S})$.

We slightly extend the sub-Gaussian condition in the main text.
\begin{assumption}\label{assump:subgaussian_app}
For each $j\in\{1,\ldots,p\}$, the variables $\{X_{ij}\}_{i=1}^n$ are i.i.d., sub-Gaussian random variables with mean zero. 
Furthermore, there exists an unknown function $\sigma(S)$ such that for any subset $S\subseteq\{1,\ldots,p\}$, the random vector $X_{i,S}$ is $\sigma^2(S)$-sub-Gaussian; that is, for any $\lambda \in \RR^{|S|}$, the linear combination $X_{i,S}^\top \lambda$ is sub-Gaussian with parameter $\|\lambda\|^2\sigma^2(S)$.
\end{assumption}
 
The general function $\sigma(S)$ in the vector sub-Gaussian condition allows for feature correlation, which is particularly relevant when the variables $X_i$ are constructed as transformations of raw features. 
In the special case where the entries of $X_i$ are independent, we may simply take $\sigma(S) = 1$. 
In high-dimensional statistics, it is common to assume $\sigma(\{1,\ldots,p\}) = \sigma^2$ for some constant $\sigma^2$~\citep{wainwright2019high,rigollet2023high}.

The proof of Theorem~\ref{thm:convergence_varsel} is provided in Appendix~\ref{subsec:proof_convergence_varsel}.

\begin{theorem}\label{thm:convergence_varsel}
Suppose that, for each $k=1,2$, the following conditions hold:
\begin{enumerate}[label=(\roman*)]
    \item \textnormal{Approximation error:} There exist fixed functions $\bar{w}, \bar{m} \colon \mathbb{R}^p \to \mathbb{R}$ such that
    $ 
    \|\bar{w} - \eta(\hat{S}^{(k)}, w^*)^\top X_{\hat{S}^{(k)}}\|_{L_2(\mathbb{P}_{\cdot \mid T=0})} = o_P(1)
    $
    and $\epsilon(\hat{S}^{(k)}; \bar{m}) = o_P(1)$, 
    with $\mathbb{E}[(Y(0) - \bar{m}(X))^2 \mid X] \leq M_1$ and $\bar{w}(X) \leq M_2$ almost surely for some constants $M_1, M_2 > 0$. In addition, $\|\eta(\hat{S}^{(k)}, w^*)\| \leq M_3$ for some constant $M_3 > 0$.
    
    \item \textnormal{Selection size:} 
    $ 
    \delta_n \|\eta(\hat{S}^{(k)}; \bar{m})\| = o_P(1)$, $\delta_n |\hat{S}^{(k)}| = o_P(1)$,  $|\hat{S}^{(k)}| \log |\hat{S}^{(k)}| = o_P(n)$. 
    \item \textnormal{Covariance eigenvalues:} 
    $ 
    \zeta(\hat{S}^{(k)}) = o_P(1)$,  $\zeta(\hat{S}^{(k)})^2 \cdot \|\Sigma_{c,\hat{S}^{(k)}}\|_{\mathrm{op}} = o_P(1)$, 
    where for any $S \subseteq \{1, \dots, p\}$,
    $ 
    \zeta(S) := \sigma(S) \cdot \nu_{\min}^{-1}(\Sigma_{c,S}) \cdot |S|^{1/2} n^{-1/2}$ where $\sigma(S)$ is defined in Assumption~\ref{assump:subgaussian_app}.
\end{enumerate}
Then there exists some fixed $\bar\theta \in \mathbb{R}$ such that $\hat{\theta}_0 = \bar{\theta} + o_P(1)$. 
Furthermore, $\bar{\theta} = \theta_0 + o_P(1)$ if either $\bar{w} = w^*$ or $\bar{m} = m_0$.
\end{theorem}

The requirements for $\hat\theta_0$ to converge to a fixed value are weaker than those needed for consistency. 
The first condition in (i), namely $\|\bar{w} - \eta(\hat{S}^{(k)},w^*)^\top X_{\hat{S}^{(k)}}\|_{L_2(\PP_{\cdot\given T=0})} = o_P(1)$, ensures that the OLS projection of $w^*$ onto the variables in $\hat{S}^{(k)}$ converges to some fixed function. 
This does not necessarily imply that $\hat{S}^{(k)}$ provides a good approximation of $w^*$ itself; that is, $\epsilon(\hat{S}^{(k)}, w^*)$ may not converge to zero. 
For example, suppose that $w^*$ is linear in the union of two disjoint subsets $S_1$ and $S_2$, but $\hat{S}^{(k)}$ only recovers $S_1$ due to variable selection error.
In this case, the projection onto $\hat{S}^{(k)}$ can still converge to a fixed function even though the overall approximation error remains nonzero.

The last two conditions control the growth rate and tail behavior of the selected balancing features.  
Condition (ii) ensures that the bias introduced by inexact balancing (with imbalance tolerance $\delta_n$) remains negligible, and that the number of selected features does not increase too rapidly with the sample size.  
Condition (iii) is designed to accommodate settings with correlated features. 
In the special case where the $X_{ij}$ are i.i.d.~standard normal, we have $\|\Sigma_{c,\hat{S}^{(k)}}\|_{\textnormal{op}} \leq 1$, $\sigma(S) = 1$, and $\nu_{\min}(\Sigma_{c,\hat{S}^{(k)}}) = 1$, so that condition (iii) is automatically satisfied whenever (ii) holds.  
Finally, if either $\bar{m} = m_0$ or $\bar{w} = w^*$, the cross-balancing estimator is consistent for $\theta_0$.

\subsection{Finite-sample analysis after variable selection}
\label{app:subsec_thm_finite}

Theorem~\ref{thm:inf_varsel_finite}, whose proof is in Appendix~\ref{app:subsec_inf_varsel_finite}, offers a finite-sample analysis relying on concentration inequalities. 
It can be used to derive Theorem~\ref{thm:convergence_varsel_alt} and Theorem~\ref{thm:inf_varsel} under the respective given conditions. 

\begin{theorem}[Finite-sample version of Theorem~\ref{thm:inf_varsel}]
\label{thm:inf_varsel_finite}
    Suppose the following conditions hold for $k=1,2$:
    \begin{enumerate}[label=(\roman*)]
        \item \textnormal{Boundedness.} There exists a constant $M_1>0$ such that $ Y(0)-m_0(X) $ is $M_1^2$-sub-Gaussian conditional on $X$, and $|m_0(X)-\eta(\hat{S}^{(k)};m_0)^\top X_{\hat{S}}|\leq M_1$. In addition, $w^*(X)\leq M_2$, $|{w}^*(X) -\eta(\hat{S}^{(k)};w^*)^\top X_{\hat{S}}|\leq M_2$  almost surely for some constant $M_2>0$, and $\|\eta(\hat{S}^{(k)},m_0)\|_1\leq M_3$ for some constant $M_3>0$. 
        \item There exists some function $\xi(\cdot)\colon 2^{[p]}\to \RR^+$ such that with probability tending to $1$, $\frac{1}{n}\|X_{c,\hat{S}^{(k)}}v\|_2^2 \geq \xi(\hat{S}^{(k)}) \|v\|_1^2$ holds for any $v\in \RR^{|\hat{S}^{(k)}|}$, where the $i$-th row of $X_{c,\hat{S}^{(k)}}$ is  $X_{i,\hat{S}^{(k)}}$, $i\in \cI_{c,k}$. 
    \end{enumerate}  
    Then, with probability at least $1- \varepsilon$, it holds that 
    \$
    | \hat\theta_0 - \hat\theta_0^* | \leq \tilde\cO \Bigg(  \frac{\sqrt{2} \cdot \epsilon_m\cdot \epsilon_w}{\sqrt{n_c}} &+ \epsilon_m \cdot 6M_2\sqrt{\log(64/\varepsilon)}\cdot \bigg(\frac{1}{\sqrt{n_c}}+\frac{1}{\sqrt{n_t}}\bigg) \\ 
    & \qquad \qquad \qquad + \frac{2\sqrt{2} \cdot \delta_n \cdot (\epsilon_m + M_1  \sqrt{2\log(64/\varepsilon)})}{\sqrt{n_c\cdot \bar{\xi}_S}}      \Bigg),
    \$
    where $\tilde\cO(\cdot)$ hides additive terms that are of a higher order in $(n_c,n_t)$ than those we display inside, and $\hat\theta_0^* =\frac{1}{n_c}\sum_{i\in \cI_{c}} {w}^*(X_i)(Y_i- m_0(X_i)) + \frac{1}{n_t}\sum_{j\in \cI_{t}}  m_0(X_j)  $ is the semi-parametrically efficient estimator. Here, we define the quantities $\epsilon_m = \max_{k=1,2}\epsilon(\hat{S}^{(k)};m_0)$, $\epsilon_w = \max_{k=1,2}\epsilon(\hat{S}^{(k)}; w^*)$, and  $\bar\xi_S = \min_{k=1,2} \xi(\hat{S}^{(k)})$. 
\end{theorem}

\subsection{Data generating processes in Section~\ref{subsubsec:simu_fitted_lowd}}
\label{app:dgp_fitted_lowd}

In this part, we detail the data generating processes in Section~\ref{subsubsec:simu_fitted_lowd}, where we study the performance of the method in Section~\ref{sec:fitted} in a low dimensional setting. 
In all settings, we first generate i.i.d.~features $X_i$, and then generate the control outcome $Y_i(0) = m_0(X_i)+\epsilon_i$ for some prognostic score function $m_0(\cdot)$ and i.i.d.~exogenous errors $\{\epsilon_i\}$, as well as independently drawn treatment assignment indicators $T_i\sim \text{Bern}(e_i)$ for some propensity score $e_i\in (0,1)$. 

Setting 1 follows that of~\cite{kang2007demystifying}. We first generate $U\in \RR^{6}$ each entry independent from $N(0,1)$, and set $X\in \RR^4$ whose entries are given by $X_1 = e^{U_1/2}$, $X_2 = U_1/(1+e^{U_1})+10$, $X_3 = (U_1U_3/25+0.6)^3$, and $X_4 = (U_2U_4+20)^2$. The outcome is $Y(0) = 210+27.4U_1+13.7U_2+13.7U_3+ 13.7U_4+\epsilon$, where $\epsilon\sim N(0,1)$. For the $i$-th sample, the  propensity score is given by $e_i = 1/(1+e^{U_{i,1}-0.5U_{i,2}+0.25U_{i,3}+0.1U_{i,4}})$.

Settings 2-5 follow those in~\cite{chattopadhyay2020balancing}. In all cases, the features are $X\in \RR^6$, where $X_{1:3}\sim N(0,\Sigma)$ with $\Sigma = \begin{pmatrix}
    2&1&-1\\1&1&-0.5\\-1&-0.5&1
\end{pmatrix}$, $X_4 \sim \textrm{Unif}([-3,3])$, $X_5 = \chi_1^2$, and $X_6 \sim \text{Bern}(0.5)$ are independent of each other. 
The outcomes ($\epsilon_i \sim N(0,1)$) and treatments ($e_i=e(X_i)$) are generated with the following two designs: 
\begin{itemize}
    \item In design A, $m_0(x) = \sum_{j=1}^6 x_j$, and $e(x) = \Phi((x_1+2x_2-2x_3-x_4-0.5x_5+x_6)/\sigma)$, where $\Phi(\cdot)$ is the cumulative distribution function of a standard Gaussian random variable.
    \item In design B, $m_0(x) = (x_1+x_2+x_5)^2$, and $e(x) = \Phi((x_1+2x_2-2x_3-x_4-0.5x_5+x_6)/\sigma)$.
\end{itemize}
Specifically, setting 2 is design A with $\sigma=10$ (strong overlap); setting 3 is design B with $\sigma=10$; setting 4 is design A with $\sigma=\sqrt{30}$ (weak overlap), and setting 5 is design B with $\sigma=\sqrt{30}$.

\subsection{Data generating processes in Section~\ref{subsubsec:simu_fitted_highd}}
\label{app:dgp_fitted_highd}

Section~\ref{subsubsec:simu_fitted_highd} studies the performance of the method in Section~\ref{sec:fitted} in a high-dimensional setting where $\hat\phi$ contains the fitted prognostic score and a preliminary weight function involving a fitted propensity score. 
Again, we generate $Y_i(0)=m_0(X_i)+\epsilon_i$ where $\epsilon_i\sim N(0,1)$, and $T_i \sim \textrm{Bern}(e(X_i))$ for some function $e(\cdot)$. 
The dimensions are fixed for $X_i\in \RR^p$, $p=100$, and each element of $X_i$ is independent from a standard Gaussian clipped within $[-1,1]$ to ensure reasonable overlap. 

To diversity the model specification scenarios, we define a nonlinear transformation of the observed features   $g(X_i)\in \RR^{p}$, where $g(x)= (e^{x_{ 1}/2}, \sin(x_{ 2}), x_{3}^2, x_6/(1+e^{x_1}), x_4,x_5,x_{7:100})$. We define $\alpha \in \RR^p$ with $\alpha_j = \ind\{j\leq 10\}$ and $\beta\in \RR^p$ with $\beta_j = \ind\{j\leq 5\}$. Then, the four settings are specified as follows:
\vspace{0.5em}
\begin{itemize}
    \item Setting 1: $m_0(x)=\alpha^\top x$ and $e(x)=1/(1+\exp(-\beta^\top x))$. 
    \item Setting 2: $m_0(x) = \alpha^\top x$ and $e(x) =1/(1+\exp(-\beta^\top g(x)))$.
    \item Setting 3: $m_0(x) = \alpha^\top g(x)$ and $e(x) = 1/(1+\exp(-\beta^\top x))$. 
    \item Setting 4: $m_0(x) = \alpha^\top (x+0.2g(x))$, and $e(x) = 1/(1+\exp(-\beta^\top (x+0.2g(x)))$. 
\end{itemize}
\vspace{0.5em}
Since we use the Lasso to estimate $\hat\phi^{(k)}(\cdot)$, all four settings are misspecified for our consistency conditions (a well-specified propensity score model should be inverse linear, not logistic). Thus, the goal is mainly to test the robustness of cross balancing with respect to different degrees of misspecification. Among them, setting 2 tests misspecification with nonlinear propensity score, setting 3 tests misspecification with nonlinear prognostic score, while setting 4 tests slight misspecification in both prognostic and propensity scores. 

\subsection{Data generating processes in Section~\ref{subsubsec:simu_varsel_raw}}
\label{app:dgp_varsel_raw}

In Section~\ref{subsubsec:simu_varsel_raw}, we examine the performance of our method in Section~\ref{section_cross2} where balancing weights are obtained after selecting among observed covariates. 
To demonstrate the robustness of the procedure to imperfect variable selection, as well as the impact of the selection set to the estimation error, 
we design the following three data generating processes, fixing $p=100$. 
Like before, the outcomes and treatments are generated via $Y_i(0)=m_0(X_i)+\epsilon_i$ where $\epsilon_i\sim N(0,1)$ and $T_i\sim \textrm{Bern}(e(X_i))$ independently.
\vspace{0.5em}
\begin{itemize}
    \item Setting 1: This setting is inspired by~\cite{kang2007demystifying}. First, we generate $U\sim \RR^p$ where each entry is i.i.d.~from Unif$([0,1])$. We then generate $X\in \RR^p$ by $X_1=3(e^{U_1/2}-1)$, $X_2 = 2U_2/(1+e^{U_5})$, $X_3 = (U_2U_3/2+0.6)^3$, $X_4 = (U_2+U_4)^2$, and $X_{5:p}=U_{5:p}$. The prognostic score is $m_0(x,u)=1.74x_2+2x_3 + 0.1 \ind\{x_5>0.5\} + 0.2 (x_{8}-0.5)^2 + 0.1 u_1$, and the propensity score is $e(x) = 1/(1+\exp(0.5x_2+x_3+2x_5+0.1x_8+0.2(x_8-0.5)^2-1))$.  
    
    In this case, including observed covariates in $\hat{S}^{(k)}$ can only approximately recover the prognostic and propensity scores, since the prognostic score contains a component $0.1U_1$ and the propensity score is logistic and contains a quadratic term. This setup also poses challenges to outcome-only and propensity-only selection. If one only uses the outcome information, it is quite likely to miss $x_5$ which is important in $e(x)$; if one only focuses on treatment relationship, it is likely to miss $x_8$ which is important in $m_0(x)$. 
    \item Setting 2: We generate $X\in \RR^p$ using the same distribution as in settings 2-5 in Appendix~\ref{app:dgp_fitted_lowd}. Then, the prognostic score is $m_0(x)=x_1+x_5+0.1x_4+0.1x_6$, and the propensity score is $e(x) = 1/(1+\exp(1.5+(x_1+0.2x_2+0.2x_3+4x_4-0.1x_5+0.5x_6-6)/\sqrt{10}))$. 

    In this case, while the prognostic score is a linear function in $(x_1,x_4,x_5,x_6)$, the propensity score is logistic instead of inverse linear, providing a test for the robustness of function approximation conditions. In addition, outcome-only selection is likely to miss $x_4$ and $x_6$ which are important in $e(x)$, while treatment-only selection is likely to miss $x_5$ which is important in $m_0(x)$. 
    \item Setting 3: We generate $X\in \RR^p$ where each entry is i.i.d.~from $N(0,1)$. We then define $\alpha \in \RR^p$ with $\alpha_j=0.4$ for $1\leq j\leq 4$, and $\alpha_j=0.1(j-6)^{-2}$ for $j\geq 7$, as well as $\beta \in \RR^p$ by $\beta_j = \ind\{j\geq 5\} (j-4)^{-2}$. Both of them are designed to create a slight misspecification inspired by~\cite{belloni2014inference}. The prognostic score is given by $\mu_0(x)=(0.5 \alpha +0.1\beta)^\top x $, and the propensity score is given by $e(x) = 1/(1+\exp(-1+2\alpha^\top x + 0.2\beta^\top x^2))$. 

    In this setting, similar to~\cite{belloni2014inference}, the decaying coefficients for $j\geq 7$ makes it easy for variable selection algorithms to miss large indices, which then tests the robustness of our function approximation conditions.
\end{itemize}

\subsection{Data generating processes in Section~\ref{subsubsec:simu_varsel_basis}}
\label{app:dgp_varsel_basis}

Section~\ref{subsubsec:simu_varsel_basis} examines the performance of our method in Section~\ref{section_cross2} when the balancing features are selected from a dictionary of basis functions. We design the following three data generating processes adapted from~\cite{kang2007demystifying}. In all settings, after generating $U\in \RR^p$ for $p=4$ with each entry i.i.d.~from Unif$([0,1])$, the observed features are given by $X=g(U)\in \RR^p$, where $X_1=e^{U_1/2}$, $X_2 = U_1/(1+e^{U_1})$, $X_3 = (U_1U_3/25+0.6)^3$, and $X_4 = (U_2+U_4)^2$. The prognostic and propensity scores are linear in $U$, and their dependence on $X$ is intricate.

\begin{itemize}
    \item Setting 1: The prognostic score is $m_0(u)=210+27.4u_1+13.7u_2+13.7u_3+13.7u_4$, and the propensity score is $e(u) = 1/(1+\exp(1+u_1-0.5u_2+0.25u_3+0.1u_4)/2)$. 
    In this setting, the 
    \item Setting 2:  The prognostic score is $m_0(u) = 210 + 27.4 e^{u_1} + 13.7 u_2^2\ind\{u_2^2>0.2\} + 13.7 (u_2+u_3)^2+13.7 u_4 + 0.05 (u_2+u_3)\cdot u_2^2\ind\{u_2^2>0.6\}$, and the propensity score is $e(u)=1/(1+\exp(1+u_1-u_2+1.5u_3+0.1u_4))$. 
    \item Setting 3: The prognostic score is $m_0(u)=210+27.4u_1+13.7u_2+13.7u_3+13.7u_4$, and the propensity score is $e(x) = 1/(1+\exp(1+x_1-0.5x_2+0.25x_3-0.1x_4)/2)$.
    
\end{itemize}

In all three settings, the prognostic and propensity scores are complex nonlinear functions of the observed variables. 
As such, the linear space of the selected basis functions needs to accurately approximate them and fulfill the conditions in Section~\ref{section_cross2} for accurate estimation and reliable inference. 

\subsection{Additional results for variable selection}
\label{app:subsec_simu_varsel}

\begin{figure}
    \centering
    \begin{subfigure}{0.49\textwidth}
        \centering
        \includegraphics[width=\textwidth]{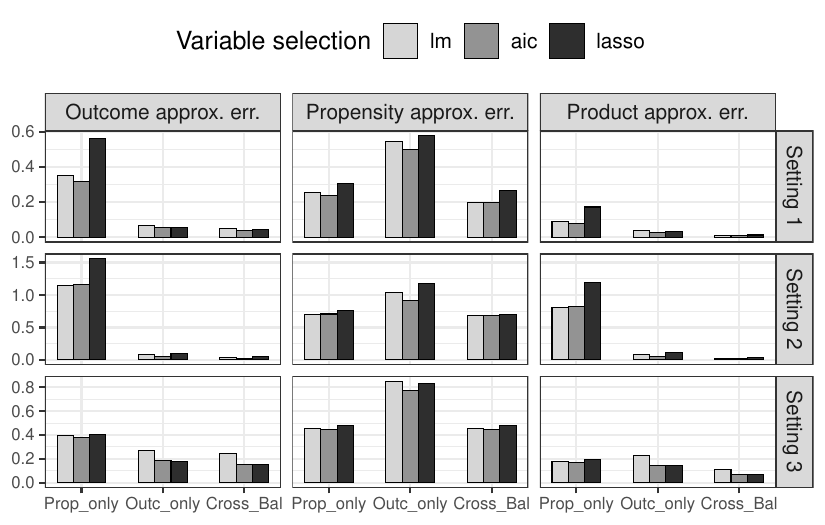} 
    \end{subfigure}
    \hfill
    \begin{subfigure}{0.49\textwidth}
        \centering
        \includegraphics[width=\textwidth]{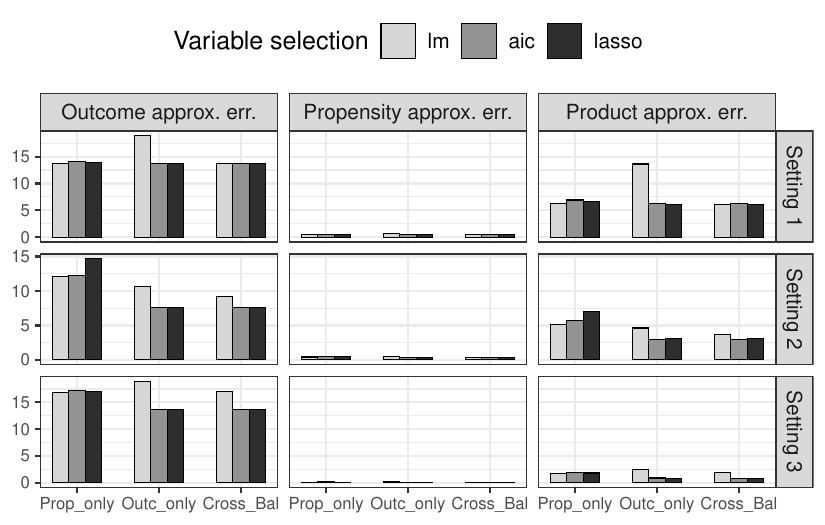} 
    \end{subfigure}
    \caption{
    {\small Approximation errors for selected variables in three estimators across all settings. The left panels display results for the settings in Section~\ref{subsubsec:simu_varsel_raw}, while the right panels correspond to those in Section~\ref{subsubsec:simu_varsel_basis}. In each column, outcome approx.~err.~is $\epsilon(m_0;\hat{S}^{(k)})$ averaged over two folds; propensity approx.~err.~is $\epsilon(w^*;\hat{S}^{(k)})$ averaged over two folds, and product approx.~err.~is their product.}}
    \label{fig:approx_err}
\end{figure}


\subsection{Additional results for case study}
\label{app:subsec_case_study}

\begin{figure}
    \centering
    \includegraphics[width=\linewidth]{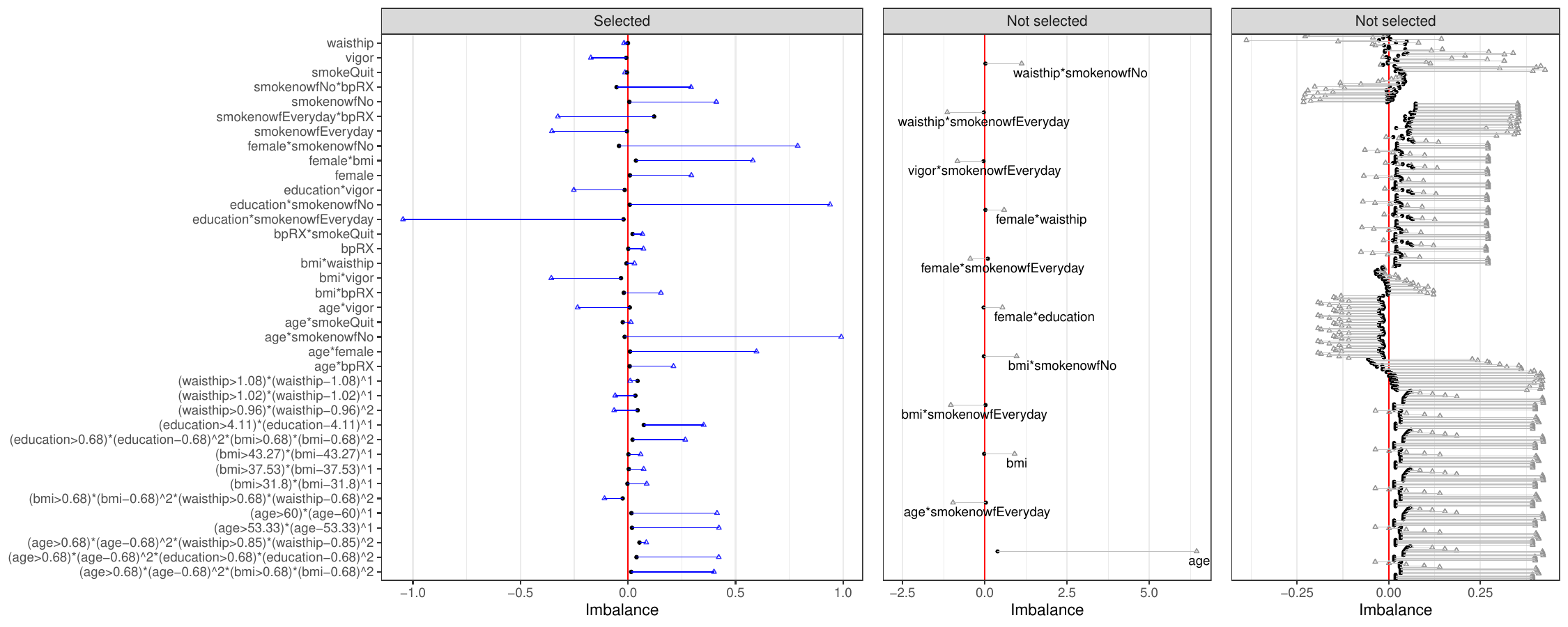}
    \caption{{\small Imbalance in the NHANES data before (triangles) and after (dots) re-weighting with cross balancing weights, where variables (basis functions being cubic splines with interaction terms) are selected via cross-validated Lasso. The left panel shows the imbalance measures for selected variables (the union of those selected in two folds) which are explicitly balanced in at least one fold. The middle and right panels show the imbalance measures for unselected variables which are not explicitly balanced, where the middle panel shows those with large before-balancing imbalance measures.}}
    \label{fig:real_imbalance_large}
\end{figure}

\section{Proof for cross-balancing with learned features}

\subsection{Proof of Theorem~\ref{thm:convergence}}
\label{app:proof_convergence}

\begin{proof}[Proof of Theorem~\ref{thm:convergence}]
    Without loss of generality, we incorporate the normalization constraint into the features $\hat\phi$ and balancing tolerance $\delta$, such that the first entry in $\hat\phi(X_i)$ is the intercept, and the first entry of $\delta$ equals zero. 
    We also fix the fold $k=1$ and the same conclusion for fold $k=2$ naturally applies. For notational simplicity, we then write $\hat\phi=\hat\phi^{(1)}$, and denote $\hat\Phi_c\in \RR^{|\cI_{c,1}|\times p}$ be the feature matrix in the control group, and $\hat\EE_t[\hat\phi] = \frac{1}{|\cI_{t,1}|}\sum_{j\in \cI_{t,1}}\hat\phi^{(1)}(X_j)$ be the sample mean of the features in the treated group. 
    
    By convex optimization theory, the Lagrangian of~\eqref{eq:opt_w2} is (writing $n = |\cI_{c,1}|$ with abuse of notations) 
    \$
    \cL(w,\lambda) = \|w\|_2^2 + \lambda_1^\top (\hat\Phi_c^\top w/n - \hat{\EE}_t[\hat\phi] -\delta) - \lambda_2^\top (\hat\Phi_c^\top w/n - \hat{\EE}_t[\hat\phi] +\delta_n),\quad \lambda_1,\lambda_2\succeq 0.
    \$
    Plugging in the value of $w$ that maximizes $\cL(w,\lambda)$,  the dual problem of~\eqref{eq:opt_w2} is 
    \$
    \cL_{\text{dual}}(\lambda) = - (\lambda_1-\lambda_2)^\top \hat\Phi_c^\top \hat\Phi_c (\lambda_1-\lambda_2)/(4n^2) - (\lambda_1-\lambda_2)^\top \hat\EE_t[\hat\phi] - (\lambda_1+\lambda_2)\delta_n.
    \$
    In addition, by duality, the optimal solution is $\hat{w} = -(2/n)\cdot\hat\Phi_c^\top (\hat\lambda_1-\hat\lambda_2) \in \RR^{n}$ where $(\hat\lambda_1,\hat\lambda_2)$ is the maximizer of $\cL_{\text{dual}}(\lambda)$. 
    Rewriting $\cL_{\text{dual}}(\lambda)$ in terms of $\lambda = \lambda_1-\lambda_2$ and $\eta = \lambda_1+\lambda_2$, and optimizing $\eta$ at $\eta=|\lambda|$, we yield the equivalent representation  $\hat{w}= \hat\Phi_c^\top \hat\lambda$, where $\hat\lambda $ minimizes the (unconstrained) convex objective  
    \$
    \hat{L}(\lambda) := \lambda^\top \hat\Sigma_c \lambda - 2\lambda^\top \hat\EE_t[\hat\phi] + 2|\lambda|^\top \delta_n, \quad \lambda \in \RR^p,
    \$
    and $\hat\Sigma_c = \hat\Phi_c^\top \hat\Phi_c/n$ is the sample covariance matrix. 

    We now consider the limiting objective 
    \$
    L(\lambda) = \lambda^\top \Sigma_c \lambda - 2\lambda^\top \EE_t[\phi^*],
    \$
    where $\EE_t[\phi^*]:= \EE[\phi^*(X)\given T=1]$, and $\Sigma_c := \EE[\phi^*(X)\phi^*(X)^\top\given T=0]$ is the limiting covariance matrix. 
    Let $\bar\lambda$ be the minimizer of $L(\lambda)$, which is $\bar\lambda = \Sigma_c^{-1} \EE_t[\phi^*]$.  
    Since $\nabla^2 \hat{L}(\lambda) = 2 \hat\Sigma_c$, $\|\hat\phi-\phi^*\|_{L_2} = o_P(1)$,  and $\phi^*$ is full rank, we know that for sufficiently large $n$, there exists some $\nu>0$ such that $\nabla^2 \hat{L}(\lambda) \succ \nu \cdot \mathbf{I}_{p\times p}$ for all $\lambda\in \RR^p$. Invoking Lemma~\ref{lem:convex} with $f=\hat{L}$, $\theta_0=\bar\lambda$ and $\theta  = \hat\lambda$, 
    we know that 
    \$
    \hat{L}(\hat\lambda)&\geq \hat{L}(\bar\lambda) + \nabla \hat{L}( \bar\lambda)^\top (\bar\lambda-\hat\lambda) + \nu/2\cdot \min \big\{ \|\hat\lambda-\bar\lambda\|^2, c\|\hat\lambda-\bar\lambda\| \big\} \\ 
    &\geq \hat{L}(\hat\lambda) + \nabla \hat{L}(\bar\lambda)^\top (\bar\lambda-\hat\lambda) + \nu/2\cdot \min \big\{ \|\hat\lambda-\bar\lambda\|^2, c\|\hat\lambda-\bar\lambda\| \big\}.
    \$
    since $\hat\lambda$ minimizes $\hat{L}$. By Cauchy-Schwarz inequality, we further have 
    \@\label{eq:bd_hat_lambda}
    \min \big\{ \|\hat\lambda-\bar\lambda\| , c \big\} \leq 2/\nu\cdot \|\nabla \hat{L}(\bar\lambda) \|.
    \@
    Using the fact that $\nabla L(\bar\lambda)=0$,  we have 
    \$
    \| \nabla \hat{L}(\bar\lambda) \| =  \| \nabla \hat{L}(\bar\lambda) - \nabla  {L}(\bar\lambda) \| 
    \leq \|  (\hat\Sigma_c-\Sigma_c) \bar\lambda \| + 2\|\hat\EE_t[\hat\phi] - \EE_t[\phi^*] \| \cdot  \|\bar\lambda\| + 2\delta_n \sqrt{p}.
    \$ 
    Since $\delta_n=o(1)$ and $\|\hat\phi-\phi^*\|_{L_2}=o_P(1)$, we know that $\|\hat\Sigma_c - \Sigma_c\|_2 = o_P(1)$ and $\|\hat\EE_t[\hat\phi] - \EE_t[\phi^*] \|= o_P(1)$. This further implies $\| \nabla \hat{L}(\bar\lambda) \| = o_P(1)$ and thus $\|\hat\lambda-\bar\lambda\|=o_P(1)$ due to~\eqref{eq:bd_hat_lambda}. We thus have 
    \$
    \big\| \hat{w} - (\Phi^*_c)^\top \bar\lambda\big\| = \big\| \hat\Phi_c^\top \hat\lambda - (\Phi^*_c)^\top \bar\lambda\big\| =   o_P(1),
    \$
    where $\Phi_c^*\in \RR^{|\cI_{c,1}|\times p}$  is the feature matrix whose $i$-th row is $\phi^*(X_i)$. Furthermore, letting $\bar{w}(X_i)=\phi^*(X_i)^\top \bar\lambda$,  the point estimate (for the current fold) is 
    \$
    \frac{1}{|\cI_{c,1}|}\sum_{i\in \cI_{c,1}} \hat{w}_i Y_i = o_P(1) + \frac{1}{|\cI_{c,1}|}\sum_{i\in \cI_{c,1}} \phi^*(X_i)^\top \bar\lambda Y_i(0) = \EE\big[ \bar{w}(X)Y(0)\given T=0  \big] + o_P(1)
    \$
    by the law of large numbers. 
    We thus complete the first part of Theorem~\ref{thm:convergence}.

    We now proceed to prove that $\bar\theta_0 = \theta_0$ under either of the given conditions. Recall that the limiting coefficient is $\bar\lambda = \Sigma_c^{-1}\EE_t[\phi^*]$, the minimizer of 
    \$
    L(\lambda) = \lambda^\top \Sigma_c \lambda - 2\lambda^\top \EE_t[\phi^*].
    \$
    Also, $\hat\theta=\bar\theta_0 = o_P(1)$, where $\bar\theta_0 = \EE[\phi^*(X)^\top\bar\lambda\cdot Y(0)\given T=0]$. 
    
    We first show the desired result under condition (i). In this case, since the true density ratio is $w^*(x)=\phi^*(x)^\top\lambda^*$, writing $\EE_c[\cdot]=\EE[\cdot\given T=0]$ and $\EE_t[\cdot]=\EE[\cdot\given T=1]$, we have 
    \$
    \bar\lambda &= \EE_c \big[\phi^*\phi^*(X)^\top \big]^{-1} \EE_t [\phi^*(X) ] = \EE_c \big[\phi^*\phi^*(X)^\top \big]^{-1} \EE_c [w(X)\phi^*(X) ]\\
    &= \EE_c \big[\phi^*\phi^*(X)^\top \big]^{-1} \EE_c [ \phi^*(X) \phi^*(X)^\top\lambda^*] = \lambda^*.
    \$
    As such, $\bar\theta_0 = \EE[\bar{w}(X)Y(0)\given T=0] = \EE[w^*(X)Y(0)\given T=0]=\theta_0$, which is the desired conclusion. 

    On the other hand, under condition (ii), we have 
    \$
    \bar\theta &=\EE[\bar{w}(X)Y(0)\given T=0] = \EE[\bar{w}(X)\phi^*(X)^\top \beta^*\given T=0] \\ 
    &= \EE[\bar\lambda^\top\phi^*(X)\phi^*(X)^\top \beta^*\given T=0] \\ 
    &= \bar\lambda^\top \Sigma_c \beta^* \\
    &= \EE_t[\phi^*(X)^\top\beta^*] = \EE[Y(0)\given T=1] = \theta_0,
    \$
    which concludes the proof of the second case. 
\end{proof}

\subsection{Proof of Proposition~\ref{prop:bias_reduc}}
\label{app:subsec_proof_bias}

\begin{proof}[Proof of Proposition~\ref{prop:bias_reduc}]
    For notational convenience, we consider the first half of cross-balancing with $k=1$, and we denote the  feature vector as $\hat\phi$, which can then be viewed as fixed. With $\delta=0$, using the duality analysis in the proof of Theorem~\ref{thm:convergence}, we see that  cross balancing returns 
\@\label{eq:form_hat_w}
\hat{w}_i = \hat\phi(X_i)^\top \hat\lambda , \quad \lambda = (\hat\Phi_c^\top \hat\Phi_c)^{-1} \hat\EE_t[\hat\phi]/n,
\@
where we recall that $\hat\Phi_c$ is the data matrix for $\hat\phi(X_i)$ in the control group, and $\bar\phi_t$ is the sample mean of $\hat\phi(X_j)$ for data in the treated group. This also means $\hat{w}_i$ is a linear combination of $\hat{m}(X_i)$ and $\hat{w}(X_i)$ together with other features in $\hat\phi(X_i)$ (note that $\hat{w}(X_i)$ is not the balancing weights $\hat{w}_i$).

We define the ``oracle'' AIPW estimator with true weights and regression functions:
\@\label{eq:oracle_aipw_est}
\hat\theta_{\aipw}^*:=
\frac{1}{|\cI_{c,1}|}\sum_{i\in \cI_{c,1}}  w^*(X_i)  \cdot \big\{ Y_i(0)-m_0(X_i) \big\} + \frac{1}{|\cI_{t,1}|}\sum_{j \in \cI_{t,1}} m_0(X_j) ,
\@
which is unbiased and most efficient. 
By the definition of the AIPW estimator, 
\$
\hat\theta_0^\aipw - \hat\theta_{\aipw}^* = &
\frac{1}{|\cI_{c,1}|}\sum_{i\in \cI_{c,1}} \big(\hat{w}(X_i)-w^*(X_i)\big) \big(m_0(X_i)-\hat{m}(X_i)\big) \\ 
&\quad + \underbrace{\frac{1}{|\cI_{c,1}|}\sum_{i\in \cI_{c,1}} (\hat{w}(X_i)-w^*(X_i))(Y_i(0)-m_0(X_i))}_{\text{term (ii)}} \\
&\quad + \underbrace{\frac{1}{|\cI_{t,1}|}\sum_{j \in \cI_{t,1}} (\hat{m}(X_j)-m_0(X_j)) -  \frac{1}{|\cI_{c,1}|}\sum_{i\in \cI_{c,1}} w^*(X_i) (\hat{m}(X_i)-m_0(X_i)) }_{\text{term (iii)}}.
\$
Note that $\EE[\text{term (ii)}] =0$ since $m_0(x)$ is the true regression function and $\hat{w}$ is obtained independent of the data it applies to. Since $w^*(X_i)$ is the true density ratio and  $\hat{m}$ is  independent of the data in $\cI_{c,k}$,  
\$
&\EE\Bigg[ \frac{1}{|\cI_{c,1}|}\sum_{i\in \cI_{c,1}} w^*(X_i) (\hat{m}(X_i)-m_0(X_i)) \Biggiven \cI_{c,2}\cup \cI_{t,2}\Bigg] \\ 
&= \EE[w^*(X)[\hat{m}-m_0](X)\given T=0] = \EE[ [\hat{m}-m_0](X)\given T=1] \\ 
&= \EE\Bigg[ \frac{1}{|\cI_{t,1}|}\sum_{j \in \cI_{t,1}} (\hat{m}(X_j)-m_0(X_j))  \Biggiven \cI_{c,2}\cup \cI_{t,2}\Bigg] ,
\$
where $(X,T)$ is an independent copy. Therefore, term (iii) is also mean zero. Finally, noting $\EE[\hat\theta_{\aipw}^*] = \theta_0$ yields the first equation in Proposition~\ref{prop:bias_reduc}.

On the other hand, by definition we have 
\$
\hat\theta_0 - \hat\theta_{\aipw}^* = &
\frac{1}{|\cI_{c,1}|}\sum_{i\in \cI_{c,1}} \hat{w}_i m_0 (X_i)   - \frac{1}{|\cI_{t,1}|}\sum_{j \in \cI_{t,1}}  m_0(X_j) \\
& + \underbrace{\frac{1}{|\cI_{c,1}|}\sum_{i\in \cI_{c,1}} \big(\hat{w}_i - w^*(X_i)\big) \big(Y_i(0) - m_0(X_i)\big)}_{\text{term (i)}} .
\$
Here, since $\hat{w}_i$ only depends on the covariates in $  \cI_{c,1}$, by the tower property we know $\EE[\text{term(i)}]=0$. 
For the remaining part in this decomposition, due to the balancing conditions, for any $\beta\in \RR^p$, it holds that 
\$
&\frac{1}{|\cI_{c,1}|}\sum_{i\in \cI_{c,1}} \hat{w}_i m_0(X_i)   - \frac{1}{|\cI_{t,1}|}\sum_{j \in \cI_{t,1}} m_0(X_j) \\ 
&= 
\frac{1}{|\cI_{c,1}|}\sum_{i\in \cI_{c,1}} \hat{w}_i \cdot \big\{ m_0(X_i) - \beta^\top \hat\phi(X_i) \big\}  - \frac{1}{|\cI_{t,1}|}\sum_{j \in \cI_{t,1}} \big\{ m_0(X_j) - \beta^\top \hat\phi(X_i)\big\}.
\$
Since $\hat\phi$ contains $\hat\mu$, equivalently, for any (potentially data-dependent) $\beta\in \RR^p$, it holds that 
\$
&\frac{1}{|\cI_{c,1}|}\sum_{i\in \cI_{c,1}} \hat{w}_im_0(X_i)   - \frac{1}{|\cI_{t,1}|}\sum_{j \in \cI_{t,1}} m_0(X_j) \\ 
&= 
\frac{1}{|\cI_{c,1}|}\sum_{i\in \cI_{c,1}} \hat{w}_i \cdot \big\{ m_0(X_i) -\hat\mu(X_i) - \beta^\top \hat\phi(X_i) \big\}  - \frac{1}{|\cI_{t,1}|}\sum_{j \in \cI_{t,1}} \big\{ m_0(X_j) -\hat\mu(X_j)- \beta^\top \hat\phi(X_j)\big\}.
\$
We take $\beta$ as the sample OLS projection coefficient of $m_0(X_i)-\hat\mu(X_i)$ onto the space of $\hat\phi(X_i)$, 
and write 
\$
\Pi^{(k)}_{\hat\phi}[\hat{m}^{(k)}   - m_0  ] (x) = \hat{m}^{(k)}(x) - m_0(x)-\beta^\top \hat\phi(x).
\$
Since $\beta$ is the projection coefficient, we have (we now use the superscript $k$ to emphasize the data fold)
\$
\frac{1}{|\cI_{c,1}|}\sum_{i\in \cI_{c,1}} \hat\phi(X_i) \cdot \Pi^{(k)}_{\hat\phi}[\hat{m}^{(k)}  - m_0  ] (X_i)  = 0.
\$
Thus, for any $\eta\in \RR^p$, it holds that 
\$
&\frac{1}{|\cI_{c,1}|}\sum_{i\in \cI_{c,1}}  \hat{w}_i m_0(X_i)   - \frac{1}{|\cI_{t,1}|}\sum_{j \in \cI_{t,1}} m_0(X_j) \\
&= 
\frac{1}{|\cI_{c,1}|}\sum_{i\in \cI_{c,1}}  \big\{ \hat{w}_i - \eta^\top \hat\phi(X_i)\big\}  \cdot \Pi^{(k)}_{\hat\phi}[\hat{m}^{(k)}  - m_0  ] (X_i) - \frac{1}{|\cI_{t,1}|}\sum_{j \in \cI_{t,1}} \Pi^{(k)}_{\hat\phi}[\hat{m}^{(k)}  - m_0  ] (X_j) \\ 
&= \frac{1}{|\cI_{c,1}|}\sum_{i\in \cI_{c,1}} \big\{ \hat{w}_i - w^*(X_i) - \eta^\top \hat\phi(X_i)\big\}  \cdot \Pi^{(k)}_{\hat\phi}[\hat{m}^{(k)}  - m_0  ] (X_i)  \\ 
&\quad + \underbrace{\frac{1}{|\cI_{c,1}|}\sum_{i\in \cI_{c,1}}    w^*(X_i)   \cdot \Pi^{(k)}_{\hat\phi}[\hat{m}^{(k)}  - m_0  ] (X_i) - \frac{1}{|\cI_{t,1}|}\sum_{j \in \cI_{t,1}} \Pi^{(k)}_{\hat\phi}[\hat{m}^{(k)}  - m_0  ] (X_j)}_{\text{term (iv)}}.
\$
Note that term (iv) has mean zero since $w^*(X_i)$ is the true density ratio between a control covariate $X_i$ and treated covariate $X_j$. 
In addition, by~\eqref{eq:form_hat_w},   $\hat{w}_i$ is a linear combination of $\hat{w}(X_i)$ and $\hat\mu(X_i)$, and thus it lives in the linear space spanned by $\hat\phi(X_i)$. Consequently, for any $\eta\in \RR^p$, there exists some $\tilde\eta\in \RR^p$ such that 
\$
&\frac{1}{|\cI_{c,1}|}\sum_{i\in \cI_{c,1}} \big\{ \hat{w}_i - w^*(X_i) - \eta^\top \hat\phi(X_i)\big\}  \cdot \Pi^{(k)}_{\hat\phi}[\hat{m}^{(k)}  - m_0  ] (X_i) \\
&= \frac{1}{|\cI_{c,1}|}\sum_{i\in \cI_{c,1}} \big\{ \hat{w}(X_i) - w^*(X_i) - \tilde\eta^\top \hat\phi(X_i)\big\}  \cdot \Pi^{(k)}_{\hat\phi}[\hat{m}^{(k)}  - m_0  ] (X_i) ,
\$ 
and vice versa (i.e., the mapping between $\eta$ and $\tilde\eta$ is a bijection). 
We now set $\eta$ as the sample OLS coefficient that projects $\hat{w}(X_i)-w^*(X_i)$ onto the space of $\hat\phi(X_i)$. This leads to 
\$
&\frac{1}{|\cI_{c,1}|}\sum_{i\in \cI_{c,1}}  \hat{w}_i m_0(X_i)   - \frac{1}{|\cI_{t,1}|}\sum_{j \in \cI_{t,1}} m_0(X_j) \\
&= \frac{1}{|\cI_{c,1}|}\sum_{i\in \cI_{c,1}} \big\{ \hat{w}(X_i) - w^*(X_i) - \tilde\eta^\top \hat\phi(X_i)\big\}  \cdot \Pi^{(k)}_{\hat\phi}[\hat{m}^{(k)}  - m_0  ] (X_i)  + \text{term(iv)} \\ 
&= \frac{1}{|\cI_{c,1}|}\sum_{i\in \cI_{c,1}} \Pi^{(k)}_{\hat\phi}[\hat{w}^{(k)}  - w^*  ] (X_i)  \cdot \Pi^{(k)}_{\hat\phi}[\hat{m}^{(k)}  - m_0  ] (X_i)  + \text{term(iv)}.
\$
The above arguments lead to 
\$
\EE\big[\hat\theta_0 +\hat\theta_\aipw^*\big] = 
\EE\Bigg[\frac{1}{|\cI_{c,1}|}\sum_{i\in \cI_{c,1}} \Pi^{(k)}_{\hat\phi}[\hat{w}^{(k)}  - w^*  ] (X_i)  \cdot \Pi^{(k)}_{\hat\phi}[\hat{m}^{(k)}  - m_0  ] (X_i) \Bigg],
\$
since $\EE[\text{term (i)}] =\EE[\text{term (iv)}] = 0$. 
That is, cross-balancing projects out the component in the space of $\hat\phi$ from the biases in $\hat{w}$ and $\hat{\mu}$, completing the proof for the second conclusion in Proposition~\ref{prop:bias_reduc}. 
\end{proof}

\subsection{Proof of Theorem~\ref{thm:inference}}
\label{app:subsec_proof_inference}

\begin{proof}[Proof of Theorem~\ref{thm:inference}]
    We first fix a fold $k=1$, and similar results naturally apply to the other fold. 
    For notational simplicity, we sometimes write $\hat\phi$ instead of $\hat\phi^{(k)}$. 
    
    Following the convex analysis in the proof of Theorem~\ref{thm:convergence}, we know 
     $\hat{w}= \hat\Phi_c^\top \hat\lambda$, where $\hat\Phi_c$ is the data matrix for $\hat\phi(X_i)$, $i\in \cI_{c,1}$, and  $\hat\lambda $ minimizes the (unconstrained) convex objective  
    \$
    \hat{L}(\lambda) := \lambda^\top \hat\Sigma_c \lambda - 2\lambda^\top \hat\EE_t[\hat\phi] + 2|\lambda|^\top \delta, \quad \lambda \in \RR^p,
    \$
    and $\hat\Sigma_c = \hat\Phi_c^\top \hat\Phi_c/n$ is the sample covariance matrix.  
    Consider the limiting objective 
    \$
    L(\lambda) = \lambda^\top \Sigma_c \lambda - 2\lambda^\top \EE_t[\phi^*],
    \$
    where $\EE_t[\phi^*]:= \EE[\phi^*(X)\given T=1]$, and $\Sigma_c := \EE[\phi^*(X)\phi^*(X)^\top\given T=0]$ is the limiting covariance matrix. 
    Let $\bar\lambda$ be the minimizer of $L(\lambda)$, which is $\bar\lambda = \Sigma_c^{-1} \EE_t[\phi^*]$.  
    Similar to the proof of Theorem~\ref{thm:convergence}, for sufficiently large $n$, there exists some $\nu>0$ such that $\nabla^2 \hat{L}(\lambda) \succ \nu \cdot \mathbf{I}_{p\times p}$ for all $\lambda\in \RR^p$. Invoking Lemma~\ref{lem:convex} with $f=\hat{L}$, $\theta_0=\bar\lambda$ and $\theta  = \hat\lambda$, 
    we know that 
    \$
    \hat{L}(\hat\lambda)&\geq \hat{L}(\bar\lambda) + \nabla \hat{L}( \bar\lambda)^\top (\bar\lambda-\hat\lambda) + \nu/2\cdot \min \big\{ \|\hat\lambda-\bar\lambda\|^2, c\|\hat\lambda-\bar\lambda\| \big\} \\ 
    &\geq \hat{L}(\hat\lambda) + \nabla \hat{L}(\bar\lambda)^\top (\bar\lambda-\hat\lambda) + \nu/2\cdot \min \big\{ \|\hat\lambda-\bar\lambda\|^2, c\|\hat\lambda-\bar\lambda\| \big\}.
    \$
    since $\hat\lambda$ minimizes $\hat{L}$. By Cauchy-Schwarz inequality, we further have 
    \@\label{eq:bd_hat_lambda}
    \min \big\{ \|\hat\lambda-\bar\lambda\| , c \big\} \leq \|\nabla \hat{L}(\bar\lambda) \|.
    \@
    Using the fact that $\nabla L(\lambda^*)=0$,  we have 
    \$
    \| \nabla \hat{L}(\bar\lambda) \| =  \| \nabla \hat{L}(\bar\lambda) - \nabla  {L}(\bar\lambda) \| 
    \leq \|  (\hat\Sigma_c-\Sigma_c) \bar\lambda \| + 2\|\hat\EE_t[\hat\phi] - \EE_t[\phi^*] \| \cdot  \|\bar\lambda\| + 2\delta.
    \$ 
    Since $\delta=o(n^{-1/4})$ and $\|\hat\phi-\phi^*\|_{L_2}=o_P(n^{-1/4})$, we know that $\|\hat\Sigma_c - \Sigma_c\|_2 = o_P(n^{-1/4})$ and $\|\hat\EE_t[\hat\phi] - \EE_t[\phi^*] \|= o_P(n^{-1/4})$. This further implies $\| \nabla \hat{L}(\bar\lambda) \| = o_P(n^{-1/4})$ and thus $\|\hat\lambda-\bar\lambda\|=o_P(n^{-1/4})$ due to~\eqref{eq:bd_hat_lambda}.  Finally, similar to the proof of Theorem~\ref{thm:convergence}, under the given conditions we have $\bar\lambda=\lambda^*$.

    We now proceed to analyze the estimator $\hat\theta_0$ related to the first fold, i.e., 
    \$
    \hat\theta_0^{(k)} := \frac{1}{|\cI_{c,k}|}\sum_{i\in \cI_{c,k}}\hat{w}_i Y_i(0).
    \$
    We will show that under the given conditions, it is asymptotically equivalent to the oracle AIPW estimator 
    \$
    \hat\theta_{\aipw}^{(k)}:= \frac{1}{|\cI_{c,k}|}\sum_{i\in \cI_{c,k}} w^*(X_i)(Y_i - m_0(X_i)) + \frac{1}{|\cI_{t,k}|}\sum_{j\in \cI_{t,k}} m_0(X_j).
    \$
    Define the difference after balancing as 
    \$
    \Delta^{(k)} = \frac{1}{|\cI_{c,1}|} \sum_{i\in \cI_{c,1}} \hat{w}_i \hat\phi(X_i) - \frac{1}{|\cI_{t,1}|} \sum_{i\in \cI_{t,1}} \hat\phi(X_i) \in \RR^p.
    \$
    The balancing condition implies $|\Delta^{(k)}|\leq \delta$ where $|\cdot|$ denotes entry-wise absolute value for a vector. As such, 
    \@\label{eq:est_decomp}
    \hat\theta_0^{(k)} = \frac{1}{|\cI_{c,1}|}\sum_{i\in \cI_{c,1}}\hat{w}_i \cdot (Y_i(0) -\hat\phi(X_i)^\top \beta^*)  + \frac{1}{|\cI_{t,1}|} \sum_{i\in \cI_{t,1}} \hat\phi(X_i)^\top \beta^* + (\beta^*)^\top \Delta^{(k)}.
    \@
    Using~\eqref{eq:est_decomp} and reorganizing terms, we thus obtain 
    \$
    \hat\theta_0^{(k)} - \hat\theta_{\aipw}^{(k)}
    = & (\beta^*)^\top \Delta^{(k)} +  \underbrace{
    \frac{1}{|\cI_{c,1}|}\sum_{i\in \cI_{c,1}} (\hat{w}_i-w^*(X_i))(Y_i-m_0(X_i))
    }_{\text{term (i)}} \\ 
    & + \underbrace{
    \frac{1}{|\cI_{c,1}|}\sum_{i\in \cI_{c,1}} (\hat{w}_i - w^*(X_i))(m_0(X_i) - \hat\phi(X_i)^\top \beta^*)
    }_{\text{term (ii)}} \\ 
    & + \underbrace{
    \frac{1}{|\cI_{c,1}|}\sum_{i\in \cI_{c,1}} w^*(X_i)(m_0(X_i) - \hat\phi(X_i)^\top \beta^*)  - \frac{1}{|\cI_{t,1}|}\sum_{j\in \cI_{t,1}} (m_0(X_j) - \hat\phi(X_j)^\top \beta^*)
    }_{\text{term (iii)}}.
    \$
    First of all, define $\cD_{1,X}:=\{X_i,T_i,Y_i\}_{i\in \cI_{c,2}} \cup \{X_i,T_i\}_{i\in \cI_{t,2}} \cup \{X_i,T_i\}_{i\in \cI_{t,1}\cup \cI_{t,2}}$, which contains data in  $\cI_{c,2}\cup\cI_{t,2}$ and the covariates in $\cI_{c,1}\cup \cI_{t,1}$. 
    By construction, $\hat{w}_i$ is measurable with respect to $\cD_{1,X}$, and conditional on $\cD_{1,X}$, every $Y_i$ is conditionally independent across $i \in \cI_{c,1}$ with mean $m_0(X_i)$.  Therefore, 
    \$
    \EE\Big[|\text{term (i)}|^2 \Biggiven \cD_{1,X} \Big] 
    = \frac{1}{|\cI_{c,1}|^2}\sum_{i\in \cI_{c,1}} \EE\Big[ (\hat{w}_i-w^*(X_i))^2(Y_i-m_0(X_i))^2 \Biggiven \cD_{1,X} \Big] .
    \$
    By the tower property and the bounded conditional variance condition, we know 
    \$
    \EE\Big[|\text{term (i)}|^2   \Big] 
    \leq M_2^2\cdot \EE\Bigg[ \frac{1}{|\cI_{c,1}|^2}\sum_{i\in \cI_{c,1}}  (\hat{w}_i-w^*(X_i))^2   \Bigg] .
    \$
    By the definition of $\hat{w}_i$ and the convergence of $\hat\phi$ to $\phi^*$ and $\hat\lambda$ to $\lambda^*$, we know that $\EE [|\text{term (i)}|^2] = o(1/n)$. Therefore, by Markov's inequality, this further implies $\text{term (i)} = o_P(n^{-1/2})$. 

    By H\"older's inequality, 
    \$
    \big|  \text{term (ii)} \big| 
    &\leq \sqrt{\frac{1}{|\cI_{c,1}|}\sum_{i\in \cI_{c,1}} (\hat{w}_i - w^*(X_i))^2} \cdot \sqrt{\frac{1}{|\cI_{c,1}|}\sum_{i\in \cI_{c,1}} 
 (m_0(X_i) - \hat\phi(X_i)^\top \beta^*)^2} \\
 & = \sqrt{\frac{1}{|\cI_{c,1}|}\sum_{i\in \cI_{c,1}} \big[ \hat\phi(X_i)^\top\hat\lambda - \phi^*(X_i)^\top \lambda^*\big]^2} \cdot \sqrt{\frac{1}{|\cI_{c,1}|}\sum_{i\in \cI_{c,1}} 
  \big[  \phi^*(X_i)^\top \beta^* - \hat\phi(X_i)^\top \beta^*\big]^2} .
    \$
    By Cauchy-Schwarz inequality, 
    \$
    \frac{1}{|\cI_{c,1}|}\sum_{i\in \cI_{c,1}} \big[ \hat\phi(X_i)^\top\hat\lambda - \phi^*(X_i)^\top \lambda^*\big]^2 \leq \frac{2}{|\cI_{c,1}|}\sum_{i\in \cI_{c,1}} \Big\{ \big[ \hat\phi(X_i)^\top\hat\lambda - \phi^*(X_i)^\top \hat\lambda \big]^2 +  \big[  \phi^*(X_i)^\top\hat\lambda - \phi^*(X_i)^\top  \lambda^* \big]^2 \Big\},
    \$
    where due to the $o_P(n^{-1/4})$ convergence rates of $\hat\phi$ to $\phi^*$, 
    \$
    \frac{2}{|\cI_{c,1}|}\sum_{i\in \cI_{c,1}}  \big[ \hat\phi(X_i)^\top\hat\lambda - \phi^*(X_i)^\top \hat\lambda \big]^2 &\leq \frac{2}{|\cI_{c,1}|}\sum_{i\in \cI_{c,1}}  \big\| \hat\phi(X_i)- \phi^*(X_i) \big\|^2 \| \hat\lambda \|^2\\
    &= \|\hat\lambda\|^2 \cdot O_P( \|\hat\phi-\phi^*\|_{L_2}^2) = o_P(n^{-1/2}).
    \$
    On the other hand, due to the $o_P(n^{-1/4})$ convergence rates of $\hat\lambda$ to $\lambda^*$,
    \$
    \frac{2}{|\cI_{c,1}|}\sum_{i\in \cI_{c,1}}  \big[  \phi^*(X_i)^\top\hat\lambda - \phi^*(X_i)^\top  \lambda^* \big]^2 
    \leq \|\hat\lambda-\lambda^*\|^2 \cdot O_P(\EE[\|\phi^*(X)\|^2\given T=0]) = o_P(n^{-1/2}).
    \$
    Similarly, we can show 
    \$
    \frac{1}{|\cI_{c,1}|}\sum_{i\in \cI_{c,1}} 
  \big[  \phi^*(X_i)^\top \beta^* - \hat\phi(X_i)^\top \beta^*\big]^2 = o_P(n^{-1/2}).
    \$
    Putting these bounds together, we know $|\text{term (ii)}| = o_P(n^{-1/2})$. 

    Since $w^*(X_i)$ is the true density ratio and $\hat\phi(\cdot)$ is obtained from another fold, we have 
    \$
    \EE\Big[ w^*(X_i)(m_0(X_i) - \hat\phi(X_i)^\top \beta^*) \Biggiven T_i=0, \cI_{c,2}\cup\cI_{t,2}\Big] = \EE\Big[   m_0(X_j) - \hat\phi(X_j)^\top \beta^*  \Biggiven T_j=0, \cI_{c,2}\cup\cI_{t,2}\Big] =: \mu_0
    \$
    for $i\in \cI_{c,1}$ and $j\in \cI_{t,1}$. Therefore, 
    \@\label{eq:term3}
    \text{term (iii)} = \frac{1}{|\cI_{c,1}|}\sum_{i\in \cI_{c,1}} \big[ w^*(X_i)(m_0(X_i) - \hat\phi(X_i)^\top \beta^*) - \mu_0\big]  - \frac{1}{|\cI_{t,1}|}\sum_{j\in \cI_{t,1}} \big[ m_0(X_j) - \hat\phi(X_j)^\top \beta^* - \mu_0\big].
    \@
    Conditional on full data in $\cI_{c,2}\cup\cI_{c,2}$ (and $\{T_i\}_{i\in \cI_{c,1}\cup\cI_{t,1}}$ to be rigorous), each item in the summation is conditionally independent and mean zero. This means 
    \$
    &\EE\Big[|\text{term (iii)}|^2\Biggiven \cI_{c,2}\cup\cI_{c,2}\cup \{T_i\}_{i\in \cI_{c,1}\cup\cI_{t,1}}\Big] \\
    &= \frac{1}{|\cI_{c,1}|} \EE\Big[ (w^*(X_i)(m_0(X_i) - \hat\phi(X_i)^\top \beta^*) - \mu_0)^2\Biggiven \cI_{c,2}\cup\cI_{c,2}\cup \{T_i\}_{i\in \cI_{c,1}\cup\cI_{t,1}}\Big] \\
    &\quad + \frac{1}{|\cI_{t,1}|}\sum_{j\in \cI_{t,1}} \EE\Big[ (m_0(X_j) - \hat\phi(X_j)^\top \beta^* - \mu_0)^2 \Biggiven \cI_{c,2}\cup\cI_{c,2}\cup \{T_i\}_{i\in \cI_{c,1}\cup\cI_{t,1}}\Big].
    \$
    It's straightforward to see that each conditional expectation in the summation  is $O_P(\|\hat\phi-\phi^*\|_{L_2}^2) = o_P(n^{-1/2})$, hence by the Markov's inequality, for any fixed $\varepsilon>0$, 
    \$
    \PP\Big(  |\text{term (iii)}|>\varepsilon \cdot \sqrt{n} \Biggiven \cI_{c,2}\cup\cI_{c,2}\cup \{T_i\}_{i\in \cI_{c,1}\cup\cI_{t,1}} \Big) \leq \frac{ o_P(1)}{\varepsilon^2} = o_P(1).
    \$
    Taking expectation over the random variable on the LHS of the above equation and applying the dominated convergence theorem (with any subsequence of $n$) then yields $|\text{term (iii)}|=o_P(n^{-1/2})$. 

    Putting the bounds on all three terms together, we have 
    \$
\hat\theta_0^{(k)} - \hat\theta_{\aipw}^{(k)} = o_P(n^{-1/2}) +(\beta^*)^\top \Delta^{(k)}  = o_P(n^{-1/2}) 
    \$
    since $|\Delta^{(k)}|\leq \delta$ and $\delta = o_P(n^{-1/2})$. Using the same result for $k=2$, we then have 
    \$
    \hat\theta_0 &= \frac{1}{n}\sum_{i=1}^n \frac{1-T_i}{n_c/n} w^*(X_i)[Y_i-m_0(X_i)] + \frac{T_i}{n_t/n} m_0(X_i) + o_P(n^{-1/2}) \\ 
    &= \frac{1}{n}\sum_{i=1}^n \frac{1-T_i}{1-p} w^*(X_i)[Y_i-m_0(X_i)] + \frac{T_i}{p} m_0(X_i) + o_P(n^{-1/2}),
    \$
    which leads to the desired result by the central limit theorem. 

\paragraph{Variance estimation.} Finally, we can estimate the asymptotic variance $\sigma_0^2$ via 
\$
\hat\sigma^2 = \frac{1}{n_0(1-\hat{p})} \sum_{T_i=0} \big[\hat{w}_i\cdot (Y_i-\hat\beta^\top \hat\phi(X_i))\big]^2 + \frac{1}{n_1\hat{p}}\sum_{T_i=1} (\hat\beta^\top \hat\phi(X_i) - \hat\beta^\top \bar\phi)^2,
\$
where $\hat{w}_i$ is the cross-balancing weights. For $i \in \cI_{c,k}\cup\cI_{t,k}$, we write  $\hat\phi(X_i)=\hat\phi^{(k)}(X_i)$, and $\bar\phi = \bar\phi^{(k)} = \frac{1}{|\cI_{t,k}|}\sum_{i\in\cI_{t,k}}\hat\phi^{(k)}(X_i)$, and $\hat\beta$ is the OLS coefficient of $\{Y_i\}_{i\in \cI_{c,k}}$ over $\{\hat\phi^{(k)}(X_i)\}_{i\in \cI_{c,k}}$.

We now show that $\hat\sigma^2 = \sigma_0^2+o_P(1)$. Under the given conditions, following the same arguments as before, we have $\|\hat{w} - w^*(\cdot)\|_2 = o_P(1)$. Since $m_0(x) = (\beta^*)^\top \phi^*(x)$, the covariance matrix of $\phi^*(\cdot)$ is full-rank, and $\|\hat\phi^{(k)}-\phi^*(\cdot)\|_{L_2} = o_P(1)$, it is straightforward to show that $\|\hat\beta^{(k)} -\beta^*\| = o_P(1)$ for $k=1,2$. This means 
\$
\hat\sigma^2 = \frac{1}{n_0(1-\hat{p})} \sum_{T_i=0} \big[{w}^* (X_i)(Y_i-m_0(X_i))\big]^2 + \frac{1}{n_1\hat{p}}\sum_{T_i=1} (m_0(X_i) - \EE[m_0(X)\given T=1])^2 + o_P(1) ,
\$
and thus $\hat\sigma^2 = \sigma_0^2+o_P(1)$,
which concludes the proof of Theorem~\ref{thm:inference}.
\end{proof}

\subsection{Proof of Theorem~\ref{thm:multi_param}}
\label{app:subsec_proof_multi_param}

\begin{proof}[Proof of Theorem~\ref{thm:multi_param}]
We first fix a fold $k=1$, and similar results naturally apply to the other fold. 
    For notational simplicity, we sometimes write $\hat\phi$ in the place of $\hat\phi^{(k)}$,   $\hat{m}_j(x)$ in the place of  $\hat{m}_{j,\hat{\theta}_j^{(k)}}(x)$, and $\hat{w}_\ell(x)$ in the place of $ e_{\ell,\hat\eta_\ell^{(k)}}(x) /[1-e_{\ell,\hat\eta_\ell^{(k)}}(x)]$. 
    Thus, the fitted features are written as 
    \$
    \hat\phi(x) = \big(\hat{m}_1(x),\dots,\hat{m}_J(x),\hat{w}_1(x),\dots,\hat{w}_L(x)\big). 
    \$
    We write $n_{c,1} = |\cI_{c,1}|$ and $n_{t,1}=|\cI_{t,1}|$. We assume $\theta_j$'s and $\eta_j$'s are of the same dimension $d_j=d_\ell \equiv d$ for some $d\in \NN^+$, for notational simplicity. 
    We also define the limiting features 
    \$
    \phi^*(x) = \big(m_{1,\bar\theta_1}(X),\dots,m_{J,\bar\theta_J}(X),  w_{1,\eta_1}(X) ,\dots, w_{L,\eta_L}(X) \big),
    \$
    where $w_{\ell,\eta_\ell}(x) = e_{\ell,\eta_\ell}(x)/[1-e_{\ell,\eta_\ell}(x)]$. 
    The given conditions imply $\|\hat\phi(\cdot)-\phi^*(\cdot)\|_{L_2} = O_P(n^{-1/2})$.  Precisely, denoting $\nabla m_{j,\theta}(x) = \frac{\partial m_{j,\theta_j}}{\partial \theta_j}(x) \mid_{\theta_j=\theta} \in \RR^d$ and similarly $\nabla w_{\ell,\eta}(x) \in \RR^d$, we have 
    \$
    \hat{m}_j(x) - m_{j,\bar\theta_j}(x) = \nabla m_{j,\bar\theta_j}(x) (\hat\theta_j^{(k)} - \bar\theta_j) + O_P(1/n) =   \frac{1}{n}\sum_{i\in \cI_{c,3-k}\cup\cI_{t,3-k}} f_j(D_i) \nabla m_{j,\bar\theta_j}(x)   + O_P(1/n).
    \$
    For notational convenience, we denote $\hat\EE^\dag[f_j] = \frac{1}{n}\sum_{i\in \cI_{c,3-k}\cup\cI_{t,3-k}} f_j(D_i)$, whose population mean is zero,  and $\dot{m}_j(x) = \nabla m_{j,\bar{\theta_j}}(x) \in \RR^d$, and similarly for these components in the weight models. We then have 
    \$
    \phi^\Delta(x)&:= \hat\phi(x) - \phi^*(x) \\ 
    &= \big(\hat{\EE}^\dag[f_1]^\top \dot{m}_1(x),\dots, \hat{\EE}^\dag[f_J]^\top \dot{m}_J(x), \hat{\EE}^\dag[g_1]^\top \dot{w}_1(x),\dots, \hat{\EE}^\dag[g_L]^\top \dot{w}_L(x) \big)^\top\in \RR^{J+L}.
    \$
    Here the randomness in $\hat\EE^\dag[\cdot]$ from the model-fitting data $\cI_{c,3-k}\cup\cI_{t,3-k}$, while $\nabla \phi^*$ is a fixed mapping. 
    
    Following the convex analysis in the proof of Theorem~\ref{thm:convergence}, we know the cross-balancing weights are 
     $\hat{w}= \hat\Phi_c^\top \hat\lambda$, where $\hat\Phi_c$ is the data matrix for $\hat\phi(X_i)$, $i\in \cI_{c,1}$, and  $\hat\lambda $ minimizes the  convex objective  
    \$
    \hat{L}(\lambda) := \lambda^\top \hat\Sigma_c \lambda - 2\lambda^\top \hat\EE_t[\hat\phi] + 2|\lambda|^\top \delta, \quad \lambda \in \RR^{J+L},
    \$ 
    and $\hat\Sigma_c = \hat\Phi_c^\top \hat\Phi_c/n_{c,1}$ is the sample covariance matrix.  
    Consider the non-penalized objective 
    \$
    \hat{L}^*(\lambda) := \lambda^\top \hat\Sigma_c \lambda - 2\lambda^\top \hat\EE_t[\hat\phi] , \quad \lambda \in \RR^{J+L},
    \$
    and denote its minimizer as $\hat\lambda^*$. 
    Similar to the proof of Theorem~\ref{thm:convergence}, for sufficiently large $n$, there exists some $\nu>0$ such that $\nabla^2 \hat{L}(\lambda) \succ \nu \cdot \mathbf{I}_{ (J+L)\times (J+L)}$ for all $\lambda\in \RR^{J+L}$ due to condition (ii). Invoking Lemma~\ref{lem:convex} with $f=\hat{L}$, $\theta_0=\hat\lambda^*$ and $\theta  = \hat\lambda$, 
    we know that 
    \$
    \hat{L}(\hat\lambda)&\geq \hat{L}(\hat\lambda^*) + \nabla \hat{L}( \hat\lambda^*)^\top (\hat\lambda^*-\hat\lambda) + \nu/2\cdot \min \big\{ \|\hat\lambda-\hat\lambda^*\|^2, c\|\hat\lambda-\hat\lambda^*\| \big\} \\ 
    &\geq \hat{L}(\hat\lambda) + \nabla \hat{L}(\hat\lambda^*)^\top (\hat\lambda^*-\hat\lambda) + \nu/2\cdot \min \big\{ \|\hat\lambda-\hat\lambda^*\|^2, c\|\hat\lambda-\hat\lambda^*\| \big\}.
    \$
    This implies 
    \@\label{eq:bd_hat_lambda_prm}
    \min \big\{ \|\hat\lambda-\hat\lambda^*\| , c \big\} \leq \|\nabla \hat{L} (\hat\lambda^*) \|.
    \@ 
    Using the fact that $\nabla L(\hat\lambda^*)=0$,  we have 
    \$
    \| \nabla \hat{L}(\hat\lambda^*) \| =  \| \nabla \hat{L}(\hat\lambda^*) - \nabla  {L}(\hat\lambda^*) \| 
    \leq  2\delta = o_P(n^{-1/2}).
    \$ 
    This implies the following expression for the cross-balancing weights:
    \@\label{eq:hat_w_prm}
    \hat{w} = \hat\Phi_c^\top \hat\lambda^* + o_P(n^{-1/2}),\quad  
    \hat\lambda^* = \hat\Sigma_c^{-1} \hat\EE_t[\hat\phi],
    \@
    where the expression of $\hat\lambda^*$ follows from its definition. Note that $\hat\Sigma_c$ and $\hat\EE_t[\cdot]$ is averaged over the current fold, while $\hat\phi$ is obtained using data from the other fold.

    We now proceed to analyze the estimator $\hat\theta_0$ related to the first fold. Based on the arguments above, 
    \$
    \hat\theta_0^{(k)} &:= \frac{1}{|\cI_{c,k}|}\sum_{i\in \cI_{c,k}}\hat{w}_i Y_i = \frac{1}{|\cI_{c,k}|} \sum_{i\in \cI_{c,k}}  Y_i   \hat\phi(X_i)^\top \hat\Sigma_c^{-1} \hat\EE_t[\hat\phi]+ o_P(n^{-1/2})
    = \hat\EE_c[  Y   \hat\phi] ^\top \hat\Sigma_c^{-1} \hat\EE_t[\hat\phi]+ o_P(n^{-1/2}).
    \$
    Here, $\hat\EE_c[\cdot]$ denotes the empirical average over $i\in \cI_{c,k}$, and $\hat\EE_t[\cdot]$ denotes the empirical average over $i\in \cI_{t,k}$. By definition, we then have 
    \$
    \hat\EE_c[  Y   \hat\phi] &= \hat\EE_c[Y\phi^*] +  \hat\EE_c[Y \phi^\Delta]   + O_P(1/n),\\ 
    \hat\EE_t[\hat\phi] &= \hat\EE_t[\phi^*] +  \hat\EE_t[\phi^\Delta]  + O_P(1/n), \\ 
    \hat\Sigma_c &= \hat\EE_c[\hat\phi\hat\phi^\top] = \hat\EE_c[\phi^*(\phi^*)^\top] +  \hat\EE_c[ \phi^* (\phi^\Delta)^\top] +  \hat\EE_c[  \phi^\Delta ( \phi^*)^\top]   + O_P(1/n).
    \$
    In cases of vectors and matrices, the notation $O_P(1/n)$ is in terms of the $L_2$-norms. 
    Noting several terms in the order of $O_P(1/\sqrt{n})$, including $\hat\EE^\top[f_j]$ and $\hat\EE^\top[g_\ell]$, we have 
    \$
    \hat\EE_c[  Y   \hat\phi]   
    &= \EE_c[Y\phi^*] +  \underbrace{\hat\EE_c[Y\phi^* -\EE_c[Y\phi^*]] + \big( (\hat\EE^\dag[f_j]^\top  \EE_c[Y\dot{m}_j])_{j=1}^J, \hat\EE^\dag[g_\ell]^\top  \EE_c[Y\dot{w}_\ell])_{\ell=1}^L \big)^\top}_{:= \epsilon_1 = O_P(1/\sqrt{n})} + O_P(1/n), \\ 
    \hat\EE_t[\hat\phi] &= \EE_t[\phi^*] + \underbrace{\hat\EE_t[\phi^* -\EE_t[\phi^*]] +  \big( (\hat\EE^\dag[f_j]^\top  \EE_t[ \dot{m}_j])_{j=1}^J, \hat\EE^\dag[g_\ell]^\top \EE_t[ \dot{w}_\ell])_{\ell=1}^L \big)^\top}_{:= \epsilon_2 = O_P(1/\sqrt{n})}  + O_P(1/n), \\ 
    \hat\Sigma_c &=   \Sigma_c^* + \underbrace{\hat\EE_c[\phi^*(\phi^*)^\top -\Sigma_c^*] +  \hat\EE_c[ A^\Delta]}_{O_P(1/\sqrt{n})}   + O_P(1/n),
    \$
    where $\Sigma_c^* = \EE_c[\phi\phi^*]$, and the matrix $A^\Delta(x)\in \RR^{(J+L)\times(J+L)}$ is given by $A_{ij}^\Delta(x) = \hat\EE^\dag [f_i]^\top \dot{m}_i(x) \bar{m}_j(x) +  \hat\EE^\dag [f_j]^\top \dot{m}_j(x) \bar{m}_i(x) $ for $i,j\leq J$, and similarly for other pairs $(i,j)$ involving the weight models. Thus, 
    \$
    \hat\Sigma_c^{-1} = (\Sigma_c^{*})^{-1} - \underbrace{(\Sigma_c^{*})^{-1} \Big(\hat\EE_c[\phi^*(\phi^*)^\top -\Sigma_c^*] +  \hat\EE_c[ A^\Delta]\Big)(\Sigma_c^{*})^{-1}}_{:=E_3 = O_P(1/\sqrt{n})} + O_P(1/n).
    \$
    These quantities give the asymptotic linear expansion of $\hat\theta_0^{(k)}$ around 
    \$
    \bar\theta_0 := \EE_c[Y\phi^*]^\top (\Sigma_c^*)^{-1} \EE_t[\phi^*].
    \$
    Namely, with the above definition of $\epsilon_1\in \RR^{J+L}$, $\epsilon_2 \in \RR^{J+L}$, and $E_3\in \RR^{(J+L)\times(J+L)}$, we have 
    \$
    \hat\theta_0^{(k)} = \bar\theta + \epsilon_1^\top  (\Sigma_c^*)^{-1} \EE_t[\phi^*] -  \EE_c[Y\phi^*]^\top E_3 \EE_t[\phi^*] +  \EE_c[Y\phi^*]^\top (\Sigma_c^*)^{-1} \epsilon_2 + O_P(1/n).
    \$
    Now, if we further consider the aggregated estimator $\hat\theta_0 = (\hat\theta_0^{(1)}-\hat\theta_0^{(2)})/2$, we have  
    \$
    \hat\theta_0^{(k)} = \bar\theta + \hat\epsilon_1^\top  (\Sigma_c^*)^{-1} \EE_t[\phi^*] -  \EE_c[Y\phi^*]^\top \hat{E}_3 \EE_t[\phi^*] +  \EE_c[Y\phi^*]^\top (\Sigma_c^*)^{-1} \hat\epsilon_2 + O_P(1/n),
    \$
    where the error terms are ($\hat\EE_c[\cdot]$ denotes empirical average over all control units, $\hat\EE_t[\cdot]$ denotes empirical average over all treated units, and $\hat\EE[\cdot]$ denotes empirical average over all units) given by 
    \$
    &\hat\epsilon_1 = \hat\EE_c[Y\phi^* -\EE_c[Y\phi^*]] + \big( (\hat\EE [f_j]^\top  \EE_c[Y\dot{m}_j])_{j=1}^J, \hat\EE [g_\ell]^\top  \EE_c[Y\dot{w}_\ell])_{\ell=1}^L \big)^\top,\\ 
    &\hat\epsilon_1 = \hat\EE_t[\phi^* -\EE_t[\phi^*]] +  \big( (\hat\EE [f_j]^\top  \EE_t[ \dot{m}_j])_{j=1}^J, \hat\EE [g_\ell]^\top \EE_t[ \dot{w}_\ell])_{\ell=1}^L \big)^\top, \\ 
    &\hat{E}_3 =(\Sigma_c^{*})^{-1} \Big(\hat\EE_c[\phi^*(\phi^*)^\top -\Sigma_c^*] +   \hat{A}^\Delta\Big)(\Sigma_c^{*})^{-1},
    \$
    and the matrix $\hat{A}^\Delta \in \RR^{(J+L)\times(J+L)}$ is given by 
    \$
    \hat{A}^\Delta_{ij} = \hat\EE[f_i]^\top \EE_c[\dot{m}_i(X)\bar{m}_j(X)] + \hat\EE[f_j]^\top \EE_c[\dot{m}_j(X)\bar{m}_i(X)].
    \$

    To summarize, so far, we have shown that there exists a function---whose expression is quite involved---$\gamma(x)\in \RR $, such that $\EE[\gamma(X)]=0$ and 
    \$
    \hat\theta_0 = \bar\theta_0 + \frac{1}{n}\sum_{i=1}^n \gamma(X_i) + O_P(1/n).
    \$
    The remaining of the proof aims to show $\bar\theta_0 = \theta_0$ under the given conditions. 

    Indeed, when $m_0(x) = m_{j,\bar\theta_j}(x)$, there exists   $\beta^*\in \RR^{(J+L)}$ whose $j$-th element is $1$ and other elements are zero, such that $m_0(x) = (\beta^*)^\top \phi^*(x)$. Thus, following the last part in the proof of Theorem~\ref{thm:convergence}, we have $\bar\theta_0 = \theta_0$. 
    On the other hand, when $e(x) = e_{\ell,\bar\eta_\ell}(x)$ hence $w^*(x) = w_{\ell,\bar\eta_\ell}(x)$, there exists $\lambda^* \in \RR^{J+L}$ whose $J+\ell$-th element is $1$ and other elements are zero, such that $w^*(x)=(\lambda^*)^\top \phi^*(x)$. Following the last part of the proof of Theorem~\ref{thm:convergence}, we have $\bar\theta_0=\theta_0$. 
    We thus conclude the proof of Theorem~\ref{thm:multi_param}. 
\end{proof}


\section{Proof for cross-balancing after variable selection}

\subsection{Proof of Theorem~\ref{thm:convergence_varsel}}
\label{subsec:proof_convergence_varsel}

\begin{proof}[Proof of Theorem~\ref{thm:convergence_varsel}]
    We first prove $\hat\theta_0=\bar\theta+o_P(1)$ for some fixed $\bar\theta\in \RR$. To this end, it suffices to prove $\hat\theta^{(k)}:=\frac{1}{|\cI_{c,k}|}\sum_{i\in \cI_{c,k}}\hat{w}_iY_i = \theta_0+o_P(1)$ for each $k=1,2$. 
    Due to symmetry between $k=1,2$, 
    we fix $k=1$  and condition on $\cI_{c,2}\cup\cI_{t,2}$ without loss of generality. Also, we write $\hat{S}:=\hat{S}^{(k)}$ for simplicity, which can then be viewed as fixed. 
    The given conditions state that $\|\bar{w} - \eta(\hat{S}^{(k)},w^*)^\top X_{\hat{S}^{(k)}}\|_{L_2(\PP_{\cdot\given T=0})}=o_P(1)$ and $\epsilon(\hat{S}^{(k)},\bar{m})=o_P(1)$. To be clear, the randomness in these $o_P(1)$ is with respect to $\cI_{c,2}\cup\cI_{t,2}$, which is conditioned on. 
    Most of our analysis in the next will be  conditional on $\hat{S}$, but sometimes after obtaining some $o_P(1)$ quantities conditional on $\cI_{c,2}\cup\cI_{t,2}$, we implicitly invoke Lemma~\ref{lem:cond_event} and Lemma~\ref{lem:cond_quant} to prove that they are also $o_P(1)$ marginally without explicitly mentioning this every time. 
    
    Write $\beta:= \beta^*(\hat{S})$ and $\eta = \eta(\hat{S},w^*)$, which are both fixed. Following the analysis in the proof of Theorem~\ref{thm:convergence}, we know that $\hat{w} = \hat\Phi_{c,\hat{S}}^\top \hat\lambda$, where $\hat\Phi_{c,\hat{S}}\in\RR^{|\cI_{c,1}|\times|\hat{S}|}$ is the data matrix, and $\lambda \in \RR^{|\hat{S}|}$ minimizes the unconstrained convex objective
    \$
    \hat{L}(\lambda) := \lambda^\top \hat\Sigma_{c,\hat{S}} \lambda - 2\lambda^\top \hat\EE_t[X_{\hat{S}}] + 2|\lambda|^\top \delta_n, \quad \lambda \in \RR^{|\hat{S}|}.
    \$
    Here $\hat\Sigma_{c,\hat{S}} = \hat\Phi_{c,\hat{S}}^\top \hat\Phi_{c,\hat{S}}/|\cI_{c,1}|$ is the sample covariance matrix, and $\hat\EE_t[\cdot]$ denotes the empirical mean over the treated samples $\cI_{t,1}$. 
    Let $\mu_{t,\hat{S}} = \EE[X_{\hat{S}}\given T=1]$ be the population mean for the treated. 


    Consider the objective (which is also deterministic since we condition on $\hat{S}$)
    \$
    L(\lambda) = \lambda^\top \Sigma_{c,\hat{S}} \lambda - 2\lambda^\top \mu_{t,\hat{S}},
    \$
    and $\bar\lambda$ be the minimizer of $L(\lambda)$, which is equivalently 
    \$
    \bar\lambda = \Sigma_{c,\hat{S}}^{-1} \mu_{t,\hat{S}} 
    = \Sigma_{c,\hat{S}}^{-1} \EE_c[w^*(X) X_{\hat{S}}] = \eta(\hat{S};w^*),
    \$
    where $\EE_c[\cdot]$ denotes expectation under $\PP(\cdot\given T=0)$, viewing $\hat{S}$ as fixed. 
    Since $\nabla^2 \hat{L}(\lambda) = 2 \hat\Sigma_{c,\hat{S}}$, 
    recalling $\nu_{\min}(\hat\Sigma_{c,\hat{S}})$ is the minimum eigenvalue of $\hat\Sigma_{c,\hat{S}}$,  Lemma~\ref{lem:convex} with $f=\hat{L}$, $\theta_0=\bar\lambda$, and $\theta  = \hat\lambda$ yields
    \$
    \hat{L}(\hat\lambda)&\geq \hat{L}(\bar\lambda) + \nabla \hat{L}( \bar\lambda)^\top (\bar\lambda-\hat\lambda) + \nu_{\min}(\hat\Sigma_{c,\hat{S}})\cdot \min \big\{ \|\hat\lambda-\bar\lambda\|^2, c\|\hat\lambda-\bar\lambda\| \big\} \\ 
    &\geq \hat{L}(\hat\lambda) + \nabla \hat{L}(\bar\lambda)^\top (\bar\lambda-\hat\lambda) + \nu_{\min}(\hat\Sigma_{c,\hat{S}})\cdot \min \big\{ \|\hat\lambda-\bar\lambda\|^2, c\|\hat\lambda-\bar\lambda\| \big\}
    \$
    for any constant $c>0$ 
    since $\hat\lambda$ minimizes $\hat{L}(\cdot)$. By the Cauchy-Schwarz inequality, we further have 
    $ 
      \|\hat\lambda-\bar\lambda\|   \leq \nu_{\min}(\hat\Sigma_{c,\hat{S}})^{-1}\cdot \|\nabla \hat{L}(\bar\lambda) \|.
    $
    Using the fact that $\nabla L(\bar\lambda)=0$,  we have 
    \$
    \| \nabla \hat{L}(\bar\lambda) \| &=  \| \nabla \hat{L}(\bar\lambda) - \nabla  {L}(\bar\lambda) \|  \leq \|   (\hat\Sigma_{c,\hat{S}}-\Sigma_{c,\hat{S}} )\bar\lambda \| + 2\big| (\hat\EE_t[X_{\hat{S}}] - \EE_t[X_{\hat{S}}] )^\top \bar\lambda \big| + 2\delta_n |\hat{S}|^{1/2} .
    \$ 
    Since $X_{\hat{S}}^\top \bar\lambda$ is $ \|\bar\lambda\|^2 \sigma^2(\hat{S})$-sub-Gaussian, by the Chernoff's inequality,  with probability at least $1-\varepsilon$, 
    \$
    \big| (\hat\EE_t[X_{\hat{S}}] - \EE_t[X_{\hat{S}}] )^\top \bar\lambda \big|
    \leq \|\bar\lambda\| \cdot  \sigma(\hat{S}) \cdot \sqrt{\frac{2 \log(2/\varepsilon)}{|\cI_{c,1}|} },
    \$  
    hence $\big| (\hat\EE_t[X_{\hat{S}}] - \EE_t[X_{\hat{S}}] )^\top \bar\lambda \big| = o_P(1)$. 
    In addition, since $X_{i,\hat{S}}$'s are $\sigma^2(\hat{S})$-sub-Gaussian, by Lemma~\ref{lem:concen_covmat},  there exists a universal constant $C>0$ such that with probability at least $1-\varepsilon$, 
    \@\label{eq:concen_eigen_S}
    \|\hat\Sigma_{c,\hat{S}}- \Sigma_{c,\hat{S}}\|_{\textrm{op}} 
    \leq C \cdot \sigma^2(\hat{S})\cdot \max \Bigg\{ \sqrt{\frac{|\hat{S}|+\log(2/\varepsilon)}{|\cI_{c,1}|}},  \frac{|\hat{S}|+\log(2/\varepsilon)}{|\cI_{c,1}|}   \Bigg\}.
    \@
    Since $|\hat{S}| = o_P(n)$, these two bounds imply 
    \$
\| \nabla \hat{L}(\bar\lambda) \| \leq O_P(\|\bar\lambda\| \cdot \sigma^2(\hat{S})\cdot  |\hat{S}|^{1/2}n^{-1/2}).
    \$ 
    Since $\nu_{\min}( \Sigma_{c,\hat{S}})  = \Omega_P(\sigma^2(\hat{S})( |\hat{S}|/n)^{1/2})$, we know 
    $
    \nu_{\min}(\hat\Sigma_{c,\hat{S}}) \geq \nu_{\min}( \Sigma_{c,\hat{S}}) - O_P(\sigma^2(\hat{S})(|\hat{S}|/n)^{1/2}),
    $
    and hence by~\eqref{eq:concen_eigen_S}
    \@\label{eq:hat_lam_conv}
    \|\hat\lambda - \bar\lambda\|\leq O_P\big( \nu_{\min}( \Sigma_{c,\hat{S}})^{-1} \cdot \|\bar\lambda\| \cdot \sigma^2(\hat{S}) \cdot  |\hat{S}|^{1/2}n^{-1/2} \big) 
    = O_P\big( \zeta(\hat{S})\cdot n^{-1/2} \big),
    \@
    where the probability is with respect to $\cI_{c,1}\cup \cI_{t,1}$, viewing $\hat{S}$ as fixed. Now using the fact that $\|\bar\lambda\|\cdot \sigma^2(\hat{S}) \cdot  |\hat{S}|^{1/2}n^{-1/2} /\nu_{\min}( \Sigma_{c,\hat{S}}) = o_P(1)$, we have $\|\hat\lambda-\bar\lambda\|=o_P(1)$ where the probability is marginally over all data. 
    
    We now proceed to analyze the estimator $\hat\theta^{(k)}$. Define the imbalance gap 
    \$
    \hat\Delta:= \frac{1}{|\cI_{c,1}|} \sum_{i\in \cI_{c,1}} w_i X_{i,\hat{S}} - \frac{1}{|\cI_{t,1}|} \sum_{j\in \cI_t^{(1)}} X_{j,\hat{S}}  \in \RR^{|\hat{S}|},
    \$
    which obeys $|\hat\Delta_j|\leq \delta_n$ for any entry $j \in [|\hat{S}|]$. Thus, by the balancing condition we know 
    \$
    \hat\theta^{(k)} = \frac{1}{|\cI_{c,1}|}\sum_{i\in \cI_{c,1}}\hat{w}_i (Y_i - \beta^\top X_{i,\hat{S}}) + \frac{1}{|\cI_{t,1}|}\sum_{j\in \cI_{t,1}} \beta^\top X_{j,\hat{S}} - \beta^\top \hat\Delta. 
    \$
    Reorganizing the above equation, we have 
    \$
    \hat\theta^{(k)}& = \frac{1}{|\cI_{c,1}|}\sum_{i\in \cI_{c,1}} \bar{w}(X_i)(Y_i-\bar{m}(X_i)) + \frac{1}{|\cI_{t,1}|}\sum_{j\in \cI_{t,1}}  \bar{m}(X_j) \\ 
    &\quad + \underbrace{
        \frac{1}{|\cI_{c,1}|}\sum_{i\in \cI_{c,1}} (\hat{w}_i - \eta(\hat{S};w^*)^\top X_{i,\hat{S}} ) (Y_i - \bar{m}(X_i))
    }_{\textrm{term (a)}} \\ 
    &\quad + \underbrace{
        \frac{1}{|\cI_{c,1}|}\sum_{i\in \cI_{c,1}} \big[ \eta(\hat{S};w^*)^\top X_{i,\hat{S}} - \bar{w}(X_i) \big] \cdot \big[Y_i - \bar{m}(X_i)\big]
    }_{\textrm{term (b)}}\\ 
    &\quad + \underbrace{
        \frac{1}{|\cI_{c,1}|}\sum_{i\in \cI_{c,1}}  \hat{w}_i (\bar{m}(X_i) - \beta^\top X_{i,\hat{S}})
    }_{\textrm{term (c)}} - \underbrace{
        \frac{1}{|\cI_{t,1}|}\sum_{j\in \cI_{t,1}}   (\bar{m}(X_j) - \beta^\top X_{j,\hat{S}})
    }_{\textrm{term (d)}} - \beta^\top \hat\Delta.
    \$

    Note that $|\beta^\top \hat\Delta|\leq \delta_n \cdot \|\beta\| = o_P(1)$ by condition (ii). In the next, we show that terms (a-d) are $o_P(1)$. First note that 
    \$
     \text{term (a)} = \frac{1}{|\cI_{c,1}|}\sum_{i\in \cI_{c,1}} (\hat{w}_i - \eta(\hat{S};w^*)^\top X_{i,\hat{S}} ) (Y_i - \bar{m}(X_i))= \frac{1}{|\cI_{c,1}|}\sum_{i\in \cI_{c,1}}  (\hat\lambda-\bar\lambda)^\top X_{i,\bar{S}} (Y_i - \bar{m}(X_i)).
    \$
    Since $\|\hat\lambda-\bar\lambda \|=o_P(1)$, we have $|\text{term (a)}| = o_P(1)$ by  Markov's inequality and noting that 
    \$
    \EE\Bigg\| \frac{1}{|\cI_{c,1}|}\sum_{i\in \cI_{c,1}}    X_{i,\hat{S}} (Y_i - \bar{m}(X_i))   \Bigg\|^2 \leq \frac{1}{|\cI_{c,1}|}\sum_{j\in \hat{S}} \EE\Big[ X_{ij}^2(Y_i-\bar{m}(X_i))^2\Biggiven T_i=0\Big] = O_P(|\hat{S}|/n)
    \$ 
    since $\EE[(Y(0)-\bar{m}(X)^2\given X]$ is bounded, $X_{i,\hat{S}}(Y_i-\bar{m}(X_i))$ is mean zero, and each entry $X_{ij}$ is sub-Gaussian. 
    Second, by the Cauchy-Schwarz inequality, 
    \$
    \big|\textrm{term (b)}\big|\leq 
    \sqrt{\frac{1}{|\cI_{c,1}|}\sum_{i\in \cI_{c,1}} \big[ \eta(\hat{S};w^*)^\top X_{i,\hat{S}} - \bar{w}(X_i) \big]^2 }\cdot \sqrt{\frac{1}{|\cI_{c,1}|}\sum_{i\in \cI_{c,1}} \big[Y_i - \bar{m}(X_i)\big]},
    \$
    where the second term is $O(1)$  $\EE[(Y(0)-\bar{m}(X)^2\given X]$ is bounded. Applying Chernoff's inequality yields 
    \@\label{eq:approx_to_bar}
    \frac{1}{|\cI_{c,1}|}\sum_{i\in \cI_{c,1}} \big[ \eta(\hat{S};w^*)^\top X_{i,\hat{S}} - \bar{w}(X_i) \big]^2 = \epsilon(\hat{S};\bar{w})^2 + O_P(1/\sqrt{n}) = o_P(1),
    \@
    which further implies $|\text{term (b)}| = o_P(1)$. Again by Cauchy-Schwarz inequality, 
    \$
    \big|\textrm{term (c)}\big|\leq 
    \sqrt{\frac{1}{|\cI_{c,1}|}\sum_{i\in \cI_{c,1}} \hat{w}_i^2 }\cdot \sqrt{\frac{1}{|\cI_{c,1}|}\sum_{i\in \cI_{c,1}} \big[\bar{m}(X_i) - \beta^\top X_{i,\hat{S}}\big]^2},
    \$
    where the second term in the product is $o_P(1)$. For the first term, noting that 
    \$
    \frac{1}{|\cI_{c,1}|}\sum_{i\in \cI_{c,1}} (\hat{w}_i - \eta(\hat{S};w^*)^\top X_{i,\hat{S}})^2 = (\hat\lambda-\bar\lambda)^\top \hat\Sigma_{c,\hat{S}} (\hat\lambda-\bar\lambda) \leq \|\hat\lambda-\bar\lambda\|^2 \cdot \|\hat\Sigma_{c,\hat{S}}\|_{\textnormal{op}} .
    \$
    Combining~\eqref{eq:concen_eigen_S} and~\eqref{eq:hat_lam_conv}, we have 
    \$
     \|\hat\lambda-\bar\lambda\|^2 \cdot \|\hat\Sigma_{c,\hat{S}}\|_{\textnormal{op}} 
     & \leq  \|\hat\lambda-\bar\lambda\|^2 \cdot  \Big(\|\Sigma_{c,\hat{S}}\|_{\textrm{op}} + O_P\big(\|\bar\lambda\|\sigma^2(\hat{S})(|\hat{S}|/n)^{1/2}\big)\Big) \\ 
     &= O_P\big(\zeta^2(\hat{S})\cdot n^{-1} \cdot  \|\Sigma_{c,\hat{S}}\|_{\textrm{op}}\big) + O_P\big(  \zeta^2(\hat{S})\cdot n^{-3/2} \cdot \nu_{\min}(\Sigma_{c,\hat{S}})\cdot \zeta(\hat{S}) \big) = O_P(1)
    \$
    by condition (iii). 
    Combining this result with equation~\eqref{eq:approx_to_bar} and $\frac{1}{|\cI_{c,1}|}\sum_{i\in \cI_{c,1}}  \bar{w}(X_i)^2=O_P(1)$, and applying Cauchy-Schwarz inequality twice, we know that the first term is $O_P(1)$. Putting them together, we obtain $|\textrm{term (c)}| = o_P(1)$. 
    Also, by similar arguments as before we know $|\textrm{term (d)}| = o_P(1)$. 
    These results altogether lead to 
    \$
    \hat\theta^{(k)}& = \frac{1}{|\cI_{c,1}|}\sum_{i\in \cI_{c,1}} \bar{w}(X_i)(Y_i-\bar{m}(X_i)) + \frac{1}{|\cI_{t,1}|}\sum_{j\in \cI_{t,1}}  \bar{m}(X_j) + o_P(1) = \bar\theta + o_P(1),
    \$
    where $\bar\theta:= \EE[(1-T)/(1-p)\cdot\bar{w}(X)(Y-\bar{m}(X)) + T/p\cdot \bar{m}(X)$. 
    Finally, noting that $\bar\theta=\theta_0$ if either $\bar{w}=w^*$ or $\bar{m}=m_0$ similar to the proof of previous theorems, we complete the proof of Theorem~\ref{thm:convergence_varsel}. 
\end{proof}

\subsection{Proof of Theorem~\ref{thm:convergence_varsel_alt}}
\label{subsec:proof_convergence_varsel_alt}


\begin{proof}[Proof of Theorem~\ref{thm:convergence_varsel_alt}]
    We first prove $\hat\theta_0=\bar\theta+o_P(1)$ for some fixed $\bar\theta\in \RR$. To this end, it suffices to prove $\hat\theta^{(k)}:=\frac{1}{|\cI_{c,k}|}\sum_{i\in \cI_{c,k}}\hat{w}_iY_i = \theta_0+o_P(1)$ for each $k=1,2$. 
    Due to symmetry between $k=1,2$, 
    we fix $k=1$  and condition on $\cI_{c,2}\cup\cI_{t,2}$ without loss of generality. Also, we write $\hat{S}:=\hat{S}^{(k)}$ for simplicity, which can then be viewed as fixed. 
    The given conditions state that $\|\bar{w} - \eta(\hat{S}^{(k)},w^*)^\top X_{\hat{S}^{(k)}}\|_{L_2(\PP_{\cdot\given T=0})}=o_P(1)$ and $\epsilon(\hat{S}^{(k)},\bar{m})=o_P(1)$. To be clear, the randomness in these $o_P(1)$ is with respect to $\cI_{c,2}\cup\cI_{t,2}$, which is conditioned on. 
    Most of our analysis in the next will be  conditional on $\hat{S}$, but sometimes after obtaining some $o_P(1)$ quantities conditional on $\cI_{c,2}\cup\cI_{t,2}$, we implicitly invoke Lemma~\ref{lem:cond_event} to prove that they are also $o_P(1)$ marginally without mentioning this. 
    
    Write $\beta:= \beta^*(\hat{S})$ and $\eta = \eta(\hat{S},w^*)$, which are both deterministic conditional on $\cI_{c,2}\cup\cI_{t,2}$. Following the analysis in the proof of Theorem~\ref{thm:convergence}, we know that $\hat{w} = X_{c,\hat{S}}^\top \hat\lambda$, where $X_{c,\hat{S}}$ is the $\RR^{|\cI_{c,1}|\times|\hat{S}|}$ data matrix whose $i$-th row is $X_{i,\hat{S}}$, and $\lambda$ minimizes the unconstrained convex objective
    \$
    \hat{L}(\lambda) := \lambda^\top \hat\Sigma_{c,\hat{S}} \lambda - 2\lambda^\top \hat\EE_t[X_{\hat{S}}] + 2 \delta_n \|\lambda\|_1, \quad \lambda \in \RR^{|\hat{S}|},
    \$
    and $\hat\Sigma_{c,\hat{S}} = X_{c,\hat{S}}^\top X_{c,\hat{S}}/|\cI_{c,1}|$ is the sample covariance matrix. We write 
    \$
    \bar\lambda := \eta(\hat{S};w^*)
    \$
    which is deterministic conditional on $\cI_{c,2}\cup\cI_{t,2}$. 
    By condition (i), we know $\|\bar{w}-\bar\lambda^\top X_{\hat{S}}\|_{L_2(\PP_{\cdot\given T=0})}=o_P(1)$. 
    For convenience, we write 
    \$
    \hat\mu_t = \hat\EE_t[X_{\hat{S}}]\quad \text{and}\quad  X_c:= X_{c,\hat{S}}.
    \$ 
    In this way, $\hat{L}(\lambda) = \lambda^\top X_c^\top X_c \lambda - 2\lambda^\top\hat\mu_t + 2|\lambda|^\top \delta_n$. 
    By the fact that $\hat\lambda$ minimizes $\hat{L}(\lambda)$, we have 
    \$
    \frac{1}{|\cI_{c,1}|}\hat\lambda^\top X_c^\top X_c \hat\lambda - 2\hat\lambda^\top\hat\mu_t + 2\delta_n \|\hat\lambda\|_1
    \leq \frac{1}{|\cI_{c,1}|}\bar\lambda^\top X_c^\top X_c \bar\lambda - 2\bar\lambda^\top\hat\mu_t + 2\delta_n \|\bar\lambda\|_1 .
    \$
    Reorganizing terms, and denoting $n = |\cI_{c,1}|$, this means 
    \@\label{eq:risk_bd1}
    \frac{1}{|\cI_{c,1}|}\big\|X_c \hat\lambda - X_c \bar\lambda \big\|_2^2\leq 2(\bar\lambda-\hat\lambda)^\top (X_cX_c^\top \bar\lambda / |\cI_{c,1}| - \hat\mu_t) + 2\delta_n \cdot (\|\bar\lambda\|_1-\|\hat\lambda\|_1).  
    \@
    Now we define  
    \$
    \varepsilon_{t,j}:= X_{j,\hat{S}} - \mu_t,\quad \mu_t:=\EE[X_{j,\hat{S}}\given T_j=1], \quad 
    \varepsilon_{c,i}:= \bar{w}(X_i) X_{i,\hat{S}} - \bar\mu_c, \quad \bar\mu_c:= \EE[\bar{w}(X_i)X_{i,\hat{S}}\given T_i=0] 
    \$
    viewing $\hat{S}$ as fixed. Thus, $\{\varepsilon_{t,j}\}_{j\in \cI_{t,1}}$ and $\{\varepsilon_{c,i}\}_{i\in \cI_{c,1}}$ are i.i.d.~mean zero random variables, respectively. We also define $\bar{\epsilon}_{t}$ as the sample mean of $\{\varepsilon_{t,j}\}_{j\in \cI_{t,1}}$, and $\bar\varepsilon_{c}$ as the sample mean of $\{\varepsilon_{c,i}\}_{i\in \cI_{c,1}}$. 
    Finally, define 
    \$
    \varepsilon_{w,i} = \bar{w}(X_i) - \bar\lambda^\top X_{i,\hat{S}},
    \$
    where $\EE[\varepsilon_{w,i}^2] = o_P(1)$ (viewing $\hat{S}$ as fixed) by condition (i), 
    and let $\varepsilon_w = (\varepsilon_{w,i})_{i\in \cI_{c,1}}^\top\in \RR^{|\cI_{c,1}|}$ be the vector of all $\varepsilon_{w,i}$'s in the control fold. We then have 
    \$
  X_c^\top \varepsilon_w /|\cI_{c,1}| + X_cX_c^\top \bar\lambda/|\cI_{c,1}|  =  \bar\mu_c + \bar\varepsilon_c.
    \$
    Returning to~\eqref{eq:risk_bd1}, by the above definitions we have 
    \@\label{eq:risk_bd2}
    &\frac{1}{|\cI_{c,1}|} \big\|X_c \hat\lambda - X_c \bar\lambda \big\|_2^2 \notag \\ 
    &\leq 2(\bar\lambda-\hat\lambda)^\top \Big(X_cX_c^\top \bar\lambda/|\cI_{c,1}| - \bar\varepsilon_t - \mu_t + \bar{\varepsilon}_c + \bar\mu_c -  X_c^\top \varepsilon_w /|\cI_{c,1}| - X_cX_c^\top \bar\lambda/|\cI_{c,1}| \Big) + 2\delta_n (\|\bar\lambda\|_1-\|\hat\lambda\|_1)  \notag\\
    & = 2(\bar\lambda-\hat\lambda)^\top \big(  - \bar\varepsilon_t - \mu_t + \bar{\varepsilon}_c + \bar\mu_c - X_c^\top \varepsilon_w /|\cI_{c,1}|  \big) + 2\delta_n (\|\bar\lambda\|_1-\|\hat\lambda\|_1) \notag \\ 
    &\leq 2\|\bar\lambda - \hat\lambda\|_1 \cdot \big( \|\bar\varepsilon_t\|_\infty + \|\mu_t-\bar\mu_c\|_{\infty} + \|\bar\varepsilon_c\|_{\infty} + \|X_c^\top \varepsilon_w /|\cI_{c,1}|\|_{\infty}\big) + 2\delta_n (\|\bar\lambda\|_1-\|\hat\lambda\|_1) .
    \@
    We now bound the four terms in the summation separately:
    \begin{itemize}
        \item Note that each entry of $\bar\varepsilon_t$ is the sample mean of i.i.d.~mean zero sub-Gaussian variables, and thus applying Chernoff's bound and a union bound yields  
        \$
        \|\bar\varepsilon_t\|_{\infty} \leq \sqrt{\frac{2\log(2|\hat{S}|/\varepsilon)}{|\cI_{t,1}|}}
        \$
        with probability at least $1-\varepsilon$, where the probability is with respect to $\cI_{c,1}\cup\cI_{t,1}$. 
        \item With similar arguments, with probability at least $1-\varepsilon$  with respect to $\cI_{c,1}\cup\cI_{t,1}$, 
        \$
        \|\bar\varepsilon_c\|_{\infty} \leq \|\bar{w}(\cdot)\|_{\infty} \sqrt{\frac{2\log(2|\hat{S}|/\varepsilon)}{|\cI_{c,1}|}}.
        \$ 
        \item In addition, note that $\bar\varepsilon_{w} = O_P(\sqrt{\EE[\varepsilon_{w,i}^2]}) = o_P(1)$ by condition (i). Since $\bar\lambda = \eta(\hat{S};w^*)$, it hold that 
        \$
        \EE\big[X_{i,\hat{S}} [w^*(X_i) - \bar\lambda ^\top X_{i,\hat{S}}]\biggiven T_i=0\big] = 0
        \$
        where $\hat{S}$ is viewed as fixed. Therefore, as $w^*$ is the true density ratio, 
        \$
        \mu_t = \EE[w^*(X_i)X_{i,\hat{S}}\given T_i=0]
        = \EE[\bar\lambda^\top X_{i,\hat{S}} X_{i,\hat{S}}\given T_i=0] &=  \EE[ (\bar{w}(X_i)- \varepsilon_{w,i})X_{i,\hat{S}}\given T_i=0] \\
        &= \bar\mu_c - \EE[  \varepsilon_{w,i} X_{i,\hat{S}}\given T_i=0],
        \$
        where the last equality uses the definition of $\bar\mu_c$. On the other hand, for any $k\in \hat{S}$, we have (deterministically once $\hat{S}$ is conditioned on)
        \$
       \big| \EE[  \varepsilon_{w,i} X_{i,k}\given T_i=0] \big|
       \leq \sqrt{\EE[  \varepsilon_{w,i}^2 \given T_i=0]\cdot \EE[ X_{i,k}^2\given T_i=0]} = O( \epsilon(\hat{S};\bar{w})),
        \$
        which means $\|\mu_t-\bar\mu_c\|_{\infty} \leq \big\|\EE[  \varepsilon_{w,i} X_{i,\hat{S}}\given T_i=0]\big\|_\infty  = o_P(1)$ (marginally) by Lemma~\ref{lem:cond_event}. 
        \item Finally, note that 
        \$
       &  \big\|X_c^\top \varepsilon_w/|\cI_{c,1}|\big\|_\infty = 
        \max_{j\in \hat{S}} \bigg| \frac{1}{|\cI_{c,1}|} \sum_{i\in \cI_{c,1}} X_{i,j} (\bar{w}(X_i) - \bar\lambda^\top X_{i,\hat{S}})  \bigg| \\ 
        &\leq \max_{j\in \hat{S}} \bigg| \frac{1}{|\cI_{c,1}|} \sum_{i\in \cI_{c,1}} \Big\{ X_{i,j} (\bar{w}(X_i) - \bar\lambda^\top X_{i,\hat{S}})  - \EE\big[X_{i,j} (\bar{w}(X_i) - \bar\lambda^\top X_{i,\hat{S}})  \biggiven T_i=0\big] \Big\} \bigg|  \\ 
        &\qquad + \max_{j\in \hat{S}} \Big| \EE\big[X_{i,j} (\bar{w}(X_i) - \bar\lambda^\top X_{i,\hat{S}}) \biggiven T_i=0\big]\Big|.
        \$
        Here the first term is the maximum over $|\hat{S}|$  sample means of i.i.d.~mean zero and sub-Gaussian variables, hence applying the Chernoff's bound with a union bound, it is upper bounded by 
        \$
     \sup_x\big| \bar{w}(x)-\bar\eta^\top x_{i,\hat{S}}\big| \cdot \sqrt{\frac{2\log(2|\hat{S}|/\varepsilon)}{|\cI_{c,1}|}}
        \$
        with probability at least $1-\varepsilon$. On the other hand, for each $j\in \hat{S}$, 
        \$
        &\Big| \EE\big[X_{i,j} (\bar{w}(X_i) - \bar\lambda^\top X_{i,\hat{S}}) \biggiven T_i=0\big]\Big| \leq \sqrt{\EE[X_{i,j}^2\given T_i=0]}\cdot \sqrt{\EE\big[  (\bar{w}(X_i) - \bar\lambda^\top X_{i,\hat{S}})^2 \biggiven T_i=0\big]} = O(\epsilon(\hat{S};\bar{w})).
        \$
    \end{itemize} 
    Combining these arguments with~\eqref{eq:risk_bd2} and applying the triangle inequality, we have 
    \@\label{eq:risk_bd33}
    &\frac{1}{|\cI_{c,1}|} \sum_{i\in \cI_{c,1}} (\hat{w}_i - X_{i,\hat{S}}^\top\bar\lambda)^2 = 
    \frac{1}{|\cI_{c,1}|} \big\|X_c \hat\lambda - X_c \bar\lambda \big\|_2^2 \notag \\
    &\leq 2 \big(\|\bar\lambda\|_1 + \| \hat\lambda\|_1\big) \cdot \big( \|\bar\varepsilon_t\|_\infty + \|\mu_t-\bar\mu_c\|_{\infty} + \|\bar\varepsilon_c\|_{\infty} + \|\bar\varepsilon_w\|_{\infty}\big) + 2\delta_n (\|\bar\lambda\|_1-\|\hat\lambda\|_1)  =  O_P(\delta_n \|\bar\lambda\|_1)
    \@
    as long as $\|\bar\varepsilon_t\|_\infty + \|\mu_t-\bar\mu_c\|_{\infty} + \|\bar\varepsilon_c\|_{\infty} + \|\bar\varepsilon_w\|_{\infty} = O_P(\delta_n)$, which is true since condition (iii) holds and we have shown that 
    \$
    \|\bar\varepsilon_t\|_\infty + \|\mu_t-\bar\mu_c\|_{\infty} + \|\bar\varepsilon_c\|_{\infty} + \|\bar\varepsilon_w\|_{\infty} = O_P\big(\sqrt{\log|\hat{S}|/n} + \epsilon(\hat{S};\bar{w})\big).
    \$

    We now proceed to analyze the estimator $\hat\theta^{(k)}$. Define the imbalance gap 
    \$
    \hat\Delta:= \frac{1}{|\cI_{c,1}|} \sum_{i\in \cI_{c,1}} w_i X_{i,\hat{S}} - \frac{1}{|\cI_{t,1}|} \sum_{j\in \cI_t^{(1)}} X_{j,\hat{S}}  \in \RR^{|\hat{S}|},
    \$
    which obeys $|\hat\Delta_j|\leq \delta_n$ for any entry $j \in [|\hat{S}|]$. Thus, by the balancing condition we know 
    \$
    \hat\theta^{(k)} = \frac{1}{|\cI_{c,1}|}\sum_{i\in \cI_{c,1}}\hat{w}_i (Y_i - \beta^\top X_{i,\hat{S}}) + \frac{1}{|\cI_{t,1}|}\sum_{j\in \cI_{t,1}} \beta^\top X_{j,\hat{S}} - \beta^\top \hat\Delta. 
    \$
    Reorganizing the above equation, we have 
    \$
    \hat\theta^{(k)}& = \frac{1}{|\cI_{c,1}|}\sum_{i\in \cI_{c,1}} \bar{w}(X_i)(Y_i-\bar{m}(X_i)) + \frac{1}{|\cI_{t,1}|}\sum_{j\in \cI_{t,1}}  \bar{m}(X_j) \\ 
    &\quad + \underbrace{
        \frac{1}{|\cI_{c,1}|}\sum_{i\in \cI_{c,1}} (\hat{w}_i - \bar\lambda^\top X_{i,\hat{S}} ) (Y_i - \bar{m}(X_i))
    }_{\textrm{term (a)}}  + \underbrace{
        \frac{1}{|\cI_{c,1}|}\sum_{i\in \cI_{c,1}} \big[ \bar\lambda^\top X_{i,\hat{S}} - \bar{w}(X_i) \big] \cdot \big[Y_i - \bar{m}(X_i)\big]
    }_{\textrm{term (b)}}\\ 
    &\quad + \underbrace{
        \frac{1}{|\cI_{c,1}|}\sum_{i\in \cI_{c,1}}  \hat{w}_i (\bar{m}(X_i) - \beta^\top X_{i,\hat{S}})
    }_{\textrm{term (c)}} - \underbrace{
        \frac{1}{|\cI_{t,1}|}\sum_{j\in \cI_{t,1}}   (\bar{m}(X_j) - \beta^\top X_{j,\hat{S}})
    }_{\textrm{term (d)}} - \beta^\top \hat\Delta.
    \$

    Note that $|\beta^\top \hat\Delta|\leq \delta_n \cdot \|\beta\|_1 = o_P(1)$ by condition (ii). In the next, we show that terms (a-d) are $o_P(1)$. First, by H\"older's inequality, 
    \$
     \big| \text{term (a)} \big|   
     \leq \sqrt{\frac{1}{|\cI_{c,1}|}\sum_{i\in \cI_{c,1}} (\hat{w}_i - \bar\lambda^\top X_{i,\hat{S}} )^2} \cdot \sqrt{\frac{1}{|\cI_{c,1}|}\sum_{i\in \cI_{c,1}}  (Y_i - \bar{m}(X_i))^2},
    \$
    where the second term is $O_P(1)$ and the first term is $O_P(\delta_n \|\bar\lambda\|_1) = o_P(1)$ by~\eqref{eq:risk_bd33}, and therefore  $|\text{term (a)}| = o_P(1)$. 
    Second, by  Cauchy-Schwarz inequality, 
    \$
    \big|\textrm{term (b)}\big|\leq 
    \sqrt{\frac{1}{|\cI_{c,1}|}\sum_{i\in \cI_{c,1}} \big[ \bar\lambda^\top X_{i,\hat{S}} - \bar{w}(X_i) \big]^2 }\cdot \sqrt{\frac{1}{|\cI_{c,1}|}\sum_{i\in \cI_{c,1}} \big[Y_i - \bar{m}(X_i)\big]},
    \$
    where the bounded conditional second moment of $Y_i-\bar{m}(X_i)$ implies that the second term is $O(1)$, and applying Chernoff's inequality yields 
    \@\label{eq:approx_to_bar}
    \frac{1}{|\cI_{c,1}|}\sum_{i\in \cI_{c,1}} \big[ \bar\lambda^\top X_{i,\hat{S}} - \bar{w}(X_i) \big]^2 = \epsilon(\hat{S};\bar{w})^2 + O_P(1/\sqrt{n}) = o_P(1),
    \@
    which further implies $|\text{term (b)}| = o_P(1)$. Then, again by Cauchy-Schwarz inequality, 
    \$
    \big|\textrm{term (c)}\big|\leq 
    \sqrt{\frac{1}{|\cI_{c,1}|}\sum_{i\in \cI_{c,1}} \hat{w}_i^2 }\cdot \sqrt{\frac{1}{|\cI_{c,1}|}\sum_{i\in \cI_{c,1}} \big[\bar{m}(X_i) - \beta^\top X_{i,\hat{S}}\big]^2},
    \$
    where the second term in the product is $o_P(1)$. For the first term, noting that by~\eqref{eq:risk_bd33} we have 
    \$
    \frac{1}{|\cI_{c,1}|}\sum_{i\in \cI_{c,1}} (\hat{w}_i - \bar\lambda^\top X_{i,\hat{S}})^2 =  o_P(1).
    \$ 
    Combining this result with equation~\eqref{eq:approx_to_bar} and applying Cauchy-Schwarz inequality twice, we know that the first term is $O_P(1)$. Putting them together, we obtain $|\textrm{term (c)}| = o_P(1)$. 
    Also, by similar arguments as before we know $|\textrm{term (d)}| = o_P(1)$. 
    These results altogether lead to 
    \$
    \hat\theta^{(k)}& = \frac{1}{|\cI_{c,1}|}\sum_{i\in \cI_{c,1}} \bar{w}(X_i)(Y_i-\bar{m}(X_i)) + \frac{1}{|\cI_{t,1}|}\sum_{j\in \cI_{t,1}}  \bar{m}(X_j) + o_P(1) = \bar\theta + o_P(1),
    \$
    where $\bar\theta:= \EE[(1-T)/(1-p)\cdot\bar{w}(X)(Y-\bar{m}(X)) + T/p\cdot \bar{m}(X)$. 
    Finally, noting that $\bar\theta=\theta_0$ if either $\bar{w}=w^*$ or $\bar{m}=m_0$ using arguments similar to the proof of Theorem~\ref{thm:convergence}, we complete the proof. 
\end{proof}

\subsection{Proof of Theorem~\ref{thm:inf_varsel}}
\label{app:proof_thm_inf_varsel}

\begin{proof}[Proof of Theorem~\ref{thm:inf_varsel}]
    We follow the analysis technique in the first part of the proof of Theorem~\ref{thm:convergence_varsel_alt}. Again, we fix $k=1$ and condition on $\cI_{c,2}\cup\cI_{t,2}$, so that $\hat{S}:=\hat{S}^{(k)}$ can be viewed as fixed. 
    We also define the orthogonal projection vectors $\beta:= \eta(\hat{S};m_0)$ and $\bar\lambda = \eta(\hat{S},w^*)$, which are both deterministic conditional on $\cI_{c,2}\cup\cI_{t,2}$. By condition (i), we know $\|w^*(X)-\bar\lambda^\top X_{\hat{S}}\|_{L_2(\PP_{\cdot\given T=0})}=o_P(1)$. 
    In addition, we write $\hat\mu_t = \hat\EE_t[X_{\hat{S}}]$ and $X_c:= X_{c,\hat{S}}$ for convenience. 
    
    Following the proof of Theorem~\ref{thm:convergence} and  Theorem~\ref{thm:convergence_varsel_alt}, we know that $\hat{w} = X_{c,\hat{S}}^\top \hat\lambda$, where $X_{c,\hat{S}}$ is the $\RR^{|\cI_{c,1}|\times|\hat{S}|}$ data matrix whose $i$-th row is $X_{i,\hat{S}}$, and $\lambda$ minimizes the unconstrained convex objective
    \$
    \hat{L}(\lambda) := \lambda^\top \hat\Sigma_{c,\hat{S}} \lambda - 2\lambda^\top \hat\EE_t[X_{\hat{S}}] + 2 \delta_n \|\lambda\|_1, \quad \lambda \in \RR^{|\hat{S}|},
    \$
    and $\hat\Sigma_{c,\hat{S}} = X_{c }^\top X_{c }/|\cI_{c,1}|$ is the sample covariance matrix. In addition, according to~\eqref{eq:risk_bd1},
    \@\label{eq:risk_bd1_inf}
    \frac{1}{|\cI_{c,1}|}\big\|X_c \hat\lambda - X_c \bar\lambda \big\|_2^2\leq 2(\bar\lambda-\hat\lambda)^\top (X_cX_c^\top \bar\lambda / |\cI_{c,1}| - \hat\mu_t) + 2\delta_n  \cdot (\|\bar\lambda\|_1-\|\hat\lambda\|_1).  
    \@
    where we recall that $X_c\in \RR^{|\cI_{c,1}|\times|\hat{S}|}$ as the data matrix whose $i$-th row is $X_{i,\hat{S}}$. 
    Define 
    \$
    \varepsilon_{t,j}:= X_{j,\hat{S}} - \mu_t,\quad \mu_t:=\EE[X_{j,\hat{S}}\given T_j=1], \quad 
    \varepsilon_{c,i}:=  {w}^*(X_i) X_{i,\hat{S}} - \mu_c, \quad  \mu_c:= \EE[w^*(X_i)X_{i,\hat{S}}\given T_i=0] 
    \$
    viewing $\hat{S}$ as fixed.  
    Thus, $\{\varepsilon_{t,j}\}_{j\in \cI_{t,1}}$ and $\{\varepsilon_{c,i}\}_{i\in \cI_{c,1}}$ are i.i.d.~mean zero random variables, respectively. We also define $\bar{\varepsilon}_{t}$ as the sample mean of $\{\varepsilon_{t,j}\}_{j\in \cI_{t,1}}$, and $\bar\varepsilon_{c}$ as the sample mean of $\{\varepsilon_{c,i}\}_{i\in \cI_{c,1}}$. In addition, since $w^*$ is the true density ratio, we have $\mu_t = \mu_c$. 
    Finally, we define 
    \$
    \varepsilon_{w,i} = {w}^*(X_i) - \bar\lambda^\top X_{i,\hat{S}},
    \$
    and $\varepsilon_w \in \RR^{|\cI_{c,1}|}$ is the vector whose $i$-th element is $\varepsilon_{w,i}$ for $i\in \cI_{c,1}$. 
    By definition, $\bar\lambda$ minimizes $\EE[(w^*(X_i)-\lambda^\top X_{i,\hat{S}})^2\given T_i=0]$; the first-order condition implies 
    \$
    \EE\big[ X_{i,\hat{S}} \varepsilon_{w,i}\biggiven T_i=0  \big] = 0,
    \$
    where again we remind that $\hat{S}$ is viewed as fixed. In addition, by definition we have 
    \$
    \bar\varepsilon_c + \mu_c =  X_c^\top \varepsilon_w/|\cI_{c,1}| + X_cX_c^\top \bar\lambda/|\cI_{c,1}| .
    \$
    Returning to~\eqref{eq:risk_bd1_inf}, we then have  
    \@\label{eq:risk_bd_inf}
     \frac{1}{|\cI_{c,1}|} \big\|X_c \hat\lambda - X_c \bar\lambda \big\|_2^2  
    &\leq 2(\bar\lambda-\hat\lambda)^\top \Big( - \bar\varepsilon_t - \mu_t + \bar{\varepsilon}_c +  \mu_c - X_c^\top \varepsilon_w/|\cI_{c,1}|  \Big) + 2\delta_n (\|\bar\lambda\|_1-\|\hat\lambda\|_1)  \notag\\
    &\leq 2\|\bar\lambda - \hat\lambda\|_1 \cdot \big( \|\bar\varepsilon_t\|_\infty +  \|\bar\varepsilon_c\|_{\infty} + \|X_c^\top \varepsilon_w/|\cI_{c,1}|\|_{\infty}\big) + 2\delta_n \cdot (\|\bar\lambda\|_1-\|\hat\lambda\|_1) ,
    \@
    where the first inequality uses the definitions, and the second inequality is due to $\mu_t = \mu_c$. 
    With similar arguments as in the proof of Theorem~\ref{thm:convergence_varsel_alt} and a union bound, with probability at least $1-2\varepsilon$, 
    \$
    \|\bar\varepsilon_t\|_{\infty} \leq \sqrt{\frac{2\log(2|\hat{S}|/\varepsilon)}{|\cI_{t,1}|}},\quad \|\bar\varepsilon_c\|_{\infty} \leq \|\bar{w}(\cdot)\|_{\infty} \sqrt{\frac{2\log(2|\hat{S}|/\varepsilon)}{|\cI_{c,1}|}}
    \$
    In addition, since $\EE[X_{i,\hat{S}}\epsilon_{w,i}\given T_i=0]=0$, we know $X_c^\top \varepsilon_w/|\cI_{c,1}|$ is the mean of i.i.d.~mean zero random vectors each entry of which is $M_2^2$-sub-Gaussian, where the constant $M_2>0$ is given in  condition (ii). Applying Chernoff's bound and a union bound over all entries in $\hat{S}$, with probability at least $1-\varepsilon$, 
    \$
    \|X_c^\top \varepsilon_w/|\cI_{c,1}|\|_{\infty} 
    \leq M_2 \sqrt{\frac{2\log(2|\hat{S}|/\varepsilon)}{|\cI_{c,1}|}}.
    \$
    Combining these arguments, we have 
    \$
    A_n := \|\bar\varepsilon_t\|_\infty +  \|\bar\varepsilon_c\|_{\infty} + \|X_c^\top \varepsilon_w/|\cI_{c,1}|\|_{\infty} = O_P(\sqrt{\log(|\hat{S}|)/n}).
    \$
    Thus, continuing with~\eqref{eq:risk_bd_inf} and applying the triangle inequality, we have 
    \@\label{eq:risk_bd3}
    & 
    \frac{1}{|\cI_{c,1}|} \big\|X_c (\hat\lambda -  \bar\lambda ) \big\|_2^2  \leq 2 A_n\cdot  \|\bar\lambda - \hat\lambda\|_1   + 2\delta_n \cdot (\|\bar\lambda\|_1-\|\hat\lambda\|_1) 
    \leq 2(A_n+\delta_n)\cdot \|\bar\lambda - \hat\lambda\|_1 .
    \@
    Applying condition (iii) to $\hat\lambda-\bar\lambda \in \RR^{|\hat{S}|}$ yields 
    \$
    \frac{1}{|\cI_{c,1}|} \big\|X_c (\hat\lambda -  \bar\lambda ) \big\|_2^2 \geq \xi(\hat{S})\cdot \|\bar\lambda - \hat\lambda\|_1 ^2,
    \$
    which further implies $\|\bar\lambda - \hat\lambda\|_1 \leq  2(A_n+\delta_n)/\xi(\hat{S})$, hence (since  $\delta_n = o_P(n^{-1/2})$)
    \$
    \frac{1}{|\cI_{c,1}|} \sum_{i\in \cI_{c,1}} (\hat{w}_i - X_{i,\hat{S}}^\top\bar\lambda)^2 = 
    \frac{1}{|\cI_{c,1}|} \big\|X_c (\hat\lambda -  \bar\lambda ) \big\|_2^2   \leq \frac{4(A_n+\delta_n)^2}{\xi(\hat{S})} = O_P\bigg(\frac{ \log|\hat{S}|}{n\cdot \xi(\hat{S})}\bigg). 
    \$ 

    We now proceed to analyze the estimator.  
    Define the imbalance gap 
    \$
    \hat\Delta:= \frac{1}{|\cI_{c,1}|} \sum_{i\in \cI_{c,1}} w_i X_{i,\hat{S}} - \frac{1}{|\cI_{t,1}|} \sum_{j\in \cI_t^{(1)}} X_{j,\hat{S}}  \in \RR^{|\hat{S}|},
    \$
    which obeys $\|\hat\Delta\|_\infty\leq \delta_n$ by the balancing constraint. Thus, by the balancing condition (recall  $\beta = \eta(\hat{S};m_0)$)
    \$
    \hat\theta^{(k)} := \frac{1}{|\cI_{c,1}|}\sum_{i\in \cI_{c,1}}\hat{w}_i Y_i= \frac{1}{|\cI_{c,1}|}\sum_{i\in \cI_{c,1}}\hat{w}_i (Y_i - \beta^\top X_{i,\hat{S}}) + \frac{1}{|\cI_{t,1}|}\sum_{j\in \cI_{t,1}} \beta^\top X_{j,\hat{S}} - \beta^\top \hat\Delta. 
    \$
    Reorganizing the above equation, we have 
    \$
    \hat\theta^{(k)}& = \frac{1}{|\cI_{c,1}|}\sum_{i\in \cI_{c,1}} {w}^*(X_i)(Y_i- m_0(X_i)) + \frac{1}{|\cI_{t,1}|}\sum_{j\in \cI_{t,1}}  m_0(X_j) \\ 
    &\quad + \underbrace{
        \frac{1}{|\cI_{c,1}|}\sum_{i\in \cI_{c,1}} (\hat{w}_i - w^*(X_i) ) (Y_i - m_0(X_i))
    }_{\textrm{term (a)}}  + \underbrace{
         \frac{1}{|\cI_{c,1}|}\sum_{i\in \cI_{c,1}}  (\hat{w}_i - w^*(X_i))(m_0(X_i) - \beta^\top X_{i,\hat{S}})
    }_{\textrm{term (b)}}\\ 
    &\quad + \underbrace{
        \frac{1}{|\cI_{c,1}|}\sum_{i\in \cI_{c,1}}   w^*(X_i) \cdot (m_0(X_i) - \beta^\top X_{i,\hat{S}})
        - \frac{1}{|\cI_{t,1}|}\sum_{j\in \cI_{t,1}}   (m_0(X_j) - \beta^\top X_{j,\hat{S}})
    }_{\textrm{term (c)}} - \beta^\top \hat\Delta.
    \$
    Conditions (i)-(ii) imply $|\beta^\top \Delta|\leq \|\beta\|_1 \|\Delta\|_\infty \leq \delta_n \|\beta\|_1=o_P(n^{-1/2})$. 
    We now bound terms (a)-(c). 

    Let $\cD_{X}$ be the union of the full data in $\cI_{c,2}\cup\cI_{t,2}$ and the covariates in $\cI_{c,1}\cup\cI_{t,1}$. 
    Since $\hat{w}_i$ only depends on the covariates in $\cI_{c,1}\cup \cI_{t,1}$, conditional on $\cD_X$, each term in the summation of (a) is independent and mean zero. Thus, term (a) is mean zero and  (since $\EE[(Y(0)-m_0(X))^2\given X]\leq M_1$) 
    \$
    \EE\big[ |\text{term (a)}|^2 \biggiven \cD_X \big] 
    &\leq \frac{M_1}{|\cI_{c,1}|^2} \sum_{i\in \cI_{c,1}} (\hat{w}_i - w^*(X_i) )^2  \\
    &\leq \frac{2M_1}{|\cI_{c,1}|^2} \sum_{i\in \cI_{c,1}} (\hat{w}_i - X_{i,\hat{S}}^\top \bar\lambda )^2 + \frac{2M_1}{|\cI_{c,1}|^2} \sum_{i\in \cI_{c,1}} (  X_{i,\hat{S}}^\top \bar\lambda - w^*(X_i))^2
    \$
    where the second line uses Cauchy-Schwarz inequality. By the tower property, 
    \$
\EE\big[ |\text{term (a)}|^2   \big] = O_P\bigg(\frac{ \log|\hat{S}|}{n^2\cdot \xi(\hat{S})} + \frac{\epsilon(\hat{S};w^*)^2}{n}\bigg) = o_P(1/n)
    \$
    since $\epsilon(\hat{S};w^*)=o_P(1)$ by condition (i) and $\log|\hat{S}|/\xi(\hat{S}) = o_P(n)$ by condition (iii). This implies $|\text{term (a)}|=o_P(n^{-1/2})$ by Markov's inequality. 
    
    Second, by H\"older's inequality, 
    \$
    \big|\text{term (b)}\big| \leq \sqrt{\frac{1}{|\cI_{c,1}|}\sum_{i\in \cI_{c,1}}  (\hat{w}_i - w^*(X_i))^2} \cdot \sqrt{\frac{1}{|\cI_{c,1}|}\sum_{i\in \cI_{c,1}}  (m_0(X_i)- \beta^\top X_{i,\hat{S}})^2},
    \$  
    where using similar ideas as how we bound term (a),  
    \$
    \sqrt{\frac{1}{|\cI_{c,1}|}\sum_{i\in \cI_{c,1}}  (\hat{w}_i - w^*(X_i))^2} = O_P\bigg( \sqrt{\frac{\log|\hat{S}|}{n\cdot \xi(\hat{S})}} + \epsilon(\hat{S};w^*)\bigg).
    \$
    On the other hand, by Markov's inequality we know 
    \$
    \sqrt{\frac{1}{|\cI_{c,1}|}\sum_{i\in \cI_{c,1}}  (m_0(X_i)- \beta^\top X_{i,\hat{S}})^2} = O_P\big(\epsilon(\hat{S};m_0)\big).
    \$
    Since $\sqrt{\log|\hat{S}|/(n\xi(\hat{S}))} \cdot \epsilon(\hat{S};m_0)= o_P(n^{-1/2})$ by condition (iii) and $\epsilon(\hat{S};w^*)\cdot\epsilon(\hat{S};m_0)= o_P(n^{-1/2})$ by condition (i), we have $|\text{term (b)}|=o_P(n^{-1/2})$. 

    Let $\mu:= \EE[m_0(X_j) - \beta^\top X_{j,\hat{S}}\given T_j=1] = \EE[w^*(X_i) \cdot (m_0(X_i) - \beta^\top X_{i,\hat{S}})\given T_i=0]$, viewing $\hat{S}$ as fixed. Then
    \$
\textrm{term (c)} = \underbrace{\frac{1}{|\cI_{c,1}|}\sum_{i\in \cI_{c,1}}  \big[ w^*(X_i) \cdot (m_0(X_i) - \beta^\top X_{i,\hat{S}}) - \mu \big]}_{\textrm{term (c,1)}}
        - \underbrace{\frac{1}{|\cI_{t,1}|}\sum_{j\in \cI_{t,1}}   \big[ m_0(X_j) - \beta^\top X_{j,\hat{S}} - \mu\big]}_{\textrm{term (c,2)}}.
    \$
    Each term in the summation in term (c,1) is i.i.d.~and mean zero, hence by Chebyshev's inequality, 
    \$
    \PP\Big(|\textrm{term (c,1)}| >\varepsilon \Biggiven \cI_{c,2}\cup \cI_{t,2} \Big) \leq \frac{\Var(w^*(X_i) \cdot (m_0(X_i) - \beta^\top X_{i,\hat{S}}))}{\varepsilon^2\cdot |\cI_{c,1}|} \leq \frac{M_2^2\cdot \epsilon(\hat{S};m_0)^2}{\varepsilon^2\cdot |\cI_{c,1}|} 
    \$
    for the  constant $M_2>0$ given in condition (ii). Thus by Lemma~\ref{lem:cond_quant} we know   $|\textrm{term (c,1)}| = O_P(\epsilon(\hat{S};m_0)/n^{1/2}) = o_P(n^{-1/2})$ by condition (i). Using exactly the same idea we have $|\textrm{term (c,2)}|  = o_P(n^{-1/2})$ as well. 

    Putting the three bounds together, we then have 
    \$
    \hat\theta^{(k)}& = \frac{1}{|\cI_{c,1}|}\sum_{i\in \cI_{c,1}} {w}^*(X_i)(Y_i- m_0(X_i)) + \frac{1}{|\cI_{t,1}|}\sum_{j\in \cI_{t,1}}  m_0(X_j)  + o_P(n^{-1/2}).
    \$
    Similar results apply to $k=2$ in the same way, hence we have 
    \$
    \hat\theta_0 = \frac{1}{n}\sum_{i=1}^n  \bigg(\frac{1-T_i}{1-p}{w}^*(X_i)(Y_i- m_0(X_i)) + \frac{T_i}{p} m_0(X_i) \bigg)  + o_P(n^{-1/2})
    \$ 
    which leads to the desired result by noting that the efficient influence function of $\theta_0$ is $\frac{1-T_i}{1-p}{w}^*(X_i)(Y_i- m_0(X_i)) + \frac{T_i}{p} m_0(X_i)$ (see, e.g., \cite{hahn1998role}).
\end{proof}

\subsection{Proof of Theorem~\ref{thm:inf_varsel_finite}}
\label{app:subsec_inf_varsel_finite}

\begin{proof}[Proof of Theorem~\ref{thm:inf_varsel_finite}]
    Following exactly the same notations and ideas in the proof of Theorem~\ref{thm:inf_varsel}, 
    \@\label{eq:risk_bd_inf_finite}
     \frac{1}{|\cI_{c,1}|} \big\|X_c \hat\lambda - X_c \bar\lambda \big\|_2^2   
    &\leq 2\|\bar\lambda - \hat\lambda\|_1 \cdot \big( \|\bar\varepsilon_t\|_\infty +  \|\bar\varepsilon_c\|_{\infty} + \|X_c^\top \varepsilon_w/|\cI_{c,1}|\|_{\infty}\big) + 2\delta_n \cdot (\|\bar\lambda\|_1-\|\hat\lambda\|_1) ,
    \@
    where we define 
    \$
    \varepsilon_{t,j}:= X_{j,\hat{S}} - \mu_t,\quad \mu_t:=\EE[X_{j,\hat{S}}\given T_j=1], \quad 
    \varepsilon_{c,i}:=  {w}^*(X_i) X_{i,\hat{S}} - \mu_c, \quad  \mu_c:= \EE[w^*(X_i)X_{i,\hat{S}}\given T_i=0] 
    \$
    viewing $\hat{S}$ as fixed, and 
    $
    \varepsilon_{w,i} = {w}^*(X_i) - \bar\lambda^\top X_{i,\hat{S}},
    $
    where $\varepsilon_w \in \RR^{|\cI_{c,1}|}$ is the vector whose $i$-th element is $\varepsilon_{w,i}$ for $i\in \cI_{c,1}$.  Accordingly, with probability at least $1-3\varepsilon$, 
    \$
    \|\bar\varepsilon_t\|_{\infty} \leq \sqrt{\frac{2\log(2|\hat{S}|/\varepsilon)}{|\cI_{t,1}|}},\quad \|\bar\varepsilon_c\|_{\infty} \leq M_2 \sqrt{\frac{2\log(2|\hat{S}|/\varepsilon)}{|\cI_{c,1}|}},\quad 
    \|X_c^\top \varepsilon_w/|\cI_{c,1}|\|_{\infty} 
    \leq M_2 \sqrt{\frac{2\log(2|\hat{S}|/\varepsilon)}{|\cI_{c,1}|}}.
    \$
    Combining these arguments, we have 
    \$
    A_n := \|\bar\varepsilon_t\|_\infty +  \|\bar\varepsilon_c\|_{\infty} + \|X_c^\top \varepsilon_w/|\cI_{c,1}|\|_{\infty} 
    \leq \sqrt{2\log(2|\hat{S}|/\varepsilon)} \cdot \bigg(\frac{1}{\sqrt{|\cI_{t,1}|}} + \frac{2M_2}{\sqrt{|\cI_{c,1}|}} \bigg)
    \$
    Thus, continuing with~\eqref{eq:risk_bd_inf} and applying the triangle inequality, we have 
    \@\label{eq:risk_bd3}
    & 
    \frac{1}{|\cI_{c,1}|} \big\|X_c (\hat\lambda -  \bar\lambda ) \big\|_2^2  \leq 2 A_n\cdot  \|\bar\lambda - \hat\lambda\|_1   + 2\delta_n \cdot (\|\bar\lambda\|_1-\|\hat\lambda\|_1) 
    \leq 2(A_n+\delta_n)\cdot \|\bar\lambda - \hat\lambda\|_1 .
    \@
    Applying condition (iii) to $\hat\lambda-\bar\lambda \in \RR^{|\hat{S}|}$ yields 
    \$
    \frac{1}{|\cI_{c,1}|} \big\|X_c (\hat\lambda -  \bar\lambda ) \big\|_2^2 \geq \xi(\hat{S})\cdot \|\bar\lambda - \hat\lambda\|_1 ^2,
    \$
    which further implies $\|\bar\lambda - \hat\lambda\|_1 \leq  2(A_n+\delta_n)/\xi(\hat{S})$, hence  with probability at least $1-3\varepsilon$, we have 
    \@\label{eq:weight_bd_finite}
    \frac{1}{|\cI_{c,1}|} \sum_{i\in \cI_{c,1}} (\hat{w}_i - X_{i,\hat{S}}^\top\bar\lambda)^2 = 
    \frac{1}{|\cI_{c,1}|} \big\|X_c (\hat\lambda -  \bar\lambda ) \big\|_2^2   \leq \frac{4(A_n+\delta_n)^2}{\xi(\hat{S})}.
    \@

    We now proceed to analyze the estimator.  
    Recall the imbalance gap 
    $
    \hat\Delta:= \frac{1}{|\cI_{c,1}|} \sum_{i\in \cI_{c,1}} w_i X_{i,\hat{S}} - \frac{1}{|\cI_{t,1}|} \sum_{j\in \cI_t^{(1)}} X_{j,\hat{S}}  \in \RR^{|\hat{S}|},
    $
    which obeys $\|\hat\Delta\|_\infty\leq \delta_n$ by the balancing constraint. Thus, by the balancing condition (recall  $\beta = \eta(\hat{S};m_0)$), and re-organizing the terms as in the proof of Theorem~\ref{thm:inf_varsel}, 
    \$
    \hat\theta^{(k)}& = \frac{1}{|\cI_{c,1}|}\sum_{i\in \cI_{c,1}} {w}^*(X_i)(Y_i- m_0(X_i)) + \frac{1}{|\cI_{t,1}|}\sum_{j\in \cI_{t,1}}  m_0(X_j) \\ 
    &\quad + \underbrace{
        \frac{1}{|\cI_{c,1}|}\sum_{i\in \cI_{c,1}} (\hat{w}_i - w^*(X_i) ) (Y_i - m_0(X_i))
    }_{\textrm{term (a)}}  + \underbrace{
         \frac{1}{|\cI_{c,1}|}\sum_{i\in \cI_{c,1}}  (\hat{w}_i - w^*(X_i))(m_0(X_i) - \beta^\top X_{i,\hat{S}})
    }_{\textrm{term (b)}}\\ 
    &\quad + \underbrace{
        \frac{1}{|\cI_{c,1}|}\sum_{i\in \cI_{c,1}}   w^*(X_i) \cdot (m_0(X_i) - \beta^\top X_{i,\hat{S}})
        - \frac{1}{|\cI_{t,1}|}\sum_{j\in \cI_{t,1}}   (m_0(X_j) - \beta^\top X_{j,\hat{S}})
    }_{\textrm{term (c)}} - \beta^\top \hat\Delta.
    \$
    Conditions (i)-(ii) imply 
    \$|\beta^\top \Delta|\leq \|\beta\|_1 \|\Delta\|_\infty \leq \delta_n \cdot M_3.
    \$ 
    
    We now proceed to bound terms (a)-(c). 
    Let $\cD_{X}$ be the union of the full data in $\cI_{c,2}\cup\cI_{t,2}$ and the covariates in $\cI_{c,1}\cup\cI_{t,1}$. 
    Since $\hat{w}_i$ only depends on the covariates in $\cI_{c,1}\cup \cI_{t,1}$, conditional on $\cD_X$, each term in the summation of (a) is independent, mean zero, and $(\hat{w}_i-w^*(X_i))^2 M_1^2$-sub-Gaussian. Thus, by the concentration inequality for sub-Gaussian random variables, for any $\varepsilon>0$, it holds that 
    \$
    \PP\Bigg(  |\text{term (a)}|\geq  \sqrt{2M_1^2 \log(2/\varepsilon)\frac{\sum_{i\in \cI_{c,1}}  (\hat{w}_i - w^*(X_i))^2}{|\cI_{c,1}|^2}} \Bigggiven \cD_X\Bigg) \leq\varepsilon. 
    \$ 
    The conditioning on $\cD_X$ can be removed by the tower property.
    The triangle inequality implies 
    \$
\sqrt{\frac{1}{|\cI_{c,1}|^2} \sum_{i\in \cI_{c,1}} (\hat{w}_i - w^*(X_i) )^2 } \leq \sqrt{\frac{1 }{|\cI_{c,1}|^2 } \sum_{i\in \cI_{c,1}} (\hat{w}_i - X_{i,\hat{S}}^\top \bar\lambda )^2} + \sqrt{\frac{1}{|\cI_{c,1}|^2 } \sum_{i\in \cI_{c,1}} (  X_{i,\hat{S}}^\top \bar\lambda - w^*(X_i))^2}
    \$
    Here the second term is an average of i.i.d.~random variable bounded by $M_2^2$ with mean $\epsilon(\hat{S};w^*)^2$. The Hoeffding's inequality implies that with probability at least $1-\varepsilon$, 
    \$
\frac{1 }{|\cI_{c,1}| } \sum_{i\in \cI_{c,1}} (  X_{i,\hat{S}}^\top \bar\lambda - w^*(X_i))^2 \leq   \epsilon(\hat{S};w^*)^2 +  M_2^2 \sqrt{\frac{2\log (2/\varepsilon)}{|\cI_{c,1}|}}.
    \$
    Combining this with~\eqref{eq:weight_bd_finite} and a union bound, we know that with probability at least $1-5\varepsilon$, 
    \@\label{eq:bd_a_fnt}
    |\text{term (a)}| &\leq \sqrt{2M_1^2\log(2/\varepsilon)}\cdot \bigg(\sqrt{ \frac{4(A_n+\delta_n)^2}{|\cI_{c,1}|\cdot \xi(\hat{S})}}+ \sqrt{ \frac{\epsilon(\hat{S};w^*)^2}{|\cI_{c,1}|} + M_2^2 \frac{\sqrt{2\log (2/\varepsilon)}}{|\cI_{c,1}|^{3/2}}} ~\bigg) \notag \\
    &\leq \sqrt{2M_1^2\log(2/\varepsilon)}  \cdot \bigg( \frac{2(A_n+\delta_n) }{\sqrt{|\cI_{c,1}|\cdot \xi(\hat{S})}}+  \frac{\epsilon(\hat{S};w^*) }{\sqrt{|\cI_{c,1}|}} + M_2  \frac{ (2\log (2/\varepsilon))^{1/4}}{|\cI_{c,1}|^{3/4}} ~\bigg)
    \@ 
    Second, by H\"older's inequality, 
    \@\label{eq:bd_b_finite}
    \big|\text{term (b)}\big| \leq \sqrt{\frac{1}{|\cI_{c,1}|}\sum_{i\in \cI_{c,1}}  (\hat{w}_i - w^*(X_i))^2} \cdot \sqrt{\frac{1}{|\cI_{c,1}|}\sum_{i\in \cI_{c,1}}  (m_0(X_i)- \beta^\top X_{i,\hat{S}})^2},
    \@ 
    where using similar ideas as how we bound term (a),  on the same event with probability at least $1-5\varepsilon$, 
    \$
    \sqrt{\frac{1}{|\cI_{c,1}|}\sum_{i\in \cI_{c,1}}  (\hat{w}_i - w^*(X_i))^2} \leq \sqrt{ \frac{4(A_n+\delta_n)^2}{|\cI_{c,1}|\cdot \xi(\hat{S})}}+ \sqrt{ \frac{\epsilon(\hat{S};w^*)^2}{|\cI_{c,1}|} + M_2^2 \frac{\sqrt{2\log (2/\varepsilon)}}{|\cI_{c,1}|^{3/2}}}.
    \$
    On the other hand, by Markov's inequality the second term being taken squared of in~\eqref{eq:bd_b_finite} is the average of i.i.d.~random variables with mean $\epsilon(\hat{S};m_0)^2$ and upper bounded by $M_1^2$. Thus it holds with probability at least $1-\varepsilon$ that 
    \$
    \frac{1}{|\cI_{c,1}|}\sum_{i\in \cI_{c,1}}  (m_0(X_i)- \beta^\top X_{i,\hat{S}})^2 \leq \epsilon(\hat{S};m_0)^2 + M_1^2 \sqrt{\frac{2\log(2/\varepsilon)}{|\cI_{c,1|}}}.
    \$
    With a union bound, it holds with probability at least $1-6\varepsilon$ that~\eqref{eq:bd_a_fnt} holds and 
    \@\label{eq:bd_b_fnt}
    \big|\text{term (b)}\big| & \leq \sqrt{\epsilon(\hat{S};m_0)^2 + M_1^2 \sqrt{\frac{2\log(2/\varepsilon)}{|\cI_{c,1|}}}} \cdot \bigg(\sqrt{ \frac{4(A_n+\delta_n)^2}{|\cI_{c,1}|\cdot \xi(\hat{S})}}+ \sqrt{ \frac{\epsilon(\hat{S};w^*)^2}{|\cI_{c,1}|} + M_2^2 \frac{\sqrt{2\log (2/\varepsilon)}}{|\cI_{c,1}|^{3/2}}}~\bigg) \notag \\
    &\leq \bigg(\epsilon(\hat{S};m_0)  + M_1  \frac{(2\log(2/\varepsilon))^{1/4}}{|\cI_{c,1|^{1/4}}}\bigg) \cdot \bigg( \frac{2(A_n+\delta_n) }{\sqrt{|\cI_{c,1}|\cdot \xi(\hat{S})}}+  \frac{\epsilon(\hat{S};w^*) }{\sqrt{|\cI_{c,1}|}} + M_2  \frac{ (2\log (2/\varepsilon))^{1/4}}{|\cI_{c,1}|^{3/4}} ~\bigg).
    \@

    Let $\mu:= \EE[m_0(X_j) - \beta^\top X_{j,\hat{S}}\given T_j=1] = \EE[w^*(X_i) \cdot (m_0(X_i) - \beta^\top X_{i,\hat{S}})\given T_i=0]$, viewing $\hat{S}$ as fixed. Then
    \$
\textrm{term (c)} = \underbrace{\frac{1}{|\cI_{c,1}|}\sum_{i\in \cI_{c,1}}  \big[ w^*(X_i) \cdot (m_0(X_i) - \beta^\top X_{i,\hat{S}}) - \mu \big]}_{\textrm{term (c,1)}}
        - \underbrace{\frac{1}{|\cI_{t,1}|}\sum_{j\in \cI_{t,1}}   \big[ m_0(X_j) - \beta^\top X_{j,\hat{S}} - \mu\big]}_{\textrm{term (c,2)}}.
    \$
    Each term in the summation in term (c,1) and term (c,2) is i.i.d.~and mean zero. By Bernstein's inequality, since $| w^*(X_i) \cdot (m_0(X_i) - \beta^\top X_{i,\hat{S}}) - \mu |\leq M_1M_2$, for any $t>0$ it holds that 
    \$
    \PP\big( |\text{term (c,1)}| \geq t  \big) \leq 2\exp\bigg( - \frac{|\cI_{c,1}|\cdot t^2}{2\Var(w^*(X)(m_0(X)-\beta^\top X_{\hat{S}})\given T=0)+2M_1M_2t/3}\bigg).
    \$
    A proper choice of $t>0$ yields that with probability at least $1-\varepsilon$, 
    \$
    |\text{term (c,1)}| &\leq \sqrt{\frac{2\Var(w^*(X)(m_0(X)-\beta^\top X_{\hat{S}})\given T=0)\log(2/\varepsilon)}{|\cI_{c,1}|}} + \frac{M_1M_2\log(2/\varepsilon)}{|\cI_{c,1}|} \\ 
    &\leq \sqrt{\frac{2M_2^2 \epsilon(\hat{S};m_0)^2 \log(2/\varepsilon)}{|\cI_{c,1}|}} + \frac{M_1M_2\log(2/\varepsilon)}{|\cI_{c,1}|} .
    \$
    Similarly, with probability at least $1-\varepsilon$, 
    \$
    |\text{term (c,2)}| &\leq \sqrt{\frac{2\Var( m_0(X)-\beta^\top X_{\hat{S}} \given T=1)\log(2/\varepsilon)}{|\cI_{t,1}|}} + \frac{ M_2\log(2/\varepsilon)}{|\cI_{t,1}|} \\ 
    &\leq \sqrt{\frac{2M_2  \epsilon(\hat{S};m_0)^2 \log(2/\varepsilon)}{|\cI_{t,1}|}} + \frac{M_1M_2\log(2/\varepsilon)}{|\cI_{t,1}|} .
    \$
    since we have $\Var( m_0(X)-\beta^\top X_{\hat{S}} \given T=1) \leq \EE[(m_0(X)-\beta^\top X_{\hat{S}})^2\given T=1] =\EE[w^*(X)(m_0(X)-\beta^\top X_{\hat{S}})^2\given T=0]\leq M_2 \epsilon(\hat{S};m_0)^2$. 
    Putting~\eqref{eq:bd_a_fnt},~\eqref{eq:bd_b_fnt} and the above two bounds together, with probability at least $1-8\varepsilon$,
    \$
   & \bigg| \hat\theta^{(k)} -  \frac{1}{|\cI_{c,1}|}\sum_{i\in \cI_{c,1}} {w}^*(X_i)(Y_i- m_0(X_i)) + \frac{1}{|\cI_{t,1}|}\sum_{j\in \cI_{t,1}}  m_0(X_j)  \bigg| \\ 
    &\leq  \bigg(\sqrt{2M_1^2\log(2/\varepsilon)} +\epsilon(\hat{S};m_0)  + M_1  \frac{(2\log(2/\varepsilon))^{1/4}}{|\cI_{c,1|^{1/4}}}\bigg) \cdot \bigg( \frac{2(A_n+\delta_n) }{\sqrt{|\cI_{c,1}|\cdot \xi(\hat{S})}}+  \frac{\epsilon(\hat{S};w^*) }{\sqrt{|\cI_{c,1}|}} + M_2  \frac{ (2\log (2/\varepsilon))^{1/4}}{|\cI_{c,1}|^{3/4}} ~\bigg) \\ 
    &\quad + \frac{\sqrt{2\log(2/\varepsilon)} \cdot M_2 \cdot   \epsilon(\hat{S};m_0) }{|\cI_{c,1}|^{1/2}} + \frac{M_1M_2\log(2/\varepsilon)}{|\cI_{c,1}|}  +  \frac{\sqrt{2\log(2/\varepsilon) } \cdot M_2 \cdot   \epsilon(\hat{S};m_0)   }{|\cI_{t,1}|^{1/2}} + \frac{M_2\log(2/\varepsilon)}{|\cI_{t,1}|}.
    \$
    Similar results apply to $k=2$ in the same way, hence with probability at least $1-32\varepsilon$, 
    \$
   & \bigg| \hat\theta_0  -  \frac{1}{n_c}\sum_{i\in \cI_{c}} {w}^*(X_i)(Y_i- m_0(X_i)) + \frac{1}{n_t}\sum_{j\in \cI_{t}}  m_0(X_j)  \bigg| \\ 
   & \leq \bigg(\sqrt{2M_1^2\log(2/\varepsilon)} + \epsilon_{m}  + M_1  \frac{(4\log(2/\varepsilon))^{1/4}}{|n_c|^{1/4}}\bigg) \cdot \bigg( \frac{\epsilon_{w^*} \sqrt{2} }{\sqrt{n_c}} + M_2  \frac{ (16\log (2/\varepsilon))^{1/4}}{n_c^{3/4}} ~\bigg) \\ 
    &\quad + \frac{\sqrt{4\log(2/\varepsilon)} \cdot M_2 \cdot   \epsilon_{m} }{n_c^{1/2}} + \frac{2M_1M_2\log(2/\varepsilon)}{n_c}  +  \frac{\sqrt{4\log(2/\varepsilon) } \cdot M_2 \cdot   \epsilon_{m}   }{ n_t ^{1/2}} + \frac{2M_2\log(2/\varepsilon)}{n_t} \\ 
    &\quad + \bigg(\sqrt{2M_1^2\log(2/\varepsilon)} + \epsilon_{m}  + M_1  \frac{(4\log(2/\varepsilon))^{1/4}}{|n_c|^{1/4}}\bigg) \cdot\bigg(   \frac{ 2\sqrt{2} \delta_n  }{\sqrt{n_c\cdot \bar\xi_S}} + 4\sqrt{\frac{2\log (2\bar{S}/\varepsilon)}{n_c\cdot \bar\xi_S}} \cdot \Big( \frac{1}{\sqrt{n_t}}+\frac{2M_2}{\sqrt{n_c}} \Big) \bigg) 
   \$
   where we define the maximum selection size $\bar{S} = \max\{|\hat{S}^{(1)}|,|\hat{S}^{(2)}|\}$, and $\bar\xi_S = \min\{\xi(\hat{S}^{(1)}),\xi(\hat{S}^{(2)})\}$, as well as $\epsilon_m = \max_{k=1,2}\{\epsilon(\hat{S}^{(k)};m_0)\}$ and $\epsilon_w = \max_{k=1,2}\{\epsilon(\hat{S}^{(k)};w^*)\}$.  Re-organizing terms, we obtain 
   \$
   & \bigg| \hat\theta_0  -  \frac{1}{n_c}\sum_{i\in \cI_{c}} {w}^*(X_i)(Y_i- m_0(X_i)) + \frac{1}{n_t}\sum_{j\in \cI_{t}}  m_0(X_j)  \bigg| \\ 
   & \leq \tilde\cO \Bigg(  \frac{\sqrt{2} \cdot \epsilon_m\cdot \epsilon_w}{\sqrt{n_c}} + \epsilon_m \cdot 6M_2\sqrt{\log(2/\varepsilon)}\cdot \bigg(\frac{1}{\sqrt{n_c}}+\frac{1}{\sqrt{n_t}}\bigg) + \frac{2\sqrt{2} \cdot \delta_n \cdot (\epsilon_m + M_1  \sqrt{2\log(2/\varepsilon)})}{\sqrt{n_c\cdot \bar{\xi}_S}}      \Bigg),
   \$
   where $\tilde\cO(\cdot)$ hides terms that are of a higher order of $(n_c,n_t)$ than those inside. 
 This concludes the proof of Theorem~\ref{thm:inf_varsel_finite}.
\end{proof}

\subsection{Proof of Proposition~\ref{prop:sel_lr}}
\label{app:proof_sel_lr}

\begin{proof}[Proof of Proposition~\ref{prop:sel_lr}]
    Define the population limits of $\hat\beta$ and $\hat\lambda$ by 
    \$
    \bar\beta := (\EE_c[XX^\top])^{-1} \EE_c[XY],\quad  \bar\lambda := (\EE_c[XX^\top])^{-1} \EE_t[X].
    \$
    Then, by definition we have $Y=X^\top\bar\beta + \epsilon_0$ where $\EE[X\epsilon_0]=0$. Thus,
    \$
    \hat\beta = \hat\Sigma_c^{-1} \hat\EE_c[XX^\top\bar\beta + X\epsilon_0] = \bar\beta + \hat\Sigma_c^{-1} \hat\EE_c[X\epsilon_0] = \bar\beta + O_P(1/\sqrt{n_c}).
    \$
    Re-organizing the terms, we obtain the asymptotic linear expansion 
    \$
    \sqrt{n_c}(\hat\beta - \bar\beta) = \frac{1}{\sqrt{n_c}}\sum_{i\in \cI_c} \Sigma_c^{-1} X_i   (Y_i-X_i^\top \bar\beta) \rightsquigarrow \cN(0, \Sigma_\beta),
    \$
    where the asymptotic covariance matrix $\Sigma_\beta = \Cov(\Sigma_c^{-1} X_i   (Y_i-X_i^\top \bar\beta)\given T_i=0)$ can be consistently estimated by 
    \$
    \hat\Sigma_{\beta} = \hat\Sigma_c^{-1} \frac{1}{n_c}\sum_{i\in \cI_c} [X_iX_i^\top (Y_i-X_i^\top \hat\beta)^2] \hat\Sigma_c^{-1}.
    \$
    That is, we have $\sqrt{n_c}(\hat\beta_j - \bar\beta_j)/\hat\sigma_j \rightsquigarrow \cN(0,1)$ for all $j\in [p]$. On the other hand,
    \$
    \hat\lambda - \bar\lambda &= \hat\Sigma_c^{-1}\hat\EE_t[X] - \Sigma_c^{-1} \EE_t[X] \\ 
    &= \hat\Sigma_c^{-1}(\hat\EE_t[X] - \EE_t[X]) + (\hat\Sigma_c^{-1} - \Sigma_c^{-1}) \EE_t[X] \\ 
    &= \Sigma_c^{-1} (\hat\EE_t[X] - \EE_t[X]) - \Sigma_c^{-1}(\hat\Sigma_c - \Sigma_c) \Sigma_c^{-1} \EE_t[X] + O_P(1/n_c).
    \$
    Re-organizing terms, we have 
    \$
    \sqrt{n_c}(\hat\lambda-\bar\lambda) = \frac{\sqrt{n_c}}{n_t} \sum_{i\in \cI_t} \Sigma_c^{-1}( X_i - \EE_t[X]) - \frac{1}{\sqrt{n_c}} \sum_{i\in \cI_c} \Sigma_c^{-1} (X_iX_i^\top - \Sigma_c) \Sigma_c^{-1} \EE_t[X] + O_P(1/\sqrt{n_c}) \rightsquigarrow \cN(0, \Sigma_\lambda),
    \$  
    where the asymptotic covariance matrix $\Sigma_\lambda = n_c/n_t \Cov(\Sigma_c^{-1}X_i\given T_i=1) + \Cov(\Sigma_c^{-1}X_iX_i^\top \Sigma_c^{-1}\EE_t[X]\given T_i=0)$ can be consistently estimated by 
    \$
    \hat\Sigma_\lambda &=   \hat\Sigma_c^{-1} \frac{1}{n_c}\sum_{i\in\cI_c} \big\{ X_iX_i^\top \hat\Sigma_c^{-1} \hat\EE_t[X]\hat\EE_t[X]^\top \hat\Sigma_c^{-1} X_iX_i^\top \big\}\hat\Sigma_c^{-1} \\ 
    & \quad - \hat\Sigma_c^{-1} (\hat\EE_t[X])(\hat\EE_t[X])^\top \hat\Sigma_c^{-1}  + \frac{n_c}{n_t} \hat\Sigma_c^{-1} \big\{\hat\Sigma_t - (\hat\EE_t[X])(\hat\EE_t[X])^\top \big\}\hat\Sigma_c^{-1} .
    \$
    That is, we have $\sqrt{n_c}(\hat\lambda_j - \bar\lambda_j) /\hat\tau_j \rightsquigarrow \cN(0,1)$ for all $j\in [p]$. Put differently, we have 
    \$
    \hat\beta_j / \hat\sigma_j = \bar\beta_j / \hat\sigma_j + Z_{j1}/\sqrt{n_c} + o_P(1/\sqrt{n_c}),\quad 
    \hat\lambda_j / \hat\tau_j = \bar\lambda_j / \hat\tau_j + Z_{j2}/\sqrt{n_c} + o_P(1/\sqrt{n_c}),
    \$
    where $Z_{j1}\sim \cN(0,1)$ and $Z_{j2}\sim \cN(0,1)$ are standard Gaussian random variables. 

    Due to the fact that $\|\epsilon_m\|_{L_2}=o_P(n^{-1/4})$ and $\|\epsilon_w\|_{L_2} = o_P(n^{-1/4})$, we see that  
    $
    \bar\beta_j = \beta^*_j + o_P(n^{-1/4}) 
    $ 
    for any $j\in S_m$ and $\bar\lambda_j = \lambda_j^* + o_P(n^{-1/4})$ for any $j \in S_w$. This implies $\hat\beta_j / \hat\sigma_j =\beta_j^* /\hat\sigma_j + Z_{j1}/\sqrt{n_c} + o_P(n^{-1/4})$ for any $j\in S_m$ and $\hat\lambda_j / \hat\tau_j =\lambda_j^* /\hat\lambda_j + Z_{j2}/\sqrt{n_c} + o_P(n^{-1/4})$ for any $j\in S_w$. Thus, as long as $\underline{\sigma} \leq c_0 / \sigma_j + o_P(n^{-1/4})$ and $\underline{\tau} \leq c_0/\tau_j + o_P(n^{-1/4})$, we have 
    \$
    \PP\big(  S_m\subseteq \hat{S}_m \text{ and } S_w \subseteq \hat{S}_w \big) \to 1.
    \$
    Noting that $\epsilon(S_m;m_0) \leq \|\epsilon_m\|_{L_2} =o_P(n^{-1/4})$ and $\epsilon(S_w;w^*) \leq \|\epsilon_w\|_{L_2} = o_P(n^{-1/4})$, this further implies 
    \$
    &\PP\big( n^{-1/4} \epsilon(\hat{S}_m \cup \hat{S}_w ;m_0) \leq \delta     \big) \\
    &\geq \PP\big( n^{-1/4} \epsilon(\hat{S}_m  ;m_0) \leq \delta     \big) \\ 
    &\geq \PP\big( n^{-1/4} \epsilon( {S}_m  ;m_0) \leq \delta  ,~ S_m\subseteq \hat{S}_m   \big) \to 1
    \$
    for any fixed $\delta>0$, hence $\epsilon(\hat{S};m_0)=o_P(n^{-1/4})$. With similar arguments, we also have $\epsilon(\hat{S};w^*) = o_P(n^{-1/4})$, thereby completing the proof. 
\end{proof}

\subsection{Proof of Proposition~\ref{prop:sel_lasso}}
\label{app:proof_sel_lasso}

\begin{proof}[Proof of Proposition~\ref{prop:sel_lasso}]
    Our analysis largely follows the seminal results on the prediction risk of the Lasso under the compatibility condition (e.g.,~\cite{tibshirani2011regression,van2000asymptotic,buhlmann2011statistics}), with minor modifications to handle the approximately linear models and implicit outcomes in the weight estimation. 

    We first consider $\hat{S}_m$ for the outcome model. Recall that $\hat\beta$ is the minimizer of the convex objective
    \$
    \hat{L}(\beta) = \beta^\top \hat\Sigma_c \beta - 2 \beta^\top \hat\EE_c[XY] + \nu_1\|\beta\|_1.
    \$
    This implies $\hat{L}(\hat\beta)\leq \hat{L}(\beta^*)$, which, after reorganizing terms, means 
    \$
    \frac{1}{n_c} \|X_c^\top \hat\beta - X_c^\top \beta^*\|_2^2 
    & \leq 2(\beta^*-\hat\beta) ( \hat\Sigma_c\beta^* - \hat\EE_c[XY]) + \nu_1 \big( \|\beta^*\|_1 - \|\hat\beta\|_1 \big) \\ 
    &\leq 2\|\beta^*-\hat\beta\|_1 \cdot \| \hat\Sigma_c\beta^* - \hat\EE_c[XY] \|_\infty + \nu_1 \big( \|\beta^*\|_1 - \|\hat\beta\|_1 \big),
    \$
    where $X_c\in \RR^{|\cI_c|\times p}$ is the data matrix of $\cI_c$. Noting that 
    \$
\hat\EE_c[XY] = \hat\EE_c[XX^\top \beta^*] + \hat\EE_c[X\epsilon_m] + \hat\EE_c[X(Y-m_0(X))],
    \$
    we have 
    \$
\| \hat\Sigma_c\beta^* - \hat\EE_c[XY] \|_\infty = \big\|\hat\EE_c[X\epsilon_m] + \hat\EE_c[X(Y-m_0(X))]\big\|_\infty \leq \big\|\hat\EE_c[X\epsilon_m] \big\|_\infty + \big\| \hat\EE_c[X(Y-m_0(X))]\big\|_\infty.
    \$
    Since each entry of $X(Y-m_0(X))$ is the average of $n_c=|\cI_c|$ many i.i.d.~mean zero random variables, each $X_j$ is sub-Gaussian, and $|Y-m_0(X)|\leq M_2$, by Chernoff's inequality and a union bound, we know that with probability at least $1-\varepsilon$, it holds that 
    \$
\big\| \hat\EE_c[X(Y-m_0(X))]\big\|_\infty 
\leq M_2 \sqrt{\frac{2\log(2p/\varepsilon)}{n_c}}.
    \$
    In addition, since $|\epsilon_m(X)|\leq M_2$ and each entry $X_j$ is sub-Gaussian, by Chernoff's inequality and a union bound, with probability at least $1-\varepsilon$, it holds that 
    \$
    \big\|\hat\EE_c[X\epsilon_m] \big\|_\infty &\leq \big\| \EE_c[X\epsilon_m] \big\|_\infty + \big\| \hat\EE_c[X\epsilon_m] - \EE_c[X\epsilon_m] \big\|_\infty \\
    &\leq M_2 \sqrt{\frac{2\log(2p/\varepsilon)}{n_c}} + \max_j \EE[X_j \epsilon_m(X)] = M_2 \sqrt{\frac{2\log(2p/\varepsilon)}{n_c}} + \|X\epsilon_m\|_\infty.
    \$
    Combining the above two arguments, we know that with probability at least $1-2\varepsilon$, it holds that 
    \@
    \frac{1}{n_c} \|X_c^\top \hat\beta - X_c^\top \beta^*\|_2^2  
    &\leq 2\|\beta^*-\hat\beta\|_1 \cdot \bigg(\|X\epsilon_m\|_\infty+ 2M_2 \sqrt{\frac{2\log(2p/\varepsilon)}{n_c}}  \bigg) + \nu_1 \big( \|\beta^*\|_1 - \|\hat\beta\|_1 \big)\label{eq:L2L1_bd1} \\ 
    &\leq \|\beta^*-\hat\beta\|_1 \cdot \bigg( 2\|X\epsilon_m\|_\infty+ 4M_2 \sqrt{\frac{2\log(2p/\varepsilon)}{n_c}} + \nu_1 \bigg). \label{eq:L2L1_bd2}
    \@   
    Using~\eqref{eq:L2L1_bd1} and the non-negativity of the LHS, and denoting the constant $A_n=\|X\epsilon_m\|_\infty+ 2M_2 \sqrt{\frac{2\log(2p/\varepsilon)}{n_c}}$, we know that 
    \$
    0&\leq 2A_n \|\beta^*_{S_m}-\hat\beta_{S_m}\|_1 + 2A_n \|\beta^*_{S_m^c}-\hat\beta_{S_m^c}\|_1 + \nu_1 (\|\beta_{S_m^c}^*\|_1 - \|\hat\beta_{S_m}\|_1 - \|\hat\beta_{S_m^c}\|_1)\\
    &\leq 2A_n \|\beta^*_{S_m}-\hat\beta_{S_m}\|_1 + 2A_n \|\beta^*_{S_m^c}-\hat\beta_{S_m^c}\|_1 + \nu_1 (\|\beta_{S_m^c}^*\|_1 - \|\hat\beta_{S_m}\|_1 - \|\hat\beta_{S_m^c} - \beta_{S_m^c}^*\|_1)\\
    &\leq 2A_n \|\beta^*_{S_m}-\hat\beta_{S_m}\|_1 + 2A_n \|\beta^*_{S_m^c}-\hat\beta_{S_m^c}\|_1 + \nu_1 (\|\beta_{S_m^c}^*  -  \hat\beta_{S_m}\|_1 - \|\hat\beta_{S_m^c} - \beta_{S_m^c}^*\|_1) \\ 
    &= (\nu_1+2A_n) \|\beta^*_{S_m}-\hat\beta_{S_m}\|_1 - (\nu_1-2A_n) \|\beta^*_{S_m^c}-\hat\beta_{S_m^c}\|_1.
    \$
    Since $\nu_1 \geq 4A_n$, the above inequality leads to 
    \$
    \|(\hat\beta-\beta^*)_{S_m^c}\|_1 \leq 3     \|(\hat\beta-\beta^*)_{S_m}\|_1.
    \$
    Invoking the compatibility condition for $\beta = \hat\beta-\beta^*$, together with~\ref{eq:L2L1_bd2}, we obtain  
    \$
    \|(\hat\beta-\beta^*)_{S_m}\|_1^2  \leq 3/2\cdot \nu_1 \|(\hat\beta-\beta^*)_{S_m}\|_1 \cdot |S_m|/\phi_0^2,
    \$
    which further implies 
    \$
    \|\hat\beta-\beta^*\|_1 \leq 4 \|(\hat\beta-\beta^*)_{S_m}\|_1 \leq \frac{6|S_m| \nu_1}{\phi_0^2}.
    \$
    Thus, the approximation error of $\hat{S}=\hat{S}_m\cup \hat{S}_w$ satisfies 
    \$
    \epsilon(\hat{S}; m_0)^2 \leq \epsilon(\hat{S}_m; m_0)^2 \leq \EE\big[ (\hat\beta^\top X - m_0(X))^2 \big],
    \$
    where  the expectation is over an independent new sample $X$, viewing $\hat\beta$ as given. The triangle inequality further yields 
    \$
    \EE\big[ (\hat\beta^\top X - m_0(X))^2 \big]
   & \leq 2 \EE\big[ ( X^\top \beta^* - m_0(X))^2 \big] 
    + 2 \EE\big[ \{(\hat\beta -\beta^*)^\top X \}^2 \big] \\ 
    & \leq 2\|\epsilon_m\|_{L_2}^2 + 2M_1^2 \|\hat\beta-\beta^*\|_1^2 \leq 2\|\epsilon_m\|_{L_2}^2 + 12M_1^2 |S_m|^2\nu_1^2/\phi_0^4. 
    \$
    We then have 
    \$
\epsilon(\hat{S}; m_0) \leq   \sqrt{2}\|\epsilon_m\|_{L_2} + 2\sqrt{3}M_1|S_m|\nu_1/\phi_0^2 = \Omega\Big(\max\big\{ \|\epsilon_m\|_{L_2}, ~|S_m|\nu_1\big\}\Big).
    \$

    We then proceed to analyze the estimator $\hat\lambda$. With similar techniques as before, we can show that 
    \$
    \frac{1}{n_c} \|X_c^\top \hat\lambda - X_c^\top \lambda^*\|_2^2 
    & \leq 2(\lambda^*-\hat\lambda) ( \hat\Sigma_c\lambda^* - \hat\EE_t[X]) + \nu_2 \big( \|\lambda^*\|_1 - \|\hat\lambda\|_1 \big) \\ 
    &\leq 2\|\lambda^*-\hat\lambda\|_1 \cdot \| \hat\Sigma_c\lambda^* - \hat\EE_t[X] \|_\infty + \nu_2 \big( \|\lambda^*\|_1 - \|\hat\lambda\|_1 \big).
    \$
    Here we note that, by the fact that $\EE_c[Xw^*(X)]=\EE_t[X]$ since $w^*(x)$ is the true density ratio, as well as the triangle inequality, 
    \$
    & \| \hat\Sigma_c\lambda^* - \hat\EE_t[X] \|_\infty  
    = \big\| - \hat\EE_c[X \epsilon_w] + \hat\EE_c[ Xw^*(X)] - \hat\EE_t[X] \big\|_\infty \\
    &\leq \big\| \hat\EE_c[X \epsilon_w]  - \EE_c[X\epsilon_w] \big\|_\infty + \big\|\EE_c[X\epsilon_w]\big\|_\infty 
    + \big\| \hat\EE_c[Xw^*(X)] - \EE_c[Xw^*(X)]\big\|_1 + \big\| \hat\EE_t[X] - \EE_t[X]\big\|_1.
    \$
    A union bound together with repeated application of Chernoff's inequality then yields, with probability at least $1-3\varepsilon$, that 
    \$
    \| \hat\Sigma_c\lambda^* - \hat\EE_t[X] \|_\infty
    \leq \big\|\EE_c[X\epsilon_w]\big\|_\infty  + 3M_2 \sqrt{\frac{2\log(2p/\varepsilon)}{n_c}}.
    \$
    Now denote $B_n = \big\|\EE_c[X\epsilon_w]\big\|_\infty  + 3M_2 \sqrt{\frac{2\log(2p/\varepsilon)}{n_c}}$. Again with similar ideas as our analysis of the outcome regression model, we can show that 
    \$
    \|\hat\lambda - \lambda^*\|_1 \leq \frac{6|S_w|\nu_2}{\phi_0^2}.
    \$
    The approximation error of $\hat{S}=\hat{S}_m\cup \hat{S}_w\supseteq \{j\colon \hat\lambda_j \neq 0\}$ thus satisfies 
    \$
    \epsilon(\hat{S};w^*)^2 \leq \epsilon(\hat{S}_w; w^*)^2 \leq \EE\big[ (\hat\lambda^\top X - w^*(X))^2\big],
    \$
    where the expectation is over an independent new sample $X$, viewing $\hat\lambda$ as given. This further yields 
    \$
    \epsilon(\hat{S};w^*)^2 \leq 2 \EE\big[ ( X^\top \lambda^* - w^*(X))^2\big] + 2 \EE\big[ \{ X^\top (\hat\lambda-\lambda^*)\}^2\big] = 2 \|\epsilon_w\|_{L_2}^2 + 2M_1^2 \|\hat\lambda-\lambda^*\|_1^2 .
    \$
    We thus have 
    \$
    \epsilon(\hat{S};w^*) \leq \Omega\Big(\max\big\{ \|\epsilon_w\|_{L_2},  |S_w|\cdot \nu_2 \big\}\Big)
    \$
    Therefore, we conclude the proof of Proposition~\ref{prop:sel_lasso}. 
\end{proof}

\section{Auxiliary lemmas}

Lemma~\ref{lem:convex} is a standard result in convex analysis, see, e.g.,~\cite[Lemma I.6]{jin2024tailored}. 

\begin{lemma} \label{lem:convex}
    Suppose $f\colon \Theta\to \RR$ is convex in $\Theta \subset \RR^p$, and there exists some constant $\nu,c>0$ such that $\nabla^2 f(\theta) \succeq \nu\cdot \mathbf{I}_{p\times p}$ for all $\theta\in \Theta$ such that $\|\theta-\theta_0\|\leq c$. Then,  $f(\theta)\geq f(\theta_0) + \nabla f(\theta_0)^\top (\theta-\theta_0) + \nu/2 \cdot \min\{ \|\theta-\theta_0\|^2, c\|\theta-\theta_0\|\}$ holds for any $\theta\in \Theta$.
\end{lemma}

Lemma~\ref{lem:concen_covmat} characterizes the concentration of sample covariance matrices (see, e.g.,~\cite{tropp2015introduction}).
\begin{lemma}
\label{lem:concen_covmat}
    Let $X_i \in \RR^d$ be i.i.d., mean zero, and $\sigma^2$-sub-Gaussian random vectors with population variance matrix $\Sigma$, and let $\hat\Sigma=\frac{1}{n}\sum_{i=1}^n X_iX_i^\top$. Then, there exists a universal constant $C>0$ such that with probability at least $1-\varepsilon$, 
    \$
    \frac{\|\hat\Sigma-\Sigma\|_{\textnormal{op}}}{\sigma^2} \leq C\max \Bigg\{ \sqrt{\frac{d+\log(2/\varepsilon)}{n}},  \frac{d+\log(2/\varepsilon)}{n}   \Bigg\}.
    \$
\end{lemma}

Lemma~\ref{lem:cond_event} formalizes arguments that go from conditional converge-in-probability events to marginal converge-in-probability events. See~\cite[Lamme I.5]{jin2024tailored} for a proof. 

\begin{lemma}\label{lem:cond_event}
Let $\cF_n$ be a sequence of $\sigma$-algebra, and $A_n$ be a sequence of non-negative random variables. If $\EE[A_n\given \cF_n]=o_P(1)$, then $A_n = o_P(1)$. 
\end{lemma}

A direct consequence of Lemma~\ref{lem:cond_event} is the next lemma allowing to turn conditionally $o_P(1)$ quantities to marginally $o_P(1)$ quantities, which will be used from time to time in analyzing cross-balancing estimators. 

\begin{lemma}\label{lem:cond_quant}
    Let $\cF_n$ be a sequence of $\sigma$-algebra, and let $R_n$ be a sequence of random variables such that 
    $
    \PP(R_n >\epsilon \given \cF_n) \to 0
    $
    as $n\to \infty$, then we also have $\PP(R_n>\epsilon)\to 0$ as $n\to \infty$. 
\end{lemma}

The proof of Lemma~\ref{lem:cond_quant} is simply to take $A_n = \ind\{R_n>\epsilon\}$ and invoke Lemma~\ref{lem:cond_event}. As a result, if we take $\cF_n$ to be the information in one fold $\cI_{c,1}\cup\cI_{t,1}$ and show some random variable $R_n$ to be $o_P(r_n)$ for some deterministic sequence $r_n$, where the probability is conditional on $\cI_{c,1}\cup\cI_{t,1}$, then $R_n = o_P(r_n)$ marginally by Lemma~\ref{lem:cond_quant}. 
We will use this idea many times in our proofs, but we do not explicitly mention it all the time. 

